\documentclass[11pt,twoside]{article}

\usepackage[usenames,x11names]{xcolor}

\usepackage[english]{babel}
\usepackage{latexsym,amssymb,amsmath,amsthm}
\usepackage{graphicx}
\usepackage[dvips,pagebackref]{hyperref}
\usepackage{enumitem}
\usepackage[a4paper]{geometry}

\hypersetup{%
	pdftitle={A Simple Sweep Line Algorithm for Counting Triangulations and Pseudo-triangulations},%
	pdfauthor={Victor Alvarez and Karl Bringmann and Saurabh Ray},%
	pdfkeywords={Algorithmic geometry, Algorithms, Geometry, Crossing-free structures, Counting algorithms, Triangulations, Pseudo-triangulations},%
	pdfdisplaydoctitle=true,%
	pdfstartview=Fit,%
	colorlinks=false,%
	linktocpage,
}

\newtheorem{theorem}{Theorem}
\newtheorem{lemma}{Lemma}

\newtheorem{definition}{Definition}

\newcommand{\R}{\mathbb{R}}
\renewcommand{\P}{\mathcal{P}}
\newcommand{\tcomp}{\Pi}
\newcommand{\T}{\mathcal{T}}
\renewcommand{\L}{\mathcal{L}}
\newcommand{\tp}[2]{p(#1, #2)}
\newcommand{\pt}[2]{pt(#1, #2)}
\newcommand{\setp}{P}
\newcommand{\W}{\mathcal{W}}
\newcommand{\F}{\mathcal{F}}
\newcommand{\Conv}{\textup{CH}}

\title{A Simple Sweep Line Algorithm for Counting Triangulations and Pseudo-triangulations}
\author{Victor Alvarez\thanks{Fachrichtung Informatik, Universit\"{a}t des Saarlandes, {\tt{alvarez@cs.uni-saarland.de}}. Partially Supported by CONACYT-DAAD of M\'{e}xico.} \and Karl Bringmann\thanks{Max-Planck-Institut f\"ur Informatik, {\tt{kbringma@mpi-inf.mpg.de}}.} \and Saurabh Ray\thanks{Max-Planck-Institut f\"ur Informatik. {\tt{saurabh@mpi-inf.mpg.de}}.}}
\date{\today}

\begin{document}
\maketitle

\begin{abstract}
Let $\setp\subset\R^{2}$ be a set of $n$ points. In~\cite{DBLP:conf/compgeom/Aichholzer99} and~\cite{DBLP:conf/wads/AichholzerRSS03} an algorithm for counting triangulations and pseudo-tri\-an\-gu\-la\-tions of $\setp$, respectively, is shown. Both algorithms are based on the divide-and-conquer paradigm, and both work by finding sub-structures on triangulations and pseudo-triangulations that allow the problems to be split. These sub-structures are called \emph{triangulation paths} for triangulations, or T-paths for short, and \emph{zig-zag paths} for pseudo-triangulations, or PT-paths for short. Those two algorithms have turned out to be very difficult to analyze, to the point that no good analysis of their running time has been presented so far. The interesting thing about those algorithms, besides their simplicity, is that they experimentally indicate that counting can be done significantly faster than enumeration.

In this paper we show two new algorithms, one to compute the number of triangulations of $\setp$, and one to compute the number of pseudo-triangulations of $\setp$. They are also based on T-paths and PT-paths respectively, but use the sweep line paradigm and not divide-and-conquer. The important thing about our algorithms is that they admit a good analysis of their running times. We will show that our algorithms run in time $O^{*}(t(\setp))$ and $O^{*}(pt(\setp))$ respectively, where $t(\setp)$ and $pt(\setp)$ is the largest number of T-paths and PT-paths, respectively, that the algorithms encounter during their execution. Moreover, we show that $t(\setp) = O^{*}(9^{n})$, which is the first non-trivial bound on $t(\setp)$ to be known. 

While the algorithm for counting triangulations of~\cite{DBLP:conf/compgeom/AlvarezBCR12} is faster in the worst case, $O^{*}\left(3.1414^{n}\right)$, than our algorithm, $O^{*}\left(9^{n}\right)$, there are sets of points where the number of T-paths is $O(2^{n})$. In such cases our algorithm may be faster. Furthermore, it is not clear whether the algorithm presented in \cite{DBLP:conf/compgeom/AlvarezBCR12} can be modified to count pseudo-triangulations so that its running time remains $O^{*}(c^n)$ for some small constant $c\in\R$. Therefore, for counting pseudo-triangulations (and possibly other similar structures) our approach seems better.
\end{abstract}

\section{Introduction}

Let $\setp\subset\R^{2}$ be a set of $n$ points. A triangulation of $\setp$ is a crossing-free structure (straight-edge plane graph) on $\setp$ such that the boundary of its outer face coincides with the convex hull,  $\Conv(\setp)$, of $\setp$, and where \emph{all} bounded faces are empty triangles. A pseudo-triangle is an empty simple polygon having \emph{exactly} three convex vertices, that is, the internal angle at those vertices is strictly less than $\pi$. An example can be seen to the left in Figure~\ref{intro:figs:3}. A pseudo-triangulation of $\setp$ is a crossing-free structure on $\setp$ such that the boundary of its outer face coincides with $\Conv(\setp)$, and where \emph{all} bounded faces are pseudo-triangles. A pseudo-triangulation can be seen to the right in Figure~\ref{intro:figs:3}.

\begin{figure}[!htb]
	\begin{center}	
		\includegraphics[height=4cm]{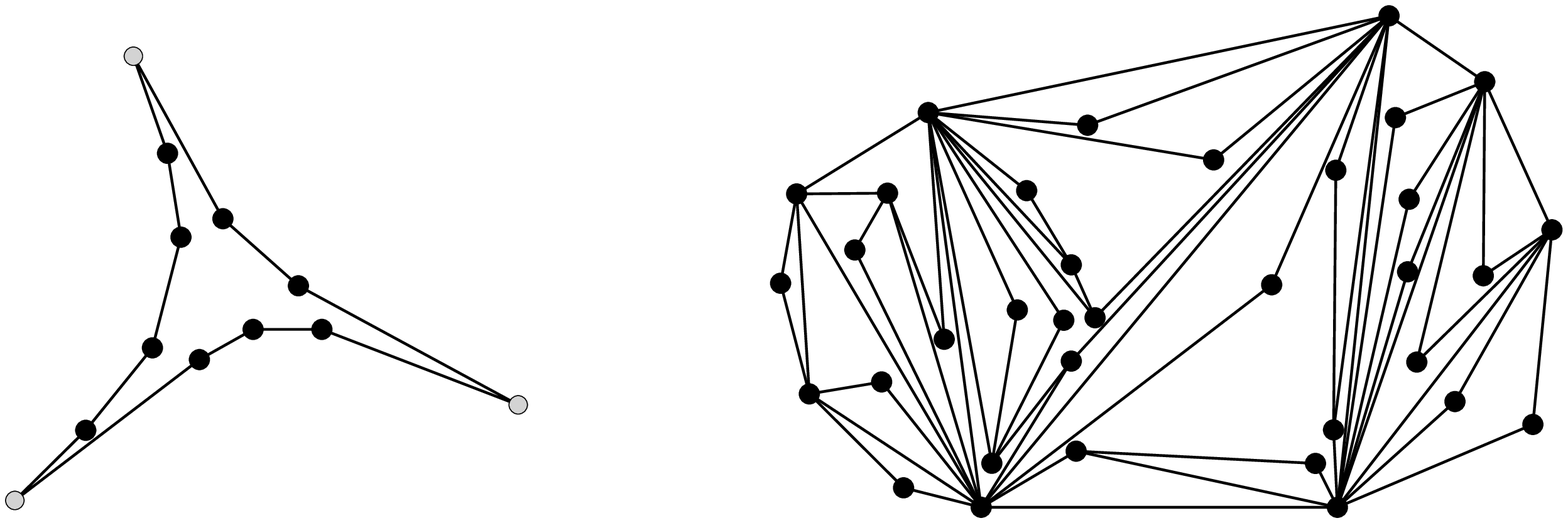}
	\end{center}
	\caption{A pseudo-triangle to the left. The three gray vertices are the three convex vertices. A pseudo-triangulation of $\setp$ can be seen to the right.}
	\label{intro:figs:3}
\end{figure}

While triangulations require essentially no introduction due to their many applications, pseudo-triangulations are way less known. Pseudo-triangulations were originally used in~\cite{DBLP:journals/dcg/PocchiolaV96} for sweeping complexes, and in~\cite{DBLP:journals/algorithmica/ChazelleEGGHSS94,DBLP:journals/jal/GoodrichT97} for ray-shooting. However, it was until a paper of Ileana Streinu appeared, see~\cite{DBLP:conf/focs/Streinu00}, that pseudo-triangulations really took off as a main research topic, due to their structural richness. In the same paper,~\cite{DBLP:conf/focs/Streinu00}, a particular kind of pseudo-triangulations was introduced, the so-called \emph{pointed} pseudo-triangulations. In a pointed pseudo-triangulation \emph{every} vertex is incident to an angle larger than $\pi$, and its characterization is very rich. The following is just a subset of equivalences found in~\cite{DBLP:conf/focs/Streinu00}:

\begin{theorem}[I. Streinu]\label{c-tri:theorems:pointed-pt}
	Let $G$ be a straight-edge plane graph on a set of points $\setp$. The following properties are equivalent:
	\begin{itemize}
		\item $G$ is a pointed pseudo-triangulation.
		\item $G$ is a pseudo-triangulation having the \emph{minimum} number of edges, and thus also the \emph{minimum} number of pseudo-triangles.
		\item The set of edges of $G$ forms a \emph{maximal}, by inclusion, planar and pointed set of edges, \emph{i.e.}, a set of edges whose union is crossing-free, and in which every vertex is incident with an angle larger than $\pi$.
	\end{itemize}
\end{theorem}

Pointed pseudo-triangulations have found interesting applications in robot arm motion planning, see~\cite{DBLP:conf/focs/Streinu00}, and have been the subject of extensive research, see the survey on pseudo-triangulations in~\cite{cg-surveys}, which is an excellent reference for most known results to date on pseudo-triangulations.

In this work we will be concerned only with pointed pseudo-triangulations, so we will drop the ``pointed'' part and we will only call them pseudo-triangulations. So, unless otherwise stated, our pseudo-triangulations are \emph{always} pointed. No confusion shall arise.

Knowing what triangulations and pseudo-triangulations are, we can talk about the classes $\F_{T}(\setp)$ and $\F_{PT}(\setp)$ of \emph{all} triangulations and \emph{all} pseudo-triangulations of a given set of $n$ points $\setp$ respectively, and ask about their sizes, how large are they? We can actually think about two flavors of this question: (\oldstylenums{1}) What is the \emph{largest} or \emph{smallest} they can get over all sets $\setp\subset\R^{2}$ of $n$ points? or (\oldstylenums{2}) Given $\setp$, what is the \emph{exact} size of a desired class? 

The first question mentioned above requires usually heavy mathematical machinery since the number of \emph{combinatorially different} configurations of $n$ points is too large to be explored by computer, see~\cite{DBLP:journals/dcg/GoodmanP86}. Thus, the first question is of rather theoretical flavor and it has actually spawned a large amount of research over almost 30 years, which started with the seminal work of Ajtai, Chv\'{a}tal, Newborn and Szemer\'{e}di, where they showed that the number of \emph{all} crossing-free structures on any set of $n$ points on the plane can be at most $10^{13n}$, see~\cite{Ajtai19829}. This bound implies that the size of \emph{each} class of crossing-free structures on $\setp$ can be upper-bounded by $c^{n}$, with $c\in\R$ depending on the particular class. Since then research has focused on fine-tuning $c$. For example, in the case of triangulations, the most popular in recent years, it is currently known that $2.4\leq c\leq 30$, see~\cite{DBLP:journals/combinatorics/SharirS11} for the upper bound and~\cite{DBLP:journals/jct/SharirSW11} for the lower bound. Thus \emph{every} set $\setp$ of $n$ points on the plane fulfills $|\F_{T}(\setp)| = \Omega(2.4^{n})$ and $|\F_{T}(\setp)| = O(30^{n})$. For the class of pseudo-triangulations not much is known. For example, it is known that $c$ attains its minimum value for sets of points in convex position, \emph{i.e.}, $c\geq 4$, see~\cite{DBLP:journals/comgeo/AichholzerAKS04}. It is also known that $|\F_{PT}(\setp)|\leq 3^{i}|\F_{T}(\setp)|$, where $i$ is the number of interior points of $\setp$, see~\cite{DBLP:conf/cccg/RandallRSS01}.

As for the second question mentioned before, we always assume that we are given a set $\setp$ of $n$ points on the plane and we are interested in computing the exact values of $|\F_{T}(\setp)|, |\F_{PT}(\setp)|$, for example, the set of 32 red points presented in Figure~\ref{c-tri:figs:0}, representing the State Capitals of Mexico, spans \emph{exactly} $6\ 887\ 011\ 250\ 368\ 237\ 767\approx 3.8787^{32}$ triangulations.

\begin{figure}[!htb]
	\begin{center}
		\includegraphics[height=6cm]{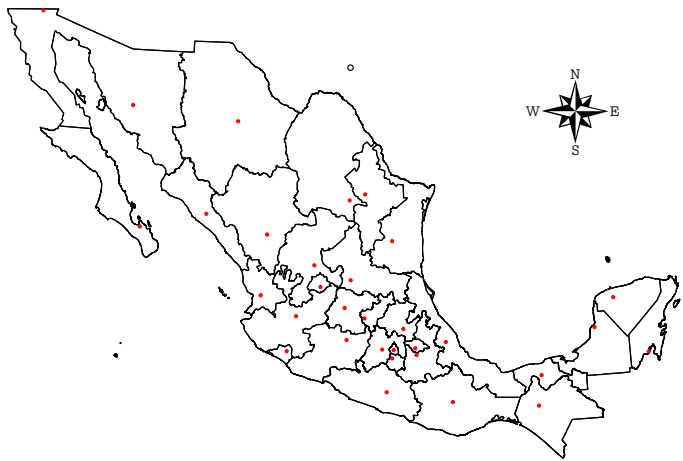}
	\end{center}
	\caption{A set of 32 points representing the State Capitals of Mexico.}
	\label{c-tri:figs:0}
\end{figure}

The second question is thus of empirical flavor, and therefore algorithmic, since no closed-form formula is known, in general, for $|\F_{T}(\setp)|, |\F_{PT}(\setp)|$. It is then important to come up with methods (algorithms) that can compute their sizes efficiently. A first approach would be to produce \emph{all} elements of the desired class, using methods for enumeration, see for example~\cite{DBLP:journals/dam/AvisF96,DBLP:conf/cccg/Bespamyatnikh02a,DBLP:journals/comgeo/Bereg05,DBLP:journals/dcg/KatohT09}, and then simply count the number of elements. This has the obvious disadvantage that the total time spent will be, at best, linear in the number of elements counted, which, by the first part, is always exponential in the size of the input. Thus, the crucial question is whether $|\F_{T}(\setp)|, |\F_{PT}(\setp)|$ can be computed faster, say, for starters, in time \emph{sub-linear} in the number of elements counted. Currently this is only known for the super class of \emph{all} crossing-free structures on the given set $\setp$ of $n$ points, see~\cite{DBLP:conf/birthday/RazenW11}. For the particular class $\F_{T}(\setp)$ a new algorithm that counts the triangulations of $\setp$ in time $n^{O(k)}$ was recently shown in~\cite{DBLP:conf/compgeom/AlvarezBCR12}, where $k$ is the number of onion layers of $\setp$. This algorithm runs in polynomial time whenever $k$ is constant, and thus it is faster than enumeration by an exponential speed-up. The authors of~\cite{DBLP:conf/compgeom/AlvarezBCR12} showed that even when $k = \Theta(n)$, their algorithm has worst-case running time of $O^{*}\left(3.1414^{n}\right)$. While that algorithm is faster in the worst case than the algorithm presented in this paper, see Theorem~\ref{c-tri:theorems:our-t-paths} on page~\pageref{c-tri:theorems:our-t-paths}, there are sets of points where the number of T-paths is $O(2^{n})$. In such cases our algorithm may be faster. Furthermore, our algorithm can easily be modified to count pseudo-triangulations - and the running time remains $O(\text{poly}(n)\cdot pt(\setp))$, where $pt(\setp)$ denotes the largest number of PT-paths with respect to a segment, see Theorem~\ref{c-tri:theorems:pt-paths} on page~\pageref{c-tri:theorems:pt-paths}. It is not clear whether the algorithm presented in \cite{DBLP:conf/compgeom/AlvarezBCR12} can be modified to count pseudo-triangulations so that its running time remains $O(c^n)$ for some small constant $c$. Therefore, for counting pseudo-triangulations (and possibly other similar structures) our approach seems better. There are also other algorithms that seem to count faster than enumeration, for $\F_{T}(\setp)$ and $\F_{PT}(\setp)$, see~\cite{DBLP:conf/compgeom/Aichholzer99,ray-seidel,DBLP:conf/wads/AichholzerRSS03}, but where no theoretical runtime guarantees are known.

\section{Our contribution}\label{c-tri:t-paths:sections:contribution}

In this paper we are fully devoted to the second question, namely, the algorithmic version of the problem of counting triangulations and pseudo-triangulations. We will only be concerned about algorithms with provable running times. 

\subsection{The result on counting triangulations}

In order to state our results we will require some definitions, which for clarity we state first:

\begin{definition}[Separating line]\label{c-tri:def:sep-line}
	Let $\setp$ be a non-empty set of points on the plane, and let $l$ be a straight line such that $l\cap\setp = \emptyset$ but $l\cap \Conv(\setp)\neq\emptyset$, then $l$ will be called a \emph{separating line} w.r.t.~$\setp$.
\end{definition}

\begin{definition}[T-path]\label{c-tri:def:t-paths}
	Given a non-empty set of points $\setp$ on the plane, a triangulation $T$ of $\setp$, and a separating line $l$ w.r.t.~$\setp$, a T-path of $T$ w.r.t.~$l$, denoted by $\tp{l}{T}$, is defined as follows: (\oldstylenums{1}) $\tp{l}{T}$ is a chain of edges of $T$ where every edge of $\tp{l}{T}$ intersects $l$. (\oldstylenums{2}) Starting and ending edges of $\tp{l}{T}$ are two edges of $\Conv(\setp)$ intersected by $l$. (\oldstylenums{3}) The area bounded by two consecutive edges of $\tp{l}{T}$ and $l$ must be empty of points of $\setp$. See to the left in Figure~\ref{c-tri:figs:1} for an example of a T-path $\tp{l}{T}$.
\end{definition}

\begin{figure}[!htb]
	\begin{center}	
		\includegraphics[height=4cm]{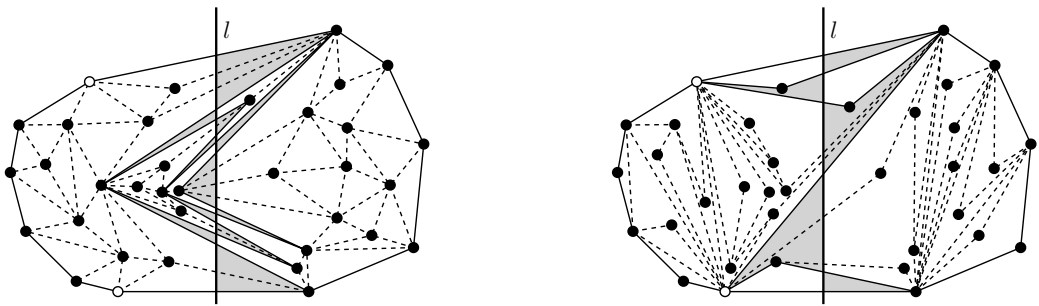}
	\end{center}
	\caption{To the left a T-path $\tp{l}{T}$, shown in solid lines, of a triangulation $T$ with vertex set $\setp$. To the right a PT-path $\pt{l}{S}$, shown also in solid lines, of a pseudo-triangulation $S$ with vertex set $\setp$. The gray areas are the areas bounded by two consecutive edges of the paths and line $l$, which are empty of points of $\setp$.}
	\label{c-tri:figs:1}
\end{figure}

T-paths were originally introduced by Oswin Aichholzer in 1999 in~\cite{DBLP:conf/compgeom/Aichholzer99}. What makes them relevant is the following theorem, also presented in~\cite{DBLP:conf/compgeom/Aichholzer99}:

\begin{theorem}[O. Aichholzer]\label{c-tri:theorems:t-pathsOswin}
	Let $\setp$ be a set of points and $l$ a separating line w.r.t.~$\setp$. Then the following holds: (\oldstylenums{1}) For every triangulation $T$ of $\setp$ there always exists a T-path $\tp{l}{T}$. (\oldstylenums{2}) $\tp{l}{T}$ is unique for $T$. (\oldstylenums{3}) If $T$ and $T^{\prime}$ are two triangulations of $\setp$, then $\tp{l}{T}$ and $\tp{l}{T^{\prime}}$ are either equal, or properly intersect each other, \emph{i.e.}, there are intersection points lying in the strict interior of their edges.
\end{theorem}

Moreover, in the same paper, Aichholzer designed an algorithm to compute $|\F_{T}(\setp)|$ based on T-paths and the divide-and-conquer paradigm. His algorithm experimentally exhibited a running time sub-linear in the number of triangulations counted, that is, that algorithm was apparently faster than enumeration. A formal proof of this fact is, however, hard to obtain since it is not clear how to show that a single T-path appears in many triangulations, even on average. Nonetheless, the running time of Aichholzer's algorithm can be bounded by the number of sub-problems that it generates. Since the algorithm is based on the divide-and-conquer paradigm, we can describe its running-time recurrence by $R(n) = 2t(n)\cdot R(n/2)$, where $t(i)$ denotes the number of T-paths encountered by the algorithm when $i$ points are considered. If we can show that $t(i)\leq a^{i}$, for some positive constant $a$, we have that $R(n)\leq 2\cdot a^{n}\cdot R(n/2)$, which gets solved to $O\left(a^{2n}\right)$. It is important to note here that $t = t(n)$ can become exponentially large, for example, Aichholzer showed that the convex polygon on $n$ vertices has roughly $O\left(2^{n}\right)$ T-paths, and in~\cite{DBLP:journals/comgeo/DumitrescuGPW01} a configuration is shown that has $\Omega\left(2^{2n - \Theta(\log(n))}\right)$ T-paths, which is essentially $4^{n}$, so the quadratic term in the running time of Aichholzer's algorithm becomes really expensive. The first contribution of ours that will be shown is the following theorem:

\begin{theorem}\label{c-tri:theorems:our-t-paths}
	Let $\setp$ be a given set of $n$ points on the plane. Then the exact value of $|\F_{T}(\setp)|$ can be computed in $O\left(n^{3}\cdot t\right)$ time, and $O(t)$ space, where $t$ is the largest number of T-paths the algorithm encounters when run on $\setp$. Moreover $t = O(9^{n})$.
\end{theorem}

Thus the running time of our algorithm for computing $|\F_{T}(\setp)|$, based on T-paths, can really be seen as an asymptotic improvement over Aichholzer's algorithm. As for the upper bound on $t$, ours is the first non-trivial bound on it to be known, however, we suspect that the real value should be closer to $4^{n}$. 

Now, no configuration of points is known having as many T-paths as triangulations. Hence, our T-path-based algorithm could potentially count triangulations \emph{asymptotically} faster than enumeration algorithms. No similar result was known before, which makes ours worth mentioning. On the negative side, the bound for the running time of our algorithm is very precise, it depends on the \emph{largest}\footnote{Since T-paths are referenced by a line, different lines might generate different numbers of T-paths.} number of T-paths the algorithm encounters when run on $\setp$, and this number can get very large, sometimes at least $\Omega(4^{n})$. 

\subsection{The result on counting pseudo-triangulations}

Pseudo-triangulations have been the subject of extensive research from the counting point of view, see~\cite{DBLP:conf/cccg/RandallRSS01,DBLP:journals/comgeo/Bereg05} and references therein. As of today it is not known whether, for \emph{any} set of points, the number of pointed pseudo-triangulations is at least as large as its number of triangulations. Observe that if we remove the pointedness condition, the answer is trivially ``yes''.

In~\cite{DBLP:conf/wads/AichholzerRSS03} the concept of \emph{zig-zag path of a pseudo-triangulation} was introduced. This concept is for pseudo-triangulations what T-paths are for triangulations. For simplicity and consistency we will call such zig-zag paths simply \emph{PT-paths}. 

\begin{definition}[PT-path]\label{c-tri:def:pt-paths}
	Given a planar set of points $\setp$, a pseudo-triangulation $S$ of $\setp$, and a separating line $l$ w.r.t.~$\setp$, a PT-path of $S$ w.r.t.~$l$, denoted by $\pt{l}{S}$, is defined as follows: (\oldstylenums{1}) $\pt{l}{S}$ is a chain of edges of $S$ whose starting and ending edges are two edges of $\Conv(\setp)$ intersected by $l$, and whose intersections with $l$ are linearly ordered along $l$. (\oldstylenums{2}) The area bounded by $\pt{l}{S}$, between two consecutive intersections with $l$, and line $l$ is an empty pseudo-triangle. (\oldstylenums{3}) The reflex vertices of the empty pseudo-triangles of (\oldstylenums{2}) are pointed in $S$. See to the right in Figure~\ref{c-tri:figs:1} for an example of a PT-path $\pt{l}{S}$.
\end{definition}

As for T-paths, an equivalent of Theorem~\ref{c-tri:theorems:t-pathsOswin} for PT-paths was proven in~\cite{DBLP:conf/wads/AichholzerRSS03}:

\begin{theorem}[O. Aichholzer, G. Rote, B. Speckmann, I. Streinu]\label{c-tri:theorems:pt-pathsOswin}
	The PT-path $\pt{l}{S}$ of a pseudo-triangulation $S$ w.r.t.~separating line $l$ always exists and is \emph{unique}.
\end{theorem}

The previous theorem does not necessarily hold if we remove the pointedness condition, that is, a non-pointed pseudo-triangulation might contain more than one PT-path for the same reference line $l$. Nonetheless, for such cases one can still define a ``canonical'' PT-path.

Again, as for T-paths, divide-and-conquer algorithms that use PT-paths can be devised to count the elements of $\F_{PT}(\setp)$, one such algorithm was already present in~\cite{DBLP:conf/wads/AichholzerRSS03}. Those algorithms, as for T-paths, end up having running times of the sort $O\left(t^{2}\right)$, where $t = t(n)$ is the largest number of PT-paths of $\setp$, w.r.t.~to some separating line $l$, that the algorithm encounters.

The result on pseudo-triangulation that we will prove is the following:

\begin{theorem}\label{c-tri:theorems:pt-paths}
	Let $\setp$ be a given set of $n$ points on the plane. Then the exact value of $|\F_{PT}(\setp)|$ can be computed in $O\left(n^{7}\cdot t\right)$ time, and $O(t)$ space, where $t$ is the largest number of PT-paths the algorithm encounters when run on $\setp$.
\end{theorem}

Thus again, our result gives a significant improvement over known algorithms for counting pseudo-triangulations. This time, however, we are not able to show an upper bound on the largest number of PT-paths that can be constructed w.r.t.~a given line.

The rest of the paper is organized as follows: In~\S~\ref{c-tri:sections:t-paths} we prove Theorem~\ref{c-tri:theorems:our-t-paths} and in~\S~\ref{c-tri:sections:pt-paths} we prove Theorem~\ref{c-tri:theorems:pt-paths}. We close the paper in~\S~\ref{c-tri:sections:conclusionsT-ST} with discussions and conclusions.

\section{Counting triangulations}\label{c-tri:sections:t-paths}

Let $T$ be a triangulation of $\setp$ and let $l$ be a separating line w.r.t.~$\setp$. Without loss of generality we will assume that $l$ is vertical. Let $e$ be an edge of $T$ properly intersecting $l$. If $e$ is not an edge of $\Conv(\setp)$, we will say that $e$ is \emph{flippable} iff the union $Q$ of the two triangles of $T$ sharing $e$ forms a convex quadrilateral. If $Q$ is non-convex, or $e$ is an edge of $\Conv(\setp)$, we will simply say that $e$ is \emph{non-flippable}. Also, for $Q$, we will call the two vertices that are not vertices of $e$, the \emph{opposite vertices} of $e$. Finally, we will say that $e$ is \emph{good} with respect to $l$ iff $e$ is flippable and its opposite vertices lie on different sides of $l$. 

Now, let $\tp{l}{T}$ be a T-path of $T$. The region between two consecutive edges $e = ab, e^{\prime}=bd$ of $\tp{l}{T}$, and delimited by $l$, defines a wedge $W = abd$ with apex at vertex $b$, see Figure~\ref{c-tri:sections:t-paths:figs:60}. 

\begin{figure}[!htb]
	\begin{center}
		\includegraphics[height=4cm]{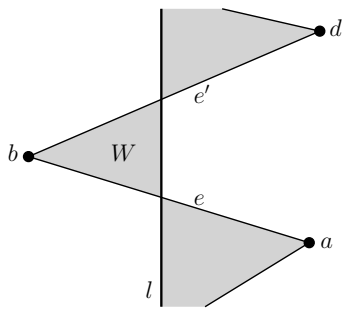}
		\caption{Vertices $a, b, d$ are three consecutive vertices of the shown T-path.}
		\label{c-tri:sections:t-paths:figs:60}
	\end{center}
\end{figure}

Observe that by part (\oldstylenums{3}) of Definition~\ref{c-tri:def:t-paths}, wedge $W$ is empty of points of $\setp$, so we can define the set $\W = \W(\tp{l}{T}) = \{W_{1}, W_{2}\ldots, W_{k}\}$, as the set of all those empty wedges. Since we are going to use wedges throughout the whole section, whenever we have three consecutive vertices $a,b,d$ of $\tp{l}{T}$, we will use the shorthand $abd$ to denote the corresponding element of $\W$  formed by the triple, in which the middle element is the apex. We now have the following observations:

\begin{lemma}\label{c-tri:sections:t-paths:lemmas:1}
	Let $T$ be a triangulation of $\setp$, let $l$ be a vertical line, and let $e$ be a good edge of $T$ w.r.t.~$l$. Then $e$ is an edge of the unique T-path $\tp{l}{T}$.
\end{lemma}
\begin{proof}
Assume for the sake of contradiction that edge $e = pq$ of $T$ is good but not an edge of $\tp{l}{T}$, that is, edge $e$ cannot be an edge of $\Conv(\setp)$.  Let $\W$ be the set of empty wedges of $\tp{l}{T}$. Observe that every element $W$ of $\W$ defines an interval on $l$, which is precisely where $W$ intersects $l$, see Figure~\ref{c-tri:sections:t-paths:figs:1}. 

\begin{figure}[!hbt]
	\begin{center}
		\includegraphics[height=4cm]{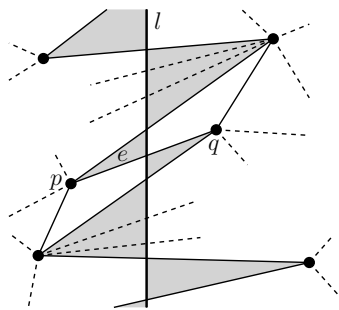}
	\end{center}
	\caption{Every empty wedge of $\tp{l}{T}$ defines an interval on $l$ where they intersect.}
	\label{c-tri:sections:t-paths:figs:1}
\end{figure}

Note that every interior point of an interval on $l$ defined by some element of $\W$ belongs only to that element of $\W$, that is, two intervals defined by two different elements of $\W$ have disjoint interiors. Denote by $x$ the point of intersection between $e$ and $l$. This point $x$ cannot be the boundary point of any interval on $l$ defined by some element of $\W$, otherwise there would be an edge $e^{\prime}\neq e$ of $\tp{l}{T}$ that crosses $l$ at $x$, but that would mean that $e$ and $e^{\prime}$ intersect, which is clearly impossible since both edges belong to $T$, see Figure~\ref{c-tri:sections:t-paths:figs:2}. 
\begin{figure}[!hbt]
	\begin{center}
		\begin{minipage}[b][6.5cm][t]{7cm}%
			\begin{center}
				\includegraphics[height=4cm]{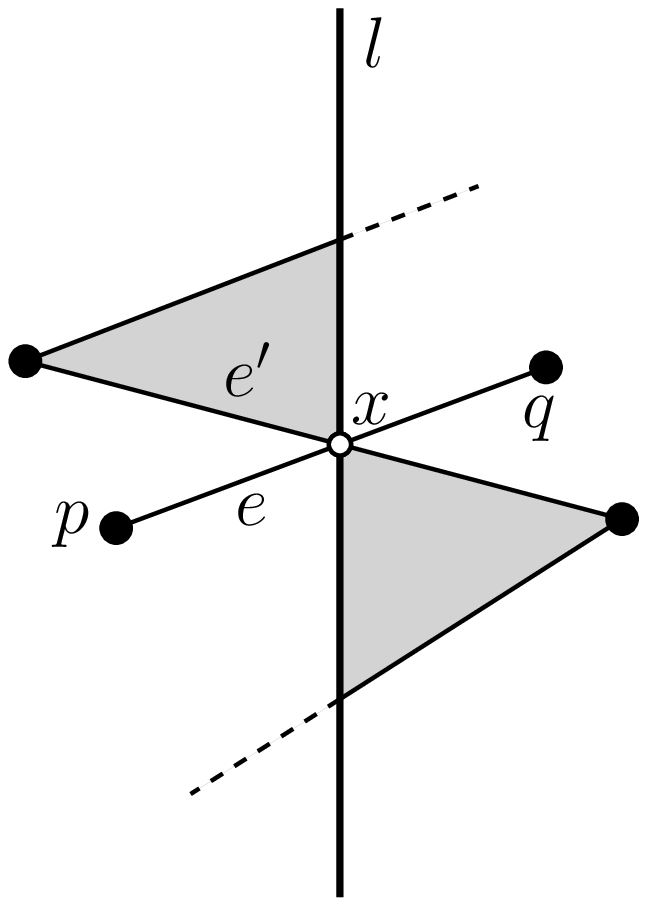}
				\caption{The intersection between $e$ and $l$ cannot be the boundary of an interval on $l$ defined by an empty wedge of $\tp{l}{T}$.}%
				\label{c-tri:sections:t-paths:figs:2}
			\end{center}
		\end{minipage}
		\quad
		\begin{minipage}[b][6.5cm][t]{7cm}%
			\begin{center}
				\includegraphics[height=4cm]{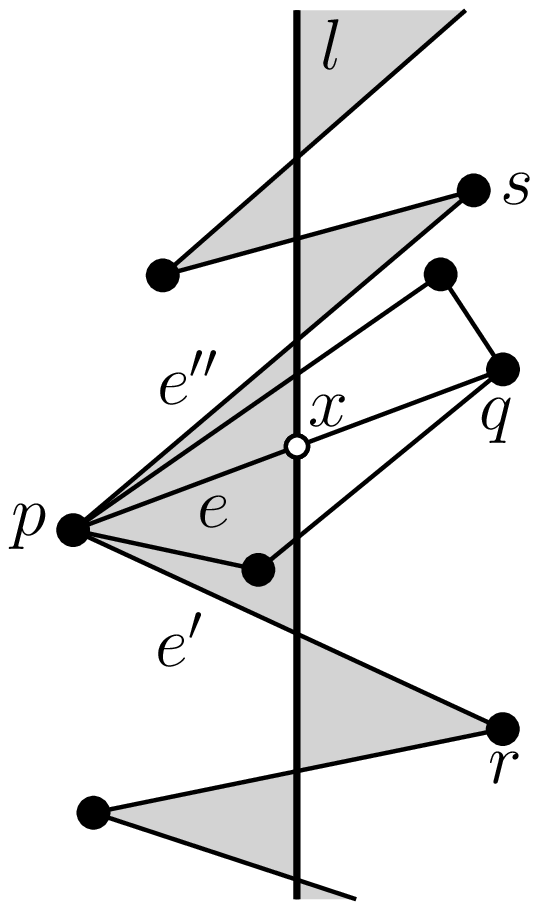}%
				\caption{Point $x$ lies in the interior of the interval of $l$ defined by the empty wedge with apex $p$. Since $e$ is good w.r.t.~$l$, the third vertex of one of the triangles of $T$ that share $e$ must lie inside $W$.}%
				\label{c-tri:sections:t-paths:figs:3}%
			\end{center}
		\end{minipage}%
	\end{center}
\end{figure}

	Thus $x$ must belong to the interior of some interval on $l$ defined by some element $W$ of $\W$. It is also clear that the apex of $W$ must be either $p$ or $q$, otherwise, either $p$ or $q$ lies inside $W$, which is not possible since $W$ is an empty wedge of $\tp{l}{T}$. Let us assume without loss of generality that the apex of $W$ is $p$, and that $W$ is defined by the two consecutive edges $e^{\prime} = rp$ and $e^{\prime\prime} = ps$ of $\tp{l}{T}$. Assume without loss of generality that $p$ lies to the left of $l$, and thus $r,q,s$ lie to the right. Note that $x$ lies between the intersection points of $e^{\prime}$ and $e^{\prime\prime}$ with $l$, see Figure~\ref{c-tri:sections:t-paths:figs:3}. Since $e$ is good w.r.t.~$l$, then the two triangles of $T$ sharing $e$ have their third vertices on different sides of $l$, which means that one of them necessarily lies inside $W$, which is again a contradiction since $W$ is empty of vertices of $T$. Thus $e$ must belong $\tp{l}{T}$.
\end{proof}
\indent Observe that in general a T-path can also contain non-flippable edges.
\begin{lemma}\label{c-tri:sections:t-paths:lemmas:2}
	Let $T$ be a triangulation with vertex set $\setp$, and let $e$ be a flippable edge of $T$. Then there exists a line $l$ such that $e$ is an edge of the T-path $\tp{l}{T}$.
\end{lemma}
\begin{proof}
	Let $e = pq$ be a given flippable edge. Then $e$ cannot be an edge of $\Conv(\setp)$, thus, $e$ is shared by two triangles of $T$, the third point of each triangle is $r$ and $s$ respectively.  Let $e^{\prime} = rs$ be the other diagonal of the convex polygon $prqs$, see Figure~\ref{c-tri:sections:t-paths:figs:4}. Let $l$ be the vertical line containing the point of intersection between $e$ and $e^{\prime}$. Then $l$ makes $e$ and $e^{\prime}$ good.\qedhere
\begin{figure}[!hbt]
	\begin{center}
		\begin{minipage}[b][6.4cm][t]{7cm}%
			\begin{center}
				\includegraphics[height=4cm]{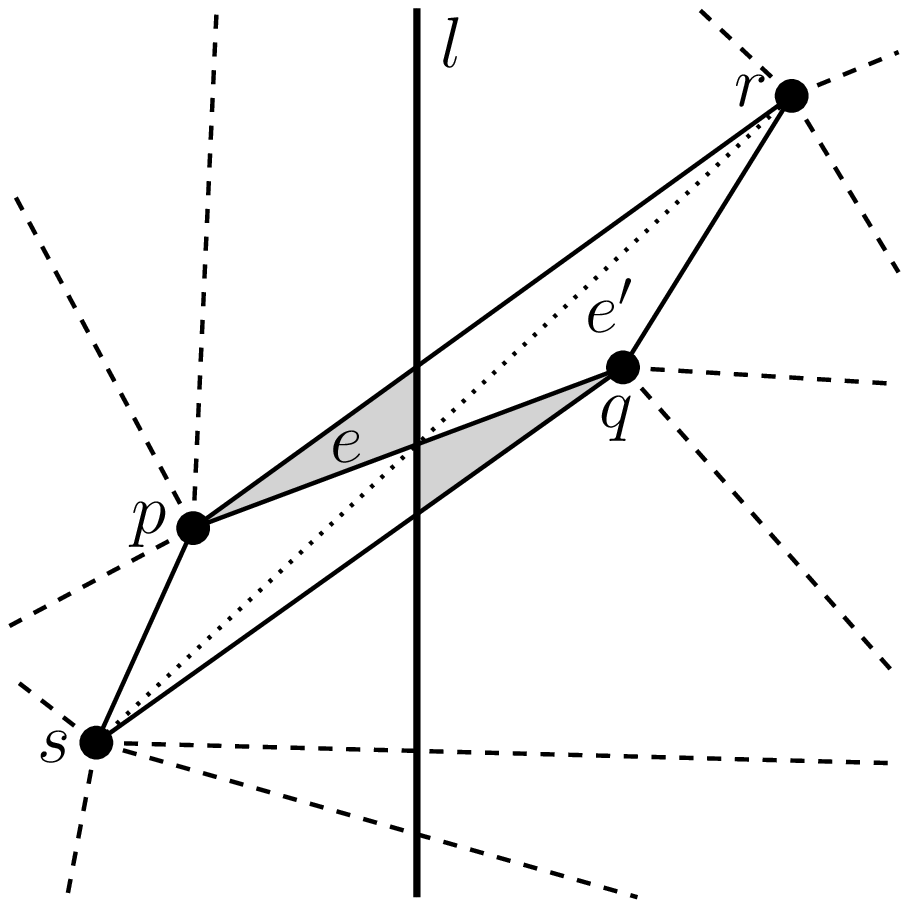}
				\caption{$e$ and $e^{\prime}$ are the two diagonals of the convex quadrilateral $prqs$. The line $l$ containing their intersection makes both, $e$ and $e^{\prime}$ good.}
				\label{c-tri:sections:t-paths:figs:4}
			\end{center}
		\end{minipage}
		\quad
		\begin{minipage}[b][6.4cm][t]{7cm}%
			\begin{center}
				\includegraphics[height=4cm]{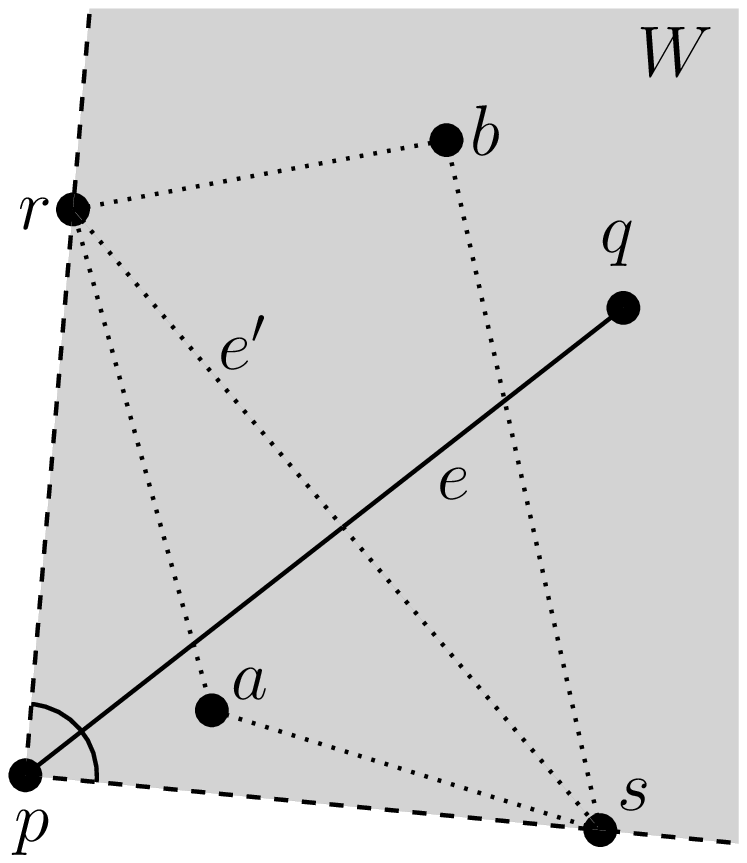}	
				\caption{Vertices $a,b$ must be in the gray zone otherwise angle $\angle rps$ would not be maximum.}
				\label{c-tri:sections:t-paths:figs:5}
			\end{center}
		\end{minipage}%
	\end{center}
\end{figure}	
\end{proof}
\begin{lemma}\label{c-tri:sections:t-paths:lemmas:3}
	Let $T$ be a triangulation with vertex set $\setp$. Then the set of all flippable edges of $T$ is enough to characterize $T$.
\end{lemma}
\begin{proof}
Let $F(T)$ be the set of all flippable edges of $T$. We have to prove that there cannot be another triangulation $T^{\prime}$ with vertex set $\setp$ such that $T\neq T^{\prime}$ but $F(T) = F(T^{\prime})$.
	
	Let us assume for the sake of contradiction that such triangulation $T^{\prime}$ exists. Define the set $NF(T) = E(T)\setminus F(T)$, which is the set of all non-flippable edges of $T$. Clearly, $NF(T)\neq NF(T^{\prime})$, otherwise $T = T^{\prime}$. That is, there must be at least one edge $e\in NF(T)$ that is properly intersected by edges of $NF(T^{\prime})$; it cannot be intersected by edges of $F(T) = F(T^{\prime})$, and both $NF(T), NF(T^{\prime})$ cannot form a set of non-crossing edges since $T$ and $T^{\prime}$ are sets of non-crossing edges of maximum cardinality, but $NF(T)\not\subseteq E(T^{\prime})$ and $NF(T^{\prime})\not\subseteq E(T)$.
	
	Now let $e = pq$, and let $e^{\prime} = rs$ be an edge of $NF(T^{\prime})$ crossing $e$. Clearly, the edges of the quadrilateral $Q = prqs$ cannot be part of either $T$ or $T^{\prime}$ because that would make $e$ and $e^{\prime}$ flippable, see Figure~\ref{c-tri:sections:t-paths:figs:5}. Assume that $e^{\prime}$ is the edge of $NF(T^{\prime})$ crossing $e$ that maximizes the angle $\angle rps$, such $e^{\prime}$ must exist. Given that all the edges of $\Conv(\setp)$ are also shared by $T$ and $T^{\prime}$ we have that $e^{\prime}$ must be shared by two triangles of $T^{\prime}$, so let $a,b$ the third point of each triangle respectively, see Figure~\ref{c-tri:sections:t-paths:figs:5}. Observe that it could happen that $p = a$, but then $b\neq q$, since quadrilateral $Q$ makes $e^{\prime}$ flippable. Or vice-versa, $b = q$, but then $p\neq a$. Then $a,b$ must be contained in the infinite wedge $W = rps$ with apex at $p$. Otherwise, say w.l.o.g.~that $a$ lies outside $W$. This means that another edge of triangle $\triangle rsa$, other than $e^{\prime}$, intersects $e$ properly. Say edge $ra$. But then angle $\angle rpa > \angle rps$, which is a contradiction since $\angle rps$ was chosen to be maximum among all the edges of $NF(T^{\prime})$ crossing $e$. Note however that if $a,b$ are contained in $W$, then the quadrilateral $rasb$ is convex, which means that $e^{\prime}$ is flippable in $T^{\prime}$, which is a contradiction since we assume that $e^{\prime}\in NF(T^{\prime})$. Hence such an edge $e^{\prime}\in NF(T^{\prime})$ crossing $e$ cannot exist, which means that $NF(T^{\prime}) = NF(T)$, since $e$ was an edge of $NF(T)$, and thus we arrive at $T = T^{\prime}$.	
\end{proof}

\begin{lemma}\label{c-tri:sections:t-paths:lemmas:4}
Let $T$ be a triangulation with vertex set $\setp$, and let $l, l^{\prime}$ be two vertical lines such that $l\neq l^{\prime}$, and the vertical slab between $l$ and $l^{\prime}$ is empty of points of $\setp$. Then $\tp{l}{T} = \tp{l^{\prime}}{T}$.
\end{lemma}

\begin{proof}
Let us assume without loss of generality that $l^{\prime}$ lies to the left of $l$. Since the vertical slab between $l^{\prime}$ and $l$ is empty of points of $\setp$, observe that there is a bijection between the set $\W(\tp{l}{T})$, the empty wedges of $\tp{l}{T}$, and the set $\W(\tp{l^{\prime}}{T})$, see Figure~\ref{c-tri:sections:t-paths:figs:6}.
\begin{figure}[!htb]
	\begin{center}
		\includegraphics[height=4cm]{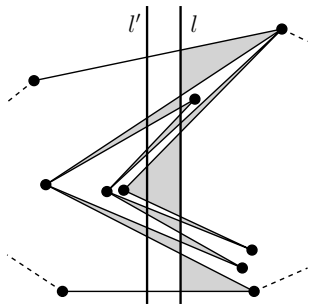}
	\end{center}
	\caption{T-path $\tp{l}{T}$ shown, along its empty wedges. Every wedge is also empty w.r.t.~$l^{\prime}$.}
	\label{c-tri:sections:t-paths:figs:6}
\end{figure}

	Thus $\tp{l}{T}$ and $\tp{l^{\prime}}{T}$ are both T-paths, by definition, of $T$ w.r.t.~$l^{\prime}$ and $l$ respectively, but every T-path of $T$ w.r.t.~some line is unique, so there is no other option but $\tp{l}{T} = \tp{l^{\prime}}{T}$.
\end{proof}

Now let us assume that $\setp$ is sorted from left to right, \emph{i.e.}, from smallest x-coordinate to the largest. We can assume that by a suitable rotation we do not have any ties in the x-coordinate, so $\setp = \{p_{1},p_{2}, \ldots, p_{n}\}$.

Let $\L = \{l_{1},\ldots l_{n-1}\}$ be a set of vertical lines such that point $p_{i}\in \setp$ lies in the vertical slab between $l_{i-1}$ and $l_{i}$, with $2\leq i\leq n-1$. Point $p_{1}$, the leftmost, lies in the unbounded vertical slab to the left of $l_{1}$, and $p_{n}$, the rightmost, lies in the unbounded vertical slab to the right of $l_{n-1}$. For a triangulation $T$ of $\setp$ let $\P(T) = \{\tp{l_i}{T}\ |\ l_{i}\in\L\}$. We now have the following result:

\begin{theorem}\label{theorems:1}
	Let $T$ be a triangulation with vertex set $\setp$. Then $\P(T)$ is enough to characterize $T$.
\end{theorem}
\begin{proof}
	We have to prove that there cannot be another triangulation $T^{\prime}$ with vertex set $\setp$ such that $T^{\prime}\neq T$, but $\P(T) = \P(T^{\prime})$. However, by Lemma~\ref{c-tri:sections:t-paths:lemmas:3} we know that the set of flippable edges of a triangulation characterizes it, hence it is enough to prove that \emph{every} flippable edge of $T$ is an edge of some T-path in $\P(T)$. 
	
	Let us assume for the sake of contradiction that there is an edge $e$ of $T$ that is flippable but that is not an edge of any T-path in $\P(T)$. By Lemma~\ref{c-tri:sections:t-paths:lemmas:2} we know that there exists one vertical line $l$ such that $e$ is an edge of the T-path $\tp{l}{T}$. Note that such a line $l$ is parallel to every line in $\L$, and that one endpoint of $e$ lies to the left of $l$ and the other to the right, so $l$ must lie inside the vertical slab between to consecutive lines of $\L$, or to the left of $l_{1}\in\L$, or to the right of $l_{n-1}\in\L$, let us assume without loss of generality that $l$ lies in the vertical slab between $l_{i}$ and $l_{i+1}$, with $1\leq i\leq n-2$. Observe however that such a slab contains exactly one point of $\setp$, thus it must happen that either, the vertical slab between $l_{i}$ and $l$ is empty of points of $\setp$, or the vertical slab between $l$ and $l_{i+1}$ is empty of points of $\setp$, say the former without loss of generality. Nevertheless, by Lemma~\ref{c-tri:sections:t-paths:lemmas:4}, we know that $\tp{l}{T} = \tp{l_i}{T}$, so $\tp{l}{T}\in\P(T)$, which is a contradiction since $e$ was a flippable edge of $T$ that was not an element of $\P(T)$. Thus, such an edge $e$ cannot exist, and there is no other option but $T = T^{\prime}$ since they share the same set of flippable edges.
\end{proof}

Therefore every triangulation $T$ having $\setp$ as vertex set has a unique set $\P(T)$ of T-paths, and thus the number of triangulations $|\F_{T}(\setp)|$ is just the number of different sets of T-paths $\P(T)$ that we can find on $\setp$. Let $\tcomp(l,\setp) = \{\tp{l}{T}\ |\ T {\textcolor{black}{\text{ is a triangulation of }}} \setp\}$ be the set of all T-paths of $\setp$ w.r.t.~line $l$. Note that while the set of lines $\L$ stays fixed, there will be in general more than one T-path that can be formed per line, thus a tuple $\{\pi_{1},\ldots, \pi_{n-1}\}$ of T-paths of $\setp$, with $\pi_{i}\in\tcomp(l_{i},\setp)$, defines a triangulation if and only if all those T-paths are pairwise non-crossing. We will say that such a pairwise non-crossing set is \emph{compatible}. It is easy to show that, in order to verify if such a set is compatible, it suffices to check that two consecutive T-paths $\pi\in\tcomp(l_{i},\setp)$ and $\pi^{\prime}\in\tcomp(l_{i+1}, \setp)$ are non-crossing, for $1\leq i\leq n-2$.

Note that there might be triangulations sharing some T-paths, for example, if $\setp$ is in convex position, its number of triangulations is $O(4^{n})$, while its number of T-paths is $O(2^{n})$, so we obtain on average $O(2^{n})$ triangulations per T-path. This motivates the following definition:
\begin{align*}
	\T(\pi_{j}) &= \{\{\pi_{1},\ldots, \pi_{j-1}\}\ |\ \{\pi_{1},\ldots,\pi_{j-1},\pi_{j}\} {\textcolor{black}{\text{ is compatible and }}} \pi_{i}\in\tcomp(l_{i},\setp)\}.
\end{align*}

We need two more definitions in order to describe our algorithm. For each $\pi^{\prime}\in\tcomp(l_{i+1},\setp)$ we define $\lambda(\pi^{\prime}) = \{\pi\in\tcomp(l_{i},\setp)\ |\ \pi {\textcolor{black}{\text{ is compatible with }}} \pi^{\prime}\}$. Similarly we define $\mu(\pi) = \{\pi^{\prime}\in\tcomp(l_{i+1}, \setp)\ |\ \pi^{\prime}{\textcolor{black}{\text{ is compatible with }}} \pi\}$ for each $\pi\in\tcomp(l_{i},\setp)$. Now we are ready to describe our algorithm.

\subsection{The sweep line algorithm}\label{c-tri:sections:t-paths:sub-sections:1}

We consider sweeping a vertical line from left to right, the \emph{event points} being the vertical lines in the set $\L$ as defined before. At any event point $l_{i}$ we maintain $\tcomp(l_{i}, \setp)$, and for each $\pi\in\tcomp(l_{i},\setp)$ we store $|\T(\pi)|$. At $i = 1$ we clearly have $|\tcomp(l_{1}, \setp)| = 1$, and for this particular $\pi\in\tcomp(l_{1},\setp)$ we have $|\T(\pi)| = 1$. We will show that each $\pi^{\prime}\in\tcomp(l_{i+1},\setp)$ can be obtained from each $\pi\in\tcomp(l_{i},\setp)$ compatible with $\pi^{\prime}$\footnote{Again, by compatibility we mean non-crossing.} by doing \emph{local changes}, which will be defined later on, for the time being the important thing to know is that the number of possible local changes for a T-path is $O\left(n^{2}\right)$. Hence, if we go through each $\pi\in\tcomp(l_{i}, \setp)$ and try all possible local changes for $\pi$, we will obtain $\tcomp(l_{i+1},\setp)$. Moreover, for each $\pi^{\prime}\in\tcomp(l_{i+1}, \setp)$ we also get the set $\lambda(\pi^{\prime})$. Observe that $|\T(\pi^{\prime})|$ is given by $\sum_{\pi\in\lambda(\pi^{\prime})}|\T(\pi)|$. Thus we are able to compute $\tcomp(l_{i+1},\setp)$ as well as $|\T(\pi^{\prime})|$ for each $\pi^{\prime}\in\tcomp(l_{i+1},\setp)$. All this takes time $O\left(n^{2}\cdot t_{i}\right)$, where $t_{j} = |\tcomp(l_{j},\setp)|$, since there are $O(n^{2})$ local changes to try for each $\pi\in\tcomp(l_{i},\setp)$, and as we will see later, the time taken per local change is constant. The overall running time of the algorithm is therefore $\sum_{l_{j}\in\L}O\left(n^{2}\cdot t_{j}\right)\leq O\left(n^{3}\cdot t\right)$, where $t = \max\{t_{j}\}$. At the end, the number we are looking for is precisely $|\F_{T}(\setp)| = |\T(\pi)|$, where $\pi$ is the unique T-path of $\tcomp(l_{n-1}, \setp)$.

Our main task now is to explain the local changes and to prove that there are indeed $O\left(n^{2}\right)$. We first need the following intermediate result:

\begin{lemma}\label{c-tri:sections:t-paths:lemmas:5}
	At times $l = l_{i}$ and $l = l_{i+1}$, point $p = p_{i+1}$ has degree zero, one, or two in every T-path $\pi\in\tcomp(l_{i}, \setp)$, as well as in every T-path $\pi^{\prime}\in\tcomp(l_{i+1}, \setp)$. However, if $\pi$ and $\pi^{\prime}$ do not cross, then $p$ cannot simultaneously have degree zero in both T-paths, that is, $p$ must be a vertex of at least one T-path.
\end{lemma}
\begin{proof}
	Let us look at the case when $l = l_{i}$, the other case, $l = l_{i+1}$ is just symmetric. If $p$ is not a vertex of $\pi$, then the degree of $p$ is zero. If $p$ is a vertex of $\pi$, then there are two cases depending on whether $p$ is a vertex of $\Conv(\setp)$ or not. Since both cases are very similar we will prove only the latter. 
	
	Since $p$ lies inside $\Conv(\setp)$ we know that $p$ is an internal vertex of $\pi$, \emph{i.e.}, the degree of $p$ in $\pi$ is at least two. To verify that it is at most two let us assume that its degree is at least four, it must be even. Let $b$ be the first neighbor of $p$ in $\pi$, when visiting $p$ while traversing $\pi$ from the first vertex to the last. Similarly, let $c$ be the last neighbor of $p$ in $\pi$ in the same traversing order, see Figure~\ref{c-tri:sections:t-paths:figs:7}. Since the degree of $p$ is at least four, there must be other two vertices $b^{\prime},c^{\prime}$ between $b$ and $c$. Observe that $p$ lies to the right of $l$, and $b,b^{\prime},c,c^{\prime}$ to the left, so there must be at least one vertex $x\neq p$ of $\pi$ connecting $b^{\prime}$ and $c^{\prime}$, however, $x$ should lie inside the vertical slab between $l$ and $l_{i+1}$, which is empty of points of $\setp$ except for $p$, see Figure~\ref{c-tri:sections:t-paths:figs:8}. Thus $x$ cannot exist, which implies that $b^{\prime},c^{\prime}$ cannot exist either. Hence, the degree of $p$ in $\pi$ is at most two, which is what we wanted to prove.
	
\begin{figure}[!htb]
	\begin{center}
		\begin{minipage}[b][5.7cm][t]{7cm}
			\begin{center}
				\includegraphics[height=4cm]{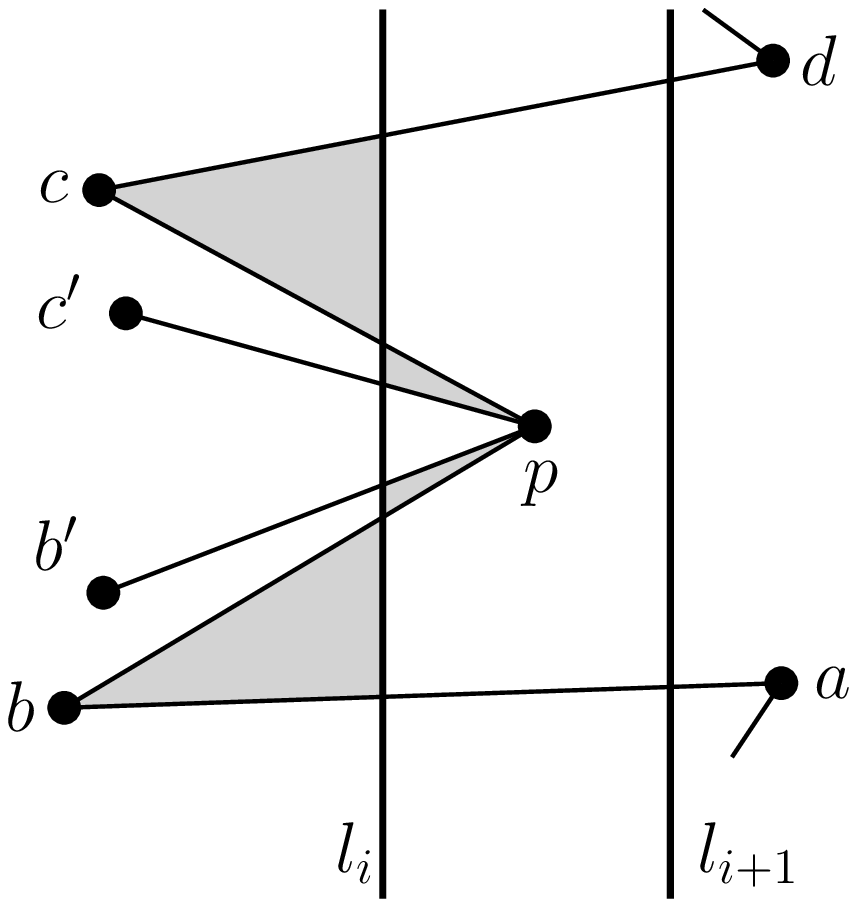}
				\caption{T-path $\pi\in\tcomp(l_{i},\setp)$ where $p$ has degree at least four shown.}
				\label{c-tri:sections:t-paths:figs:7}
			\end{center}
		\end{minipage}
		\quad
		\begin{minipage}[b][5.7cm][t]{7cm}
			\begin{center}
				\includegraphics[height=4cm]{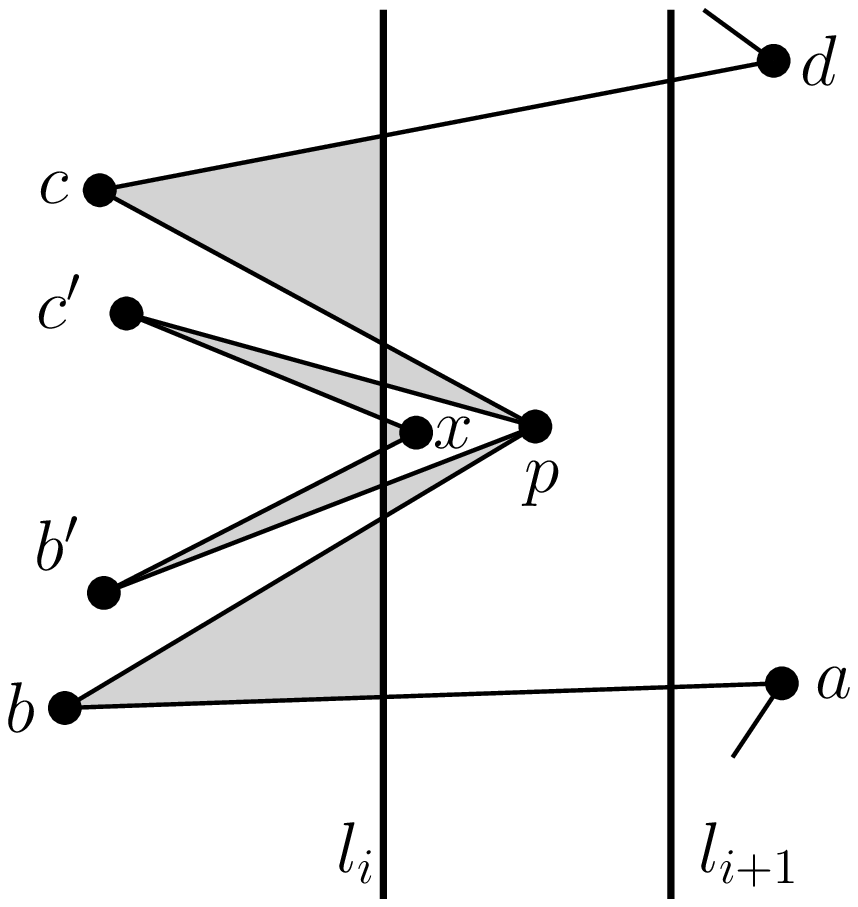}
				\caption{Vertex $x$ of $\pi$ cannot exist because $p$ is the only point in that vertical slab.}
				\label{c-tri:sections:t-paths:figs:8}
			\end{center}
		\end{minipage}
	\end{center}
\end{figure}

	It remains to prove that $p$ cannot have degree zero in both T-paths, $\pi$ and $\pi^{\prime}$, if they do not cross. To see this, note that if neither $\pi$ nor $\pi^{\prime}$ has $p$ as a vertex, then clearly $p$ cannot be on $\Conv(\setp)$, so $p$ must lie in the interior of $\Conv(\setp)$, and thus it also lies inside the triangles $\triangle abd$, and $\triangle a^{\prime}b^{\prime}d^{\prime}$, where $a,b,d$ and $a^{\prime},b^{\prime},d^{\prime}$ are consecutive vertices of $\pi$ and $\pi^{\prime}$ respectively, see Figures~\ref{c-tri:sections:t-paths:figs:9} and~\ref{c-tri:sections:t-paths:figs:10}. Note however that this case can only happen if either $\triangle abd$ and $\triangle a^{\prime}b^{\prime}d^{\prime}$ intersect, or if one lies entirely inside the other, since both triangles contain $p$ in their interior. In the first case we have obviously an intersection between $\pi$ and $\pi^{\prime}$, which is a contradiction. In the second case, assume without loss of generality that $\triangle abd$ lies inside $\triangle a^{\prime}b^{\prime}d^{\prime}$. But then observe that since $a,d$ and $b^{\prime}$ lie on the same side of $l_{i+1}$, then the wedge $a^{\prime}b^{\prime}d^{\prime}$ of $\pi^{\prime}$ is not empty, which is clearly not possible since $\pi^{\prime}$ is a T-path and edges $a^{\prime}b^{\prime}$ and $b^{\prime}d^{\prime}$ are consecutive in $\pi^{\prime}$, see Figure~\ref{c-tri:sections:t-paths:figs:10}. Thus, Lemma~\ref{c-tri:sections:t-paths:lemmas:5} follows.\qedhere

\begin{figure}[!htb]
	\begin{center}
		\begin{minipage}[b][5.7cm][t]{7cm}
			\begin{center}
				\includegraphics[height=4cm]{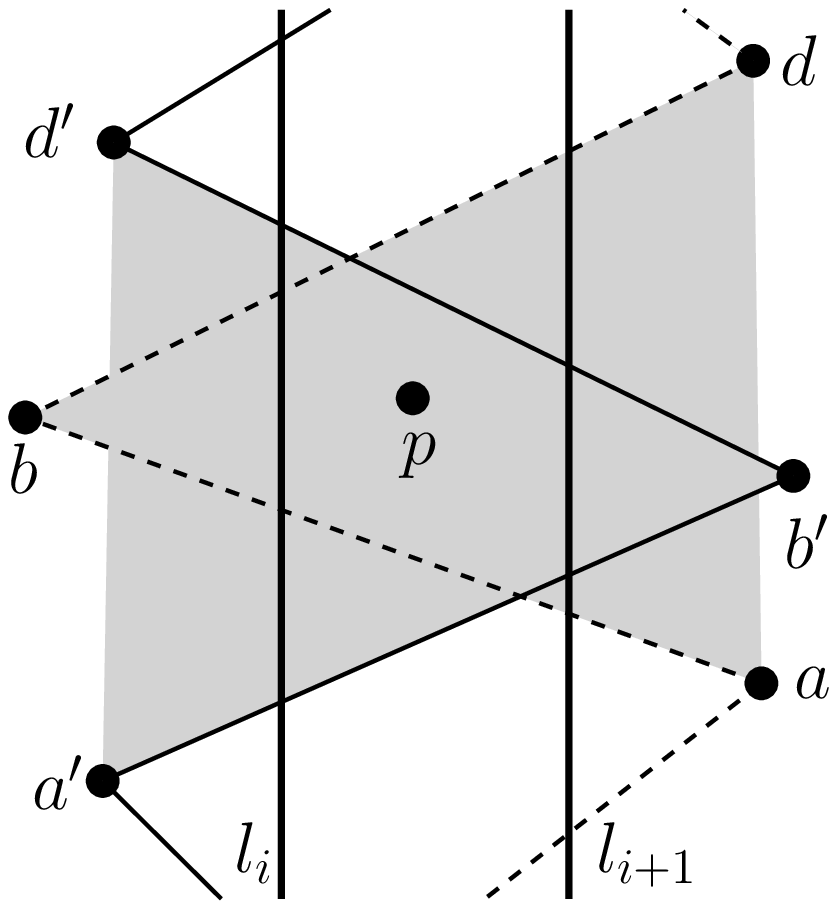}
				\caption{$\pi$ is shown in solid lines, and $\pi^{\prime}$ in dashed lines.}
				\label{c-tri:sections:t-paths:figs:9}
			\end{center}
		\end{minipage}
		\quad
		\begin{minipage}[b][5.7cm][t]{7cm}
			\begin{center}
				\includegraphics[height=4cm]{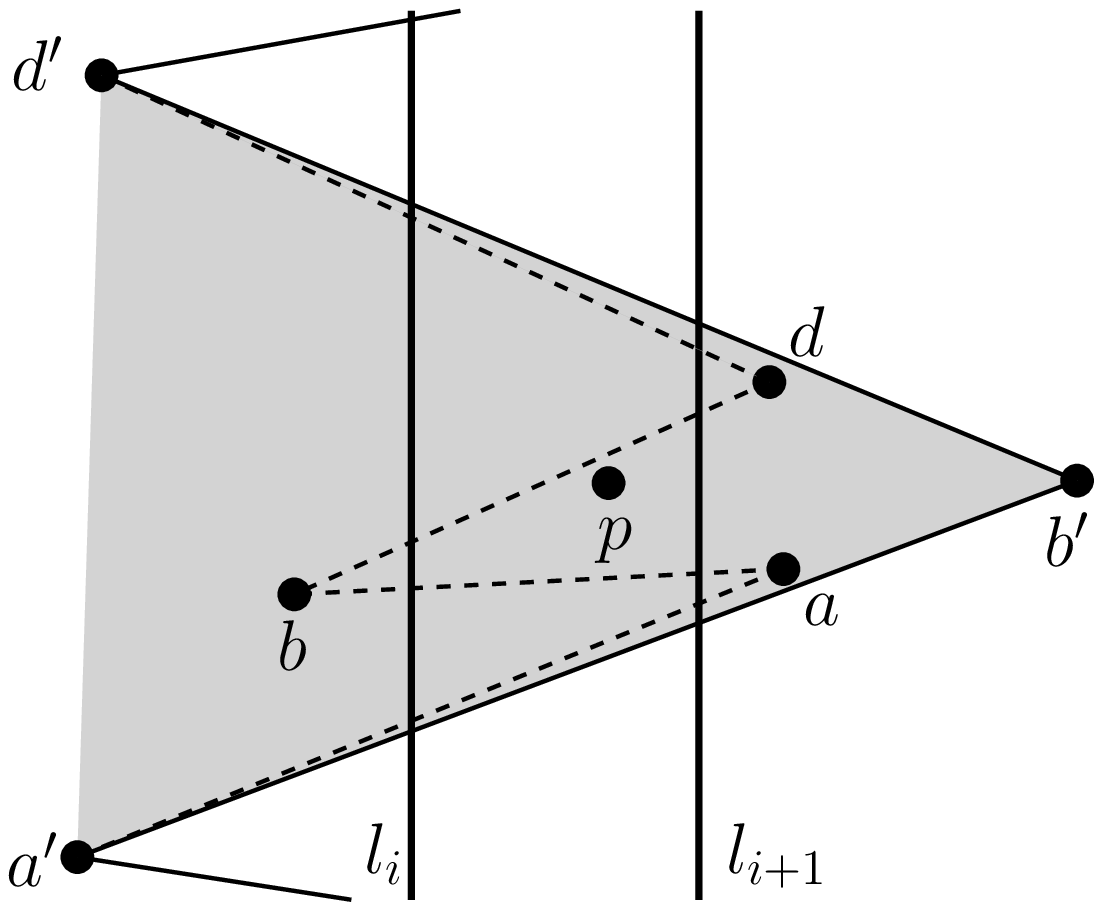}
				\caption{If $\triangle abd$ lies inside $\triangle a^{\prime}b^{\prime}d^{\prime}$, then the wedge $a^{\prime}b^{\prime}d^{\prime}$ with apex $b^{\prime}$ and delimited by $l_{i+1}$ is not empty.}	
				\label{c-tri:sections:t-paths:figs:10}
			\end{center}
		\end{minipage}
	\end{center}
\end{figure}
\end{proof}

We are now ready to explain the local changes carefully: From Lemma~\ref{c-tri:sections:t-paths:lemmas:4} we know that $\tp{l}{T} = \tp{l^{\prime}}{T}$ for a triangulation $T$ of $\setp$ as long as the vertical slab between $l$ and $l^{\prime}$ is empty of points of $\setp$.  This in turn implies that $\tcomp(l,\setp) = \tcomp(l^{\prime},\setp)$. Now assume that $l^{\prime} = l_{i}$ and $l = l_{i+1}$, that is, the vertical slab between $l$ and $l^{\prime}$ is no longer empty, but contains point $p = p_{i+1}$. It is clear that during the continuous movement from $l_{i}$ to $l_{i+1}$ the only ways a T-path can change, are the ones involving $p$ in the following two senses: If $p$ is not a vertex of the current T-path $\pi\in\tcomp(l_{i},\setp)$, then the only empty wedge of $\pi$ that cannot be made an empty wedge of a T-path $\pi^{\prime}\in\tcomp(l_{i+1},\setp)$ is the one that during the sweeping process starts containing $p$, see Figure~\ref{c-tri:sections:t-paths:figs:12}. If  on the other hand, $p$ is a vertex of $\pi$,  then its neighbors in $\pi$ lie to the left of $l_{i}$, since $p$ lies to the right, see Figure~\ref{c-tri:sections:t-paths:figs:13}. But then $p$ along with its neighbors lie to the left of $l_{i+1}$, so those adjacencies cannot be part of a T-path w.r.t.~$l_{i+1}$. Thus we will obtain $\mu(\pi)$, for every T-path $\pi\in\tcomp(l_i,\setp)$, by locally changing $\pi$ around $p$. We will have two cases to consider depending on whether $p$ appears as a vertex of the current T-path $\pi\in\tcomp(l_i, \setp)$ we are considering, or not. We will study each case in turn, however, there is a case analysis that one has to do, so in order to avoid going through all cases, we will describe the general setting from which all the cases can be obtained. Let again $\pi\in\tcomp(l_{i},\setp)$:

\begin{figure}[!htb]
	\begin{center}
		\begin{minipage}[b][5.7cm][t]{7cm}
			\begin{center}
				\includegraphics[height=4cm]{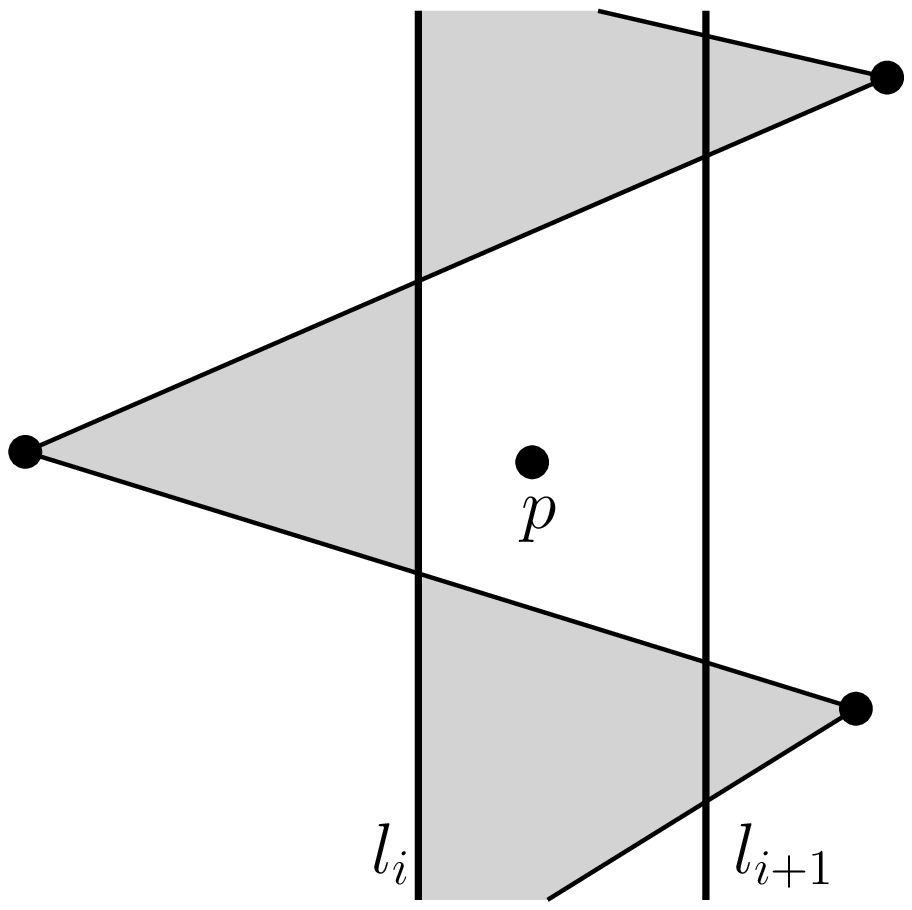}
				\caption{Sweeping from $l_{i}$ to $l_{i+1}$ results in a wedge containing $p$.}
				\label{c-tri:sections:t-paths:figs:12}
			\end{center}
		\end{minipage}
		\quad
		\begin{minipage}[b][5.7cm][t]{7cm}
			\begin{center}
				\includegraphics[height=4cm]{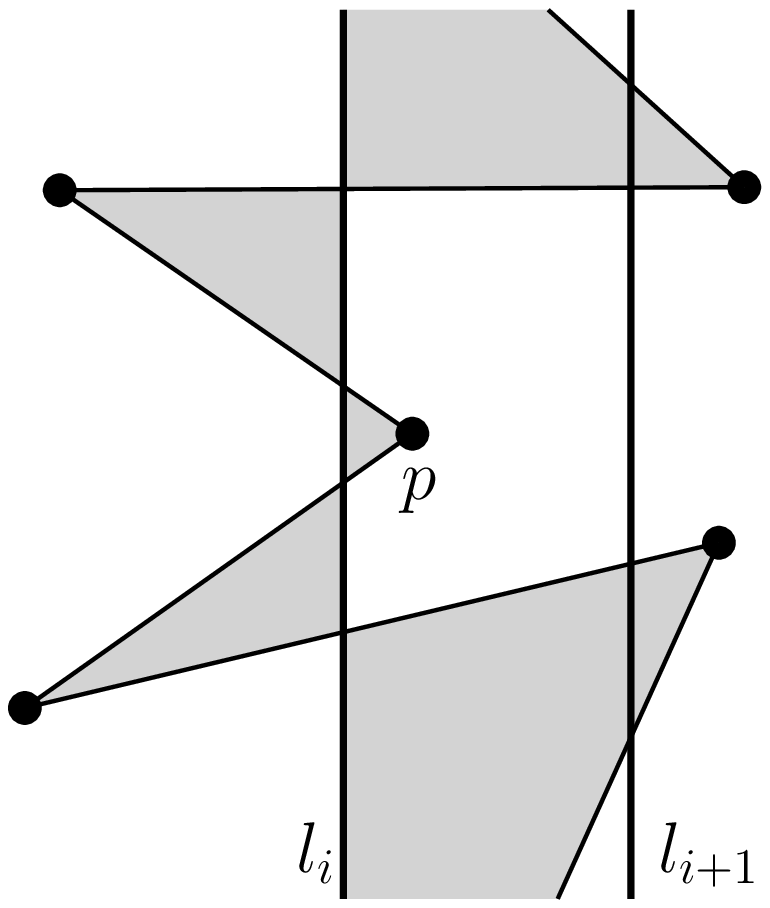}
				\caption{Sweeping from $l_{i}$ to $l_{i+1}$ results in the adjacencies of $p$ being on the same side of $l_{i+1}$.}
				\label{c-tri:sections:t-paths:figs:13}
			\end{center}
		\end{minipage}
	\end{center}
\end{figure}
\begin{enumerate}[label=(\oldstylenums{\arabic*})]

\item Assume that $p$ appears as a vertex of $\pi$, and let us first consider the case when $p$ lies in the interior of $\Conv(\setp)$. By Lemma~\ref{c-tri:sections:t-paths:lemmas:5} point $p$ must have degree exactly two in $\pi$.

Now take vertices $a,b,c,d$ of $\pi$ as displayed in Figure~\ref{c-tri:sections:t-paths:figs:14}. Look for all pairs of points $b^{\prime},c^{\prime}\in \setp$ such that the substitution of the pattern $(a,b,p,c,d)$ in $\pi$ to $(a,b,b^{\prime},p,c^{\prime},c,d)$ results in a T-path w.r.t.~$l_{i+1}$, see Figure~\ref{c-tri:sections:t-paths:figs:14}.

\begin{figure}[!htb]
	\begin{center}
		\begin{minipage}[b][5cm][t]{7cm}
			\begin{center}
				\includegraphics[height=4cm]{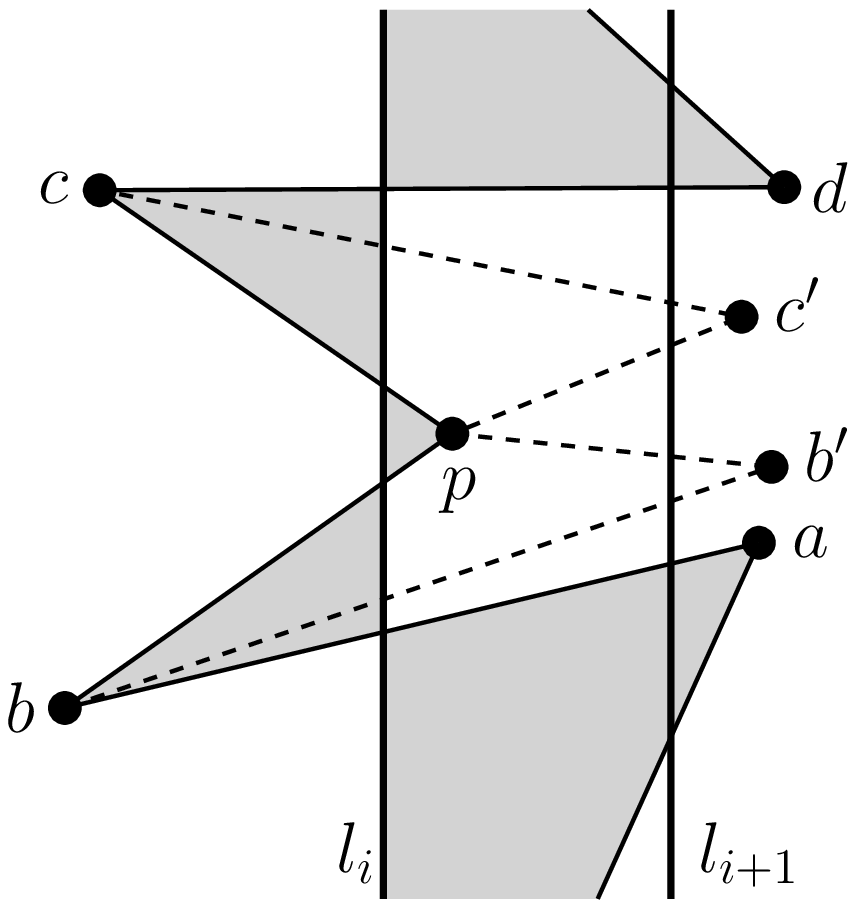}
				\caption{Substitution $(a,b,p,c,d)\rightarrow(a,b,b^{\prime},p,c^{\prime},c,d)$ is only one of the possibilities.}
				\label{c-tri:sections:t-paths:figs:14}
			\end{center}
		\end{minipage}
	\quad
		\begin{minipage}[b][5cm][t]{7cm}
			\begin{center}
				\includegraphics[height=4cm]{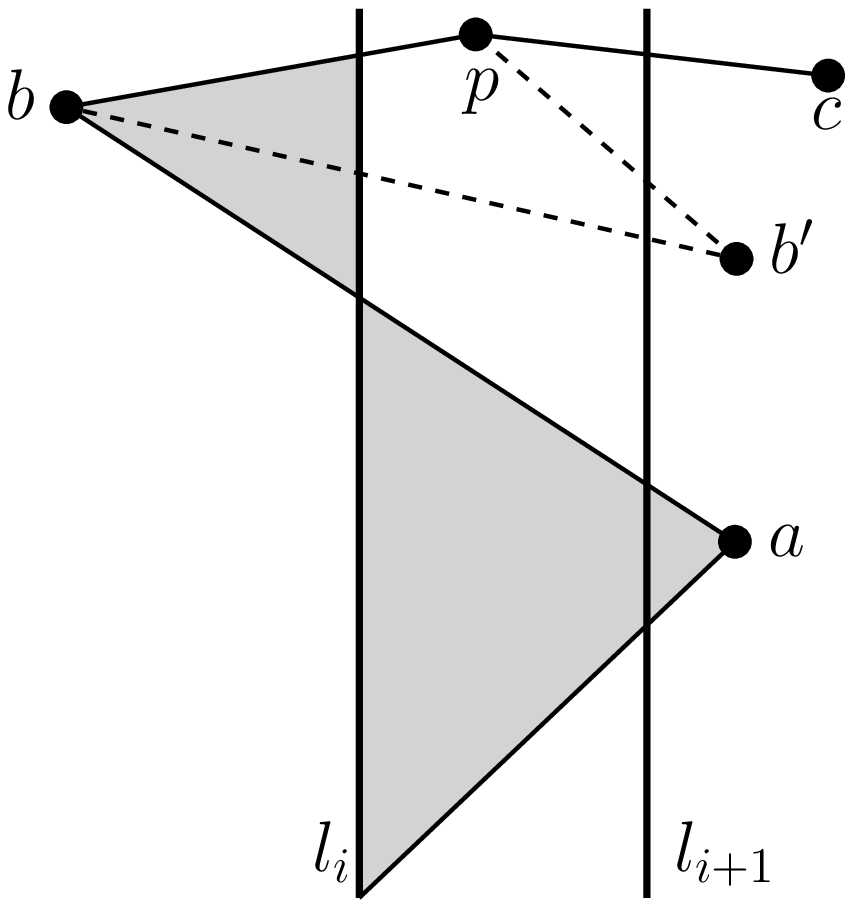}
				\caption{Substitution $(a,b,p)\rightarrow(a,b,b^{\prime},p,c)$.}
				\label{c-tri:sections:t-paths:figs:15}
			\end{center}
		\end{minipage}
	\end{center}
\end{figure}

Observe that as particular cases we could have $b^{\prime} = c^{\prime} = a = d$, which would result in the substitution $(a,b,p,c,d)\rightarrow (a)$, or we could have $b^{\prime} = c^{\prime} = a$, $a\neq d$, which would result in $(a,b,p,c,d)\rightarrow(a,c,d)$. Since there are many cases, we would have to exhaust all choices for $b^{\prime},c^{\prime}$, however, they all occur inside the same region.

If $p\in \Conv(\setp)$, then $p$ could be the very first vertex of $\pi$, or the very last, or the second, or second-to-the-last. Let us consider when $p$ is the last, it is symmetric to the case when $p$ is the first. Let the last three vertices of $\pi$ be $a,b,p$ in that order, so $b\in \Conv(\setp)$ as well, and $bp$ is intersected by $l_{i}$. We are looking in general for the substitution $(a,b,p)\rightarrow (a,b,b^{\prime},p, c)$, where $c\in \Conv(\setp)$ is the other neighbor of $p$ on $\Conv(\setp)$. Observe that $pc$ is intersected by $l_{i+1}$, see Figure~\ref{c-tri:sections:t-paths:figs:15}. We could for example have $b^{\prime} = c$ or $b^{\prime} = a$ as particular cases, among others.

\item Now assume $p$ does not appear as a vertex of $\pi$. Then $p$ cannot be a vertex of $\Conv(\setp)$ either, as otherwise one of the edges of $\Conv(\setp)$ having $p$ as a vertex would intersect $l_{i}$, and thus $p$ would necessarily appear in $\pi$ by definition. Thus $\pi$ must look locally as in Figure~\ref{c-tri:sections:t-paths:figs:16}, that is, the point $p$ must be contained inside the triangle $\triangle abd$, where $a,b,d$ are consecutive on $\pi$, point $b$ lies on one side of $l_{i}$, and $a,d$ on the other side.  Thus observe that the adjacency $bp$ is forced in any triangulation containing $\pi$, since $p$ is the only point of $\setp$ contained in the vertical slab between $l_{i}$ and $l_{i+1}$. The reader will be able to verify that this case is a particular case of (\oldstylenums{1}) in which $b = c$, and we could have, for example, the substitutions $(a,b=c,d)\rightarrow(a,p,d)$, or $(a,b=c,d)\rightarrow(a,b,b^{\prime},p,c^{\prime},b,d)$, among others, see Figure~\ref{c-tri:sections:t-paths:figs:15}.

\begin{figure}[!htb]
	\begin{center}
			\begin{center}
				\includegraphics[height=4cm]{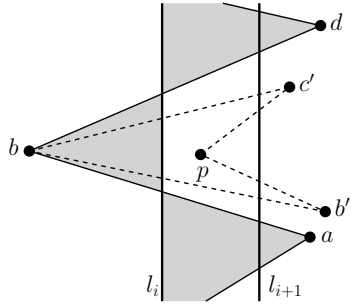}
				\caption{Substitution $(a,b,d)\rightarrow(a,b,b^{\prime},p,c^{\prime},b,d)$.}
				\label{c-tri:sections:t-paths:figs:16}
			\end{center}
	\end{center}
\end{figure}

\end{enumerate}

Note that the substitutions can be done in reverse order, that is, imagine that we go back in time, from time $l = l_{i+1}$ to time $l = l_{i}$, so we would be sweeping the plane from right to left, and therefore the pattern $(b,b^{\prime},p,c^{\prime},c)$ of some $\pi^{\prime}\in\tcomp(l_{i+1}, \setp)$ could become pattern $(b,p,c)$ of some $\pi\in\tcomp(l_{i},\setp)$, upon proper relabeling of points, see Figures~\ref{c-tri:sections:t-paths:figs:14}, ~\ref{c-tri:sections:t-paths:figs:15} and~\ref{c-tri:sections:t-paths:figs:16}. So $\pi^{\prime}$ is obtained from $\pi$ in one direction, and $\pi$ is obtained from $\pi^{\prime}$ in the opposite direction, this relation will be denoted by $\pi\leftrightarrow\pi^{\prime}$. We have finally the following result:

\begin{lemma}\label{lemmas:local}
Given $\tcomp(l_{i}, \setp)$, every T-path of $\tcomp(l_{i+1}, \setp)$ is produced by the local changes just explained. Moreover, for each $\pi\in\tcomp(l_{i}, \setp)$, the cardinality of $\mu(\pi)$ is $O\left(n^{2}\right)$, and we can correctly compute $\lambda(\pi^{\prime})$ for each $\pi^{\prime}\in\tcomp(l_{i+1},\setp)$ in time $O(n^{2}\cdot t_{i})$.
\end{lemma}
\begin{proof}
Let again $p = p_{i+1}\in\setp$. For the first part let $\pi^{\prime}\in\tcomp(l_{i+1}, \setp)$. We will prove that $\pi^{\prime}$ produces at least one T-path $\pi\in\tcomp(l_{i}, \setp)$. The result will then follow by the relation $\pi\leftrightarrow\pi^{\prime}$ explained before. For the second part we have to show that $|\mu(\pi)| = O\left(n^{2}\right)$ for each $\pi\in\tcomp(l_{i}, \setp)$, and that we are able to correctly compute $\lambda(\pi^{\prime})$ for each $\pi^{\prime}\in\tcomp(l_{i+1}, \setp)$ in time $O\left(n^{2}\cdot t_{i}\right)$. That is, we will prove that if $\pi\not\leftrightarrow\pi^{\prime}$ then both T-paths cross, and thus $\pi\not\in\lambda(\pi^{\prime})$. For both parts we have two cases depending on whether $p$ is a vertex of $\pi^{\prime}$ or not, but for simplicity we will only consider the case when $p$ is not a vertex of $\pi^{\prime}$, the other case in both parts follows using similar arguments.

Let $W$ be the empty wedge of $\pi^{\prime}$ that cannot be extended to an empty wedge $W^{\prime}$ of $\pi$ due to $p$. Thus $p$ lies inside the triangle $\triangle abd$, where $a,b,d$ are consecutive vertices, see Figure~\ref{c-tri:figs:12:a}. Let $ap$, $pd$ be two new adjacencies. Observe that $a,d$ lie to the left of $l_{i}$, and $p,b$ lies to the right. If the substitution $(a,b,d)\rightarrow(a,p,d)$ results in a T-path of $\tcomp(l_{i}, \setp)$, we are done, if not, then the triangle $\triangle bap$, or the triangle $\triangle pdb$ is not empty, probably even both. Let us assume without loss of generality that the former is the one that is not empty, and that this is the only one. If both triangles contain points of $\setp$ we can proceed in the same way on both of them. Call this non-empty triangle $\triangle^{\prime}$, and observe that there is at least one point $c^{\prime}\in \setp$ contained in $\triangle^{\prime}$. Choose it and create the adjacencies $bc^{\prime}, c^{\prime}p$. Now do the substitution $(a,b,d)\rightarrow (a,b,c^{\prime},p,d)$, and again test if the new path is an element of $\tcomp(l_{i},\setp)$. If yes, we are done, if not, set $\triangle^{\prime} = bc^{\prime}p$, and thus, there must be again some point of $\setp$ inside $\triangle^{\prime}$. Choose one of those points, label it with $c^{\prime}$, and repeat. Observe that every new point we take lies to the left of $l_{i}$. Since $\setp$ is finite, we will eventually arrive at $\triangle^{\prime}$ being empty, and at that point, we would have created an element of $\tcomp(l_{i},\setp)$, see Figure~\ref{c-tri:figs:12:b}.

\begin{figure}[!htb]
	\begin{center}
		\begin{minipage}[b][5.7cm][t]{7cm}
			\begin{center}
				\includegraphics[height=4cm]{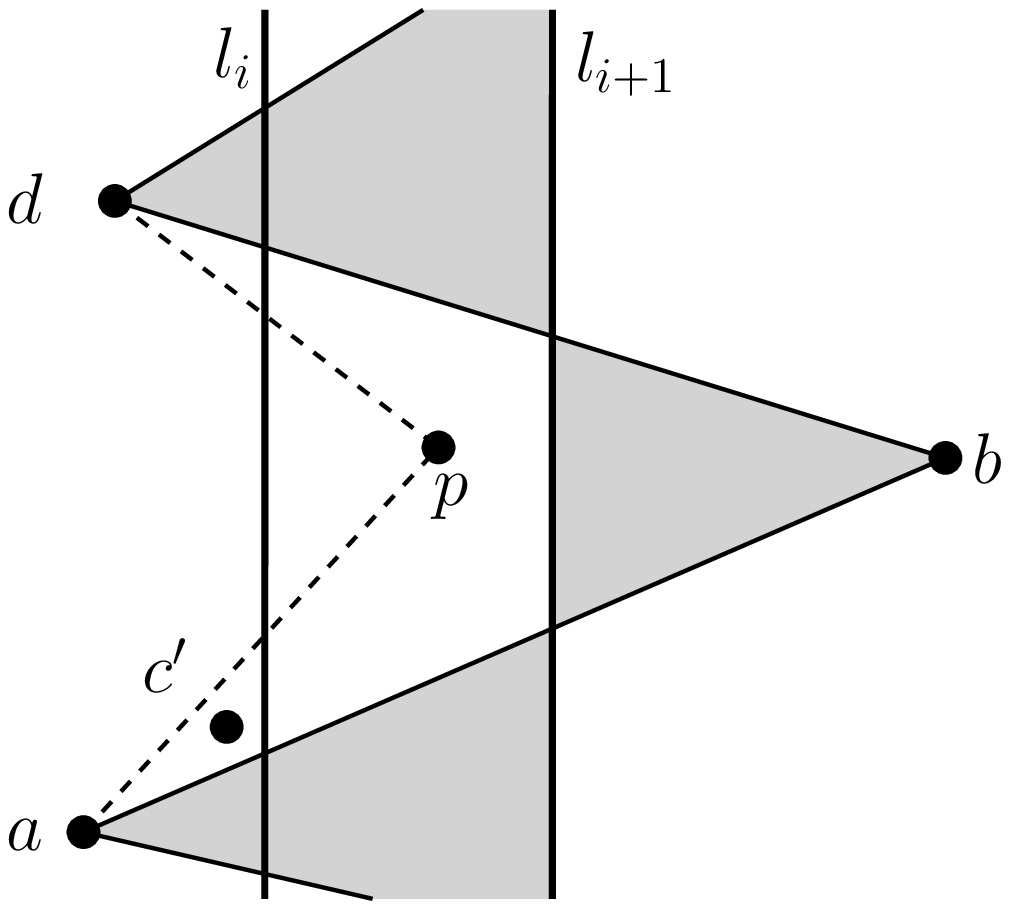}
				\caption{T-path $\pi^{\prime}$ is shown in solid.}
				\label{c-tri:figs:12:a}
			\end{center}
		\end{minipage}
	\quad
		\begin{minipage}[b][5.7cm][t]{7cm}
			\begin{center}
				\includegraphics[height=4cm]{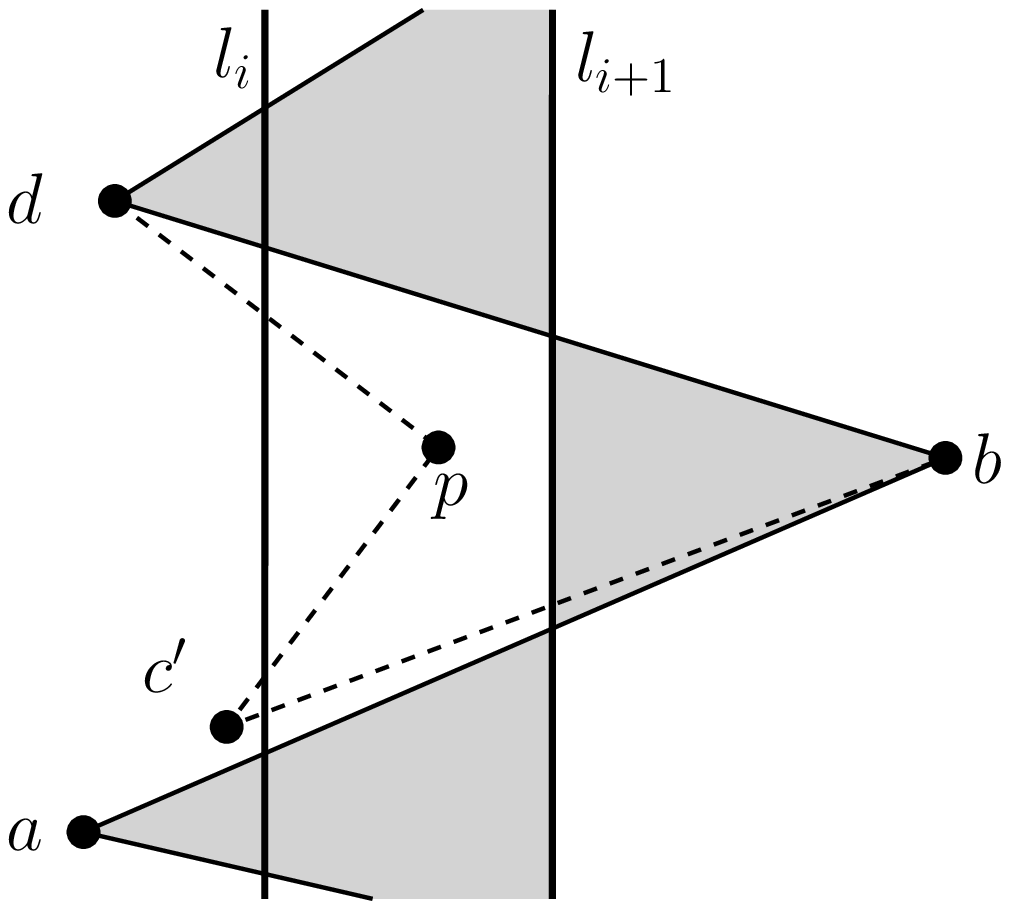}
				\caption{T-path $\pi^{\prime}$ being extended to a T-path $\pi\in\tcomp(l_{i}, \setp)$.}
				\label{c-tri:figs:12:b}
			\end{center}
		\end{minipage}
	\end{center}
\end{figure}

For the second part, by the way the local changes are made, it is clear that from a T-path $\pi\in\tcomp(l_{i}, \setp)$ we cannot obtain more than $O(n^{2})$ T-paths of $\tcomp(l_{i+1}, \setp)$, since when trying local changes of $\pi$ around $p$, at most every pair of points of $\setp$ will be tested, and thus every such a pair can produce at most one T-path of $\tcomp(l_{i+1}, \setp)$. We now have to prove that if $\pi\not\leftrightarrow\pi^{\prime}$ then $\pi$ and $\pi^{\prime}$ cross. Remember that we are still assuming that $p$ is not a vertex of $\pi^{\prime}$, thus $p$ is still inside triangle $\triangle abd$, where $a,b,d$ are three consecutive vertices of $\pi^{\prime}$. Let us assume for the sake of contradiction that $\pi\not\leftrightarrow\pi^{\prime}$, but $\pi\in\lambda(\pi^{\prime})$, \emph{i.e.}, those two paths are non-crossing. Since $\pi\in\lambda(\pi^{\prime})$ there must be at least one triangulation of $\setp$ containing both T-paths. Let $T$ be one of those triangulations, and observe that in $T$, vertex $p$ must have at least two adjacencies to the left of $l_{i}$, since the degree of $p$ in $\pi^{\prime}$ is zero. Among all these adjacencies keep just the first and the last in the radial order around $p$ in clockwise order. Let $b^{\prime},c^{\prime}$ be those two neighbors of $p$ respectively, see Figure~\ref{c-tri:figs:13}. Clearly $b^{\prime}$ and $c^{\prime}$ must be adjacent to $b$, but then the substitution $(a,b,d)\rightarrow(a,b,b^{\prime},p,c^{\prime},b,d)$ creates a T-path $\pi^{\prime\prime}\in\tcomp(l_{i},\setp)$, that is, $\pi^{\prime\prime}\leftrightarrow\pi^{\prime}$, and thus we have that $\pi\neq\pi^{\prime\prime}$ since $\pi\not\leftrightarrow\pi^{\prime}$. But $\pi^{\prime\prime}$ is also a T-path of $T$ w.r.t.~$l_{i}$, which is a contradiction since the T-path of a triangulation w.r.t.~a given line is \emph{unique}, hence such $\pi^{\prime}$ cannot exist.
 
\begin{figure}[!htb]
	\begin{center}
		\begin{minipage}[b][5.7cm][t]{7cm}
			\begin{center}
				\includegraphics[height=4cm]{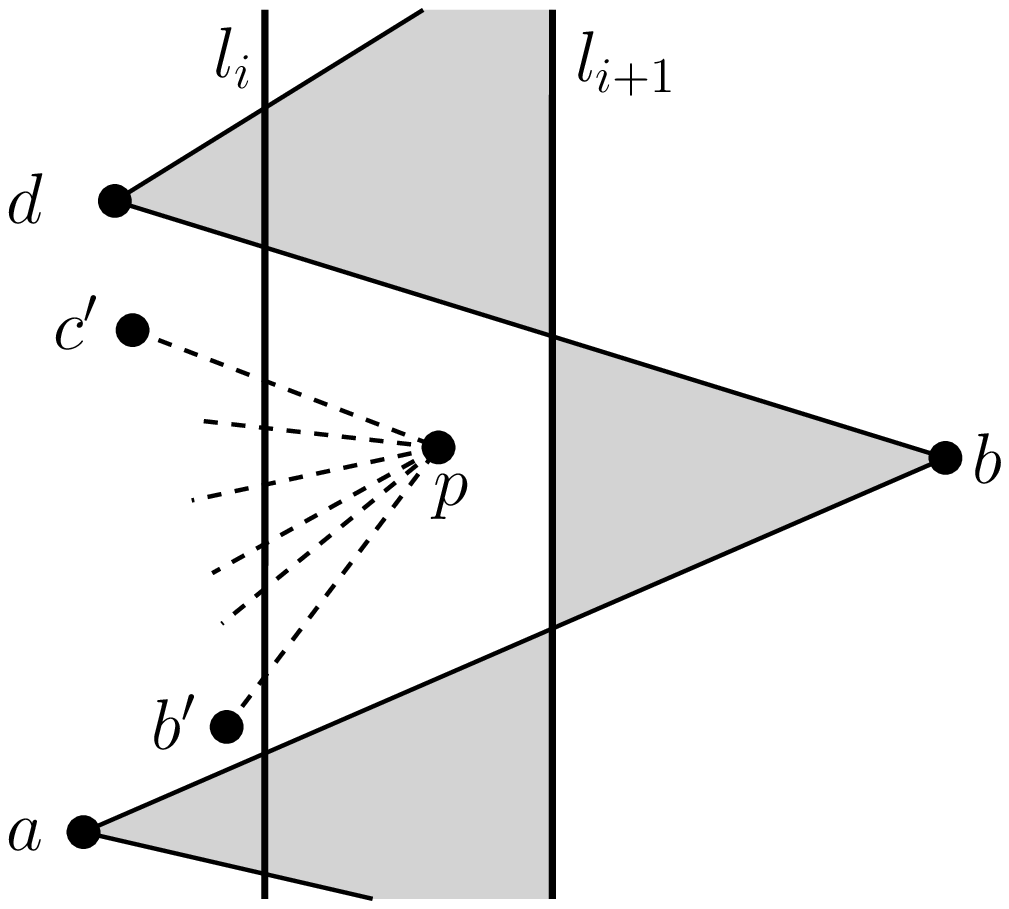}
				\caption{In any triangulation of $\setp$ containing $\pi^{\prime}$, vertex $p$ must have at least two adjacencies to the left of $l_{i}$.}
				\label{c-tri:figs:13}
			\end{center}
		\end{minipage}
	\quad
		\begin{minipage}[b][5.7cm][t]{7cm}
			\begin{center}
				\includegraphics[height=4cm]{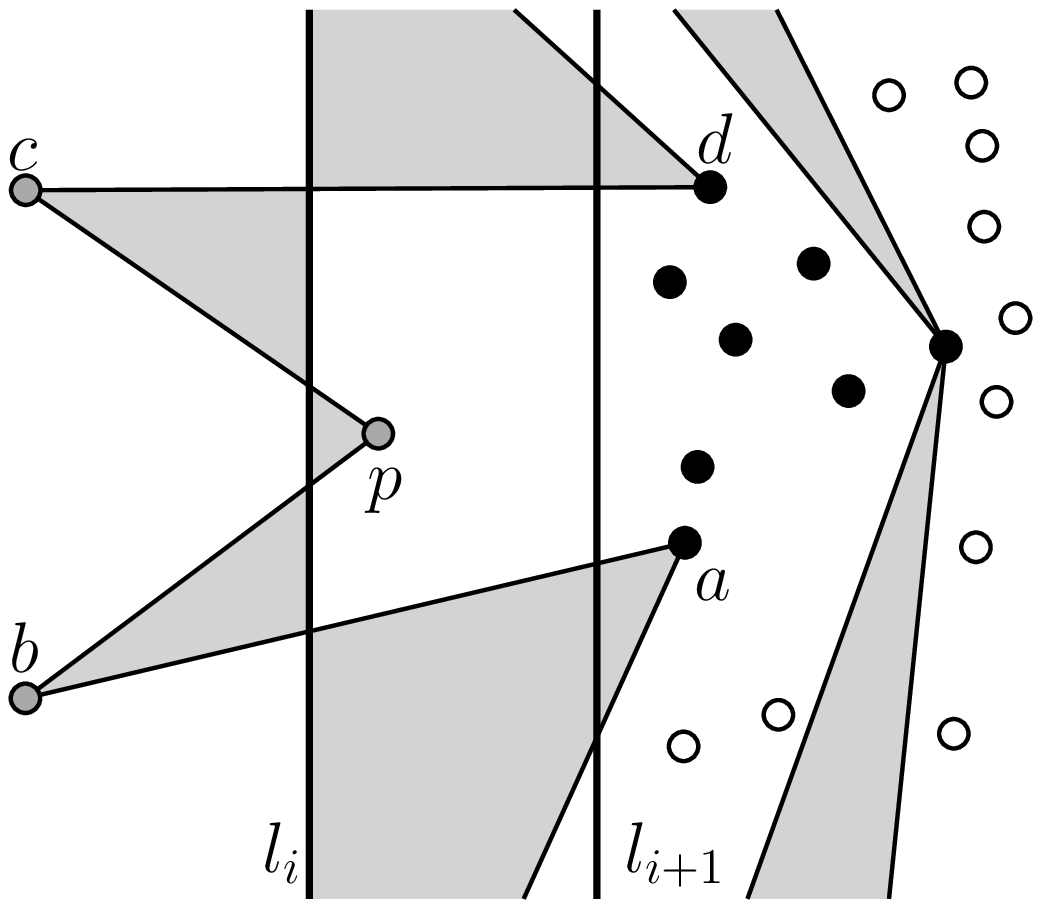}
				\caption{All possibilities for $b^{\prime},c^{\prime}$ are shown as black points. The white points are visible from neither $b$ nor $c$.}
				\label{c-tri:sections:t-paths:figs:17}
			\end{center}
		\end{minipage}
	\end{center}
\end{figure}

It remains to prove that $\lambda(\pi^{\prime})$ can be computed in time $O(n^{2}\cdot t_{i})$ for each $\pi^{\prime}\in\tcomp(l_{i+1}, \setp)$, where $t_{i} = |\tcomp(l_{i}, \setp)|$. From the discussion above we obtain that $\pi\in\lambda(\pi^{\prime})$ if and only if $\pi\leftrightarrow\pi^{\prime}$. The relation $\pi\leftrightarrow\pi^{\prime}$ is obtained by guessing pairs of points $b^{\prime},c^{\prime}$, and checking if the new adjacencies, attached to $\pi$, produce $\pi^{\prime}$. For example, let us assume we want to obtain the possible substitutions for the pattern $(a,b,p,c,d)$, with $p = p_{i+1}$, like in Figure~\ref{c-tri:sections:t-paths:figs:17}. We just have to look for $b^{\prime},c^{\prime}$ among all the points of $\setp$ that are visible from $b$ or $c$, having the edges of $\pi$ as obstacles, see Figure~\ref{c-tri:sections:t-paths:figs:17}. All these points can be obtained in $O\left(n^{2}\right)$ time, since the number of edges of $\pi$ is $O(n)$. Once we obtain this list of candidates, one list $B$ for $b$ and another list $C$ for $c$, we try every possible pair $b^{\prime}, c^{\prime}$ such that $b^{\prime}\in B$, and $c^{\prime}\in C$, for adjacencies that would create $\pi^{\prime}$, for example, we could try adjacencies $bb^{\prime}, b^{\prime}p, pc^{\prime}, c^{\prime}c$ to obtain the substitution $(a,b,p,c,d)\rightarrow (a,b,b^{\prime},p,c^{\prime},c,d)$, but if $c^{\prime} = d$ occurs, then we would have to try substitution $(a,b,p,c,d)\rightarrow (a,b,b^{\prime},p,d)$, and so on depending on the particular configuration. If we pre-process $\setp$ in such a way that we can answer in constant time if a given triangle with vertices in $\setp$ is empty or not, we can also test the correctness of the adjacencies in constant time per pair $b^{\prime},c^{\prime}$. Thus we spend overall $O\left(n^{2}\right)$ time per path $\pi$ of $\tcomp(l_{i},\setp)$. If we have that $\pi\leftrightarrow\pi^{\prime}$, then we also have that $\pi\in\lambda(\pi^{\prime})$, and thus after $O\left(n^{2}\cdot t_{i}\right)$ we have constructed $\lambda(\pi^{\prime})$, for every $\pi^{\prime}\in\tcomp(l_{i+1}, \setp)$, where $t_{i} = |\tcomp(l_{i},\setp)|$. This completes the proof.
\end{proof}

The above discussion implies the algorithmic part of Theorem~\ref{c-tri:theorems:our-t-paths}. The next subsection addresses the second part of the same theorem, \emph{i.e.}, a rough upper bound, depending only on $n$, for the running time of the algorithm just presented will be given.

\subsection{On the number of triangulation paths}\label{c-tri:sections:t-paths:sub-sections:2}

It is known that if $\setp$ is in convex position, then the largest number of T-paths that we can find w.r.t.~some line is $O(2^{n})$, see~\cite{DBLP:conf/compgeom/Aichholzer99}. However, there could be configurations for which this number is much larger. In~\cite{DBLP:journals/comgeo/DumitrescuGPW01} a set $\setp$ is shown for which we can find $\Omega(4^{n - \Theta(\log(n))})$ T-paths w.r.t.~to some line. This number is essentially $4^{n}$, thus we can see that the number of T-paths that one needs to consider is also large. Up to now there have been no results about the largest number of T-paths, over all sets of $n$ points on the plane, and over all possible lines we can define T-paths on. The main result presented here is the following:

\begin{theorem}\label{c-tri:theorems:num-t-paths}
The largest number of T-paths, w.r.t.~a line, of a set of $n$ points $\setp$ on the plane is at most $O(9^{n})$.
\end{theorem}

Before the actual proof, let us first explain how we are going to count T-paths. Let $\setp$ be a set of $n$ points whose elements are labeled with the integers from $1$ to $n$, and let $\pi$ be a T-path of $\setp$ w.r.t.~some given line $l$. Without loss of generality assume that $\pi$ starts at the edge of $\Conv(\setp)$ with the lowest intersection with $l$, and thus it ends at the edge of $\Conv(\setp)$ with the highest intersection with $l$. Observe that given $l$, the starting and ending edges of any T-path w.r.t.~$l$ are always the same two edges of $\Conv(\setp)$. Without loss of generality we will assume that $\pi$ starts to the left of $l$, unless it is otherwise explicitly stated. If $\pi$ starts to the right of $l$ then we would have a symmetric conversation.  

Now orient the edges of $\pi$ as traversing it from the starting edge to the ending edge. The starting edge, by assumption, crosses $l$ from left to right, the second from right to left, the third from left to right again, and so on until we arrive at the ending edge. Observe that the edges of $\pi$ appear sorted bottom-up on $l$ as they intersect $l$, so the starting edge has the lowest intersection with $l$, the second edge has the second lowest intersection with $l$, and so on. Thus the starting vertex of $\pi$ and the edges of $\pi$ that cross $l$ from left to right are enough to characterize $\pi$. There is no other way one can complete adjacencies, since in-between two edges $e,e^{\prime\prime}$ crossing $l$ from left to right, there must be an edge $e^{\prime}$ crossing from right to left and interconnecting $e$ and $e^{\prime\prime}$, and vice-versa, see Figure~\ref{c-tri:sections:t-paths:figs:18}. The starting vertex of $\pi$ tells us if the starting edge crosses $l$ from left to right or from right to left. Now let $e = p_{i}p_{j}$ be an edge of $\pi$ that crosses $l$ from left to right. Let us mark the intersection of $e$ and $l$ with the pair $(i,j)$. Doing this for every edge of $\pi$ that crosses from left to right we obtain a sequence $N$ of pairs of integers on $l$, which along with the first vertex of $\pi$ can be considered as the ``signature'' of $\pi$, since we know at each of those intersection points which edge of $\pi$ crosses $l$, and in which direction. There is the particular case when $\pi$ also ends to the left of $\pi$, and thus its last edge crosses $l$ from right to left, and under our labeling scheme, the last vertex of $\pi$ might not appear in any pair of integers on $l$, however, given $l$, the last edge of $\pi$ is fixed, thus there is no confusion as how to complete $\pi$ see Figure~\ref{c-tri:sections:t-paths:figs:18}. Now, observe that the sequence $N$ of pairs of integers along $l$ can be partitioned into the sequence $N^{-}$ of vertices of $\pi$ lying to the left of $l$, and the sequence $N^{+}$ of vertices of $\pi$ lying to the right. Both sequences $N^{-}$ and $N^{+}$ can be seen as sequences of integers that are sorted w.r.t.~the order they appear on $l$ bottom-up. The way we are going to upper-bound the number of T-paths of $\setp$ w.r.t.~$l$ is by upper-bounding the number of different sequences that represent $N^{-}$. The same bound will obviously hold for the number of different sequences that represent $N^{+}$.  The final bound will come out essentially from combining the two bounds obtained.

\begin{figure}[!ht]
	\begin{center}
		\includegraphics[height=4cm]{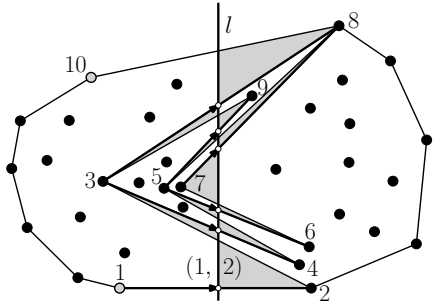}
		\caption{A T-path $\pi$. The first and last vertices are shown in gray. The edges of $\pi$ crossing $l$ from left to right are shown with arrows, and the intersection point is shown as a white dot. The integer sequence $N^{-}$ for $\pi$ is $1,3,5,7,5,3$.}
		\label{c-tri:sections:t-paths:figs:18}
	\end{center}
\end{figure}

\begin{proof}[Proof of Theorem~\ref{c-tri:theorems:num-t-paths}]
	To create a sequence of integers that represent $N^{-}$ we just need the elements of $\setp$ that lie to the left of $l$. Let us denote this subset of points by $\setp^{-}\subset\setp$. Let $\setp_{k}\subseteq\setp^{-}$ be a subset of $\setp^{-}$ of $k$ elements. Imagine that the sequence $N^{-}$ will be obtained using only elements of $\setp_{k}$, but \emph{every} element of $\setp_{k}$ must appear in $N^{-}$ at least once. Let us assume without loss of generality that $1$ is the leftmost point of $\setp_{k}$. Since $1$ must appear in $N^{-}$, it means that there must be at least one straight-line segment $s$ that connects $1$ with $l$, this segment can be thought of as the left part of an edge of a T-path where $1$ appears. Moreover, assume that $s$ is the segment that creates the last entry of $1$ in $N^{-}$, that is, point $1$ is not connected to $l$ at a higher point than the one that $s$ defines. Thus $s$ divides the problem into two sub-problems, since we want to keep everything non-crossing. Let $\setp^{-}_{k}$ be the set of points of $\setp_{k}$ above segment $s$, and let $\setp^{ +}_{k}$ be the set of points of $\setp_{k}$ below $s$ but also including $1$. There are $k$ possibilities for $\setp^{-}_{k}$, since we can rotate $s$ around $1$ clockwise to make the cardinality of $\setp^{-}_{k}$ vary from $0$ to $k-1$, and thus the cardinality of $\setp^{+}_{k}$ varies from $k$ to $1$. Since we are assuming that $s$ is the segment that connects point $1$ for the last time to $l$, then point $1$ does not form part of the sub-problem defined by $\setp^{-}_{k}$, thus this sub-problem is totally independent and we can recurse directly on it. However, point $1$ does play a role in the sub-problem defined by $\setp^{+}_{k}$. If $f(k) = f_{k}$ represents the total number of different possibilities for $N^{-}$ when $k$ points are involved, then we get the following recurrence for $f_{k}$: 
	\begin{align*}
		f_{k} &= g_{k} + \sum_{i = 1}^{k-1} f_{i}\cdot g_{k - i}
	\end{align*}
	where $g_{j}$ represents the sub-problem defined by $\setp^{+}_{k}$, for every $1\leq j\leq k$. Note that for $j = k$ we obtain that $\setp^{-}_{k}$ is empty, and thus $|\setp^{+}_{k}| = k$, which is represented by the term $g_{k}$ of $f_{k}$. Observe that in the case $j = k$, the sub-problem defined by $\setp^{+}_{k}$ is of the same size as the original problem, however, it has a slightly different structure, since in $\setp_{k}^{+}$ we know that point $1$ is already connected to $l$, so the immediate lower connection of $1$ to $l$, if any, cannot be consecutive: This would mean that there are two consecutive edges $e,e^{\prime\prime}$, of some T-path, crossing $l$ from left to right, and sharing vertex $1$ as endpoint, but between $e,e^{\prime\prime}$ there must be exactly one edge $e^{\prime}$ of the same T-path that crosses $l$ from right to left, see Figure~\ref{c-tri:sections:t-paths:figs:19}. If we assume that $e$ intersects $l$ below $e^{\prime\prime}$, then $e^{\prime}$ intersects $l$ in-between, and connects the right endpoint of $e$ with the left endpoint of $e^{\prime\prime}$, thus $e = e^{\prime}$, but in a T-path every edge is used exactly once, hence there cannot be two consecutive appearances of an integer in $N^{-}$. The summation term of $f_{k}$ accounts for the other $k-1$ possibilities for $\setp^{-}_{k}$ and $\setp^{+}_{k}$.
	
\begin{figure}[!ht]
	\begin{center}
		\includegraphics[height=4cm]{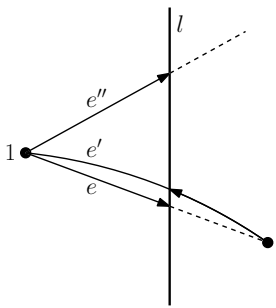}
		\caption{Edges $e,e^{\prime\prime}$ are consecutive edges, of a T-path, that cross $l$ from left to right and share vertex $1$.}
		\label{c-tri:sections:t-paths:figs:19}
	\end{center}
\end{figure}

	The recurrence for $g_{k}$ is very similar; once we enter sub-problem $\setp^{+}_{k}$ we just have to take care of not connecting point $1$ to $l$ consecutively, so we have that: 
	\begin{align*}
		g_{k} &= h_{k} + f_{k-1} + \sum_{i = 1}^{k-1} f_{i}\cdot g_{k-i}
	\end{align*}
	where the term $f_{k-1}$ means that point $1$ is not used in $\setp^{+}_{k}$. If on the other hand, point $1$ is used, then the summation will again account for all the possibilities. The term $h_{k}$ is technical, and its definition is: $h_{k} = 1\Leftrightarrow k = 1$, and $0$ otherwise. With it we can safely define our boundary condition $f_{0} = g_{0} = h_{0} = 0$, and we obtain $f_{1} = g_{1} = h_{1} = 1$, which makes the recursion safe.
	
We are now interested in the asymptotic behavior of $f$. We will obtain it by using ordinary generating functions. We will not explain every single step in detail since we will be using standard techniques. The interested reader is referred to~\cite{generatingfuncs,DBLP:books/daglib/0076724} for the common techniques to obtain generating functions from recurrences.

Introducing the ordinary generating functions $F(x) = \sum_{k=0}^{\infty} f_{k}\cdot x^{k}, G(x) = \sum_{k = 0}^{\infty} g_{k}\cdot x^{k}$, $H(x) = \sum_{k = 0}^{\infty} h_{k}\cdot x^{k} = x$, we obtain for $f_{k}, g_{k}$ the following: 
\begin{align*}
 F(x) &= G(x) + F(x)\cdot G(x)\\
 G(x) &= H(x) + x\cdot F(x) + F(x)\cdot G(x)
\end{align*}
\indent We can now solve this system of equations in unknowns $F(x), G(x)$ to obtain two possible solutions for $F(x)$:
\begin{align*}
	F_{1} &= F(x) = \frac{1 - \sqrt{1 - 8x}}{3 + \sqrt{1 - 8x}}\ \ \ \ {\color{black} \textup{and}}\ \ G_{1} = G(x) = \frac{1 - \sqrt{1 - 8x}}{4}\\
	F_{2} &= F(x) = \frac{-1 - \sqrt{1 - 8x}}{\sqrt{1 - 8x} - 3}\ \ {\color{black} \textup{and}}\ \ G_{2} = G(x) = \frac{1 + \sqrt{1 - 8x}}{4}\\
\end{align*}
\indent However, we know that $F(0)$ must be $0$, and this condition is only met by $F_{1}$, so $\frac{1 - \sqrt{1 - 8x}}{3 + \sqrt{1 - 8x}}$ is the generating function of our sequence $f$, \emph{i.e.}, the coefficients of the Taylor expansion of $F_{1}$ around $0$ are precisely the terms $f_{0} = 0, f_{1} = 1, f_{2} = 3, f_{3} = 13, f_{4} = 67, f_{5} = 381, f_{6} = 2307\ldots$, which turned out to be known as sequence A064062 of ``The On-Line Encyclopedia of Integer Sequences'', but with term $f_{0} = 1$, which makes no difference for the asymptotics of $f$, see~\cite{oeis}. The generating function of A064062 is $F_{A} = \frac{1}{1 - x C(2x)}$, where $C(y) = \frac{1 - \sqrt{1 - 4y}}{2y}$ is the generating function of the Catalan numbers, see~\cite{oeis} and references therein. It is now easy to verify that $F_{A} = F_{1} + 1$, since $F_{A}$ and $F_{1}$ differ only at $f_{0} = 1$.

It is known that the $i$-th term $f_{i}$ of $F_{A}$, for sufficiently large $i$, grows roughly as $\frac{8^{i}}{36i\sqrt{\pi\cdot i}} < 8^{i}$, see~\cite{oeis} and Theorem~3 of~\cite{DBLP:journals/jgaa/BonichonGH05}. 

Thus the number of different possibilities for $N^{-}$ that we can obtain from a set of cardinality $k$ is upper-bounded by $8^{k}$. It remains to consider every possible set $\setp_{k}\subseteq\setp^{-}$. If $|\setp^{-}| = a$, then the absolute number $t^{-}$ we are looking for is upper-bounded by $\sum_{i = 0}^{a}\binom{a}{i} 8^{i} = 9^{a}$. The same bound holds for the number $t^{+}$ of different sequences that represent $N^{+}$. If we partition the original set $\setp$ into $\setp^{-}$ of cardinality $a$, and $\setp^{+}$ of cardinality $b$, such that $a + b = n$, then the number of ways we can create T-paths of $\setp$ w.r.t.~$l$ that start to the left of $l$ is upper-bounded by $t^{-}\cdot t^{+} = 9^{a}\cdot 9^{b} = 9^{n}$. The same bound holds for T-paths that start to the right of $l$, thus obtaining $O(9^{n})$ overall possibilities. The theorem follows.
\end{proof}
This concludes the proof of Theorem~\ref{c-tri:theorems:our-t-paths}.

\section{Counting pseudo-triangulations}\label{c-tri:sections:pt-paths}

The main idea behind our algorithm for counting pseudo-triangulations is to mimic with PT-paths what we did with T-paths for counting triangulations. Thus, here we will have equivalent results to the ones we proved in~\S~\ref{c-tri:sections:t-paths}. We will first explain how a PT-path $\pt{l}{S}$ of a pseudo-triangulation $S$, with respect to line $l$, can be constructed, but in order to do so, we need to define some terms first.

Let $l$ be a separating line, and let $S$ be a pseudo-triangulation of $\setp$. Let us denote by $E_{l}$ the set of edges of $S$ that are intersected by $l$. Let $e\in E_{l}$ and denote by $\overline{e}$ and $\underline{e}$ the edges of $E_{l}$ right above and below $e$ respectively. We will say that $e\in E_{l}$ of $S$ is \emph{good}\footnote{Such an edge $e$ is called \emph{signpost} in~\cite{DBLP:conf/wads/AichholzerRSS03}.} w.r.t.~$l$ iff the intersections of the supporting line of $e$ with the supporting lines of $\overline{e}$ and $\underline{e}$ lie on different sides of $l$, or if $e$ is an edge of $\Conv(\setp)$.

Let us now explain how a PT-path $\pt{l}{S}$ of a pseudo-triangulation $S$, and with respect to line $l$, can be constructed. The following method was originally described in~\cite{DBLP:conf/wads/AichholzerRSS03}: Remove from $S$ all edges of $E_{l}$ that are not good. This leaves a plane graph $S^{*}$ of $\setp$. Let $e$ and $e^{\prime}$ be two consecutive good edges w.r.t.~$l$, and connect them using the common face $f$ of $S^{*}$ that they are part of according to the following rule: If the supporting lines of $e$ and $e^{\prime}$ intersect to the left of $l$, then we use the edges of $f$ that lie to the left. Otherwise we use the edges of $f$ that lie to the right of $l$, see Figure~\ref{c-tri:figs:3}.

\begin{figure}[!htb]
	\begin{center}	
		\includegraphics[height=4cm]{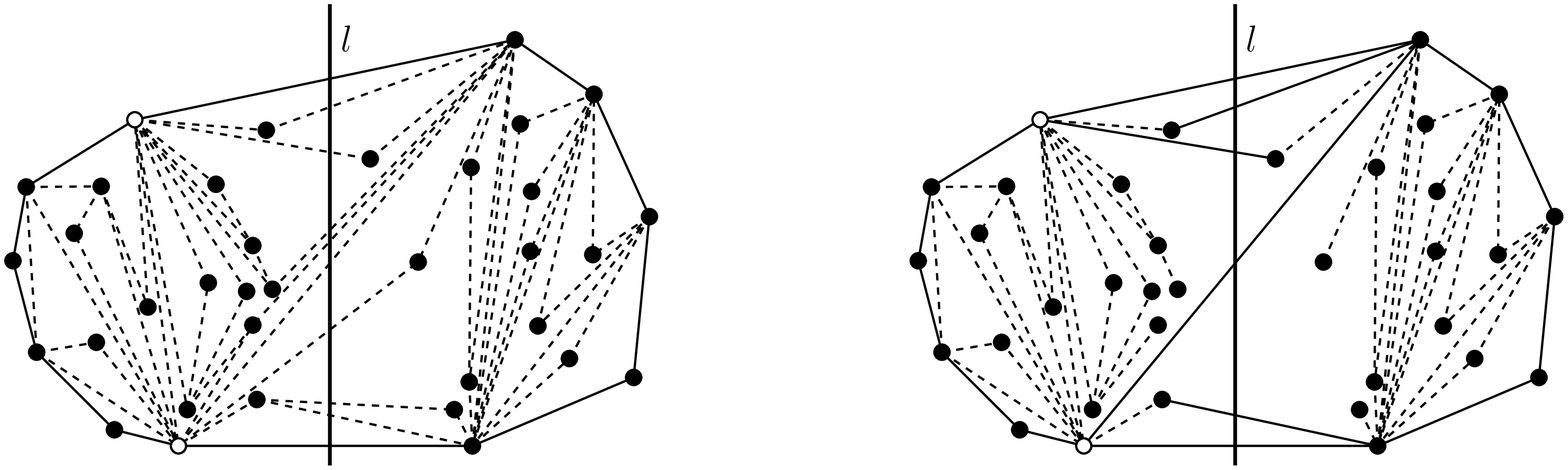}
	\end{center}
	\caption{To the left a pseudo-triangulation $S$. To the right we have the plane graph $S^{*}$ obtained from $S$ by removing all \emph{non-good} edges of $E_{l}$. Joining two consecutive good edges of $E_{l}$ by the rules described before results in the PT-path shown in Figure~\ref{c-tri:figs:1} on page~\pageref{c-tri:figs:1}.}
	\label{c-tri:figs:3}
\end{figure}

Observe that the polygonal chain of edges created by the method described above always exists. In~\cite{DBLP:conf/wads/AichholzerRSS03} it was proven that it fulfills the properties of a PT-path, see Definition~\ref{c-tri:def:pt-paths} on page~\pageref{c-tri:def:pt-paths}. Thus, by Theorem~\ref{c-tri:theorems:pt-pathsOswin}, also on page~\pageref{c-tri:theorems:pt-pathsOswin}, it follows that it is unique.

Let $\L = \{l_{1},\ldots l_{n-1}\}$ be again a set of vertical lines such that point $p_{i}\in\setp$ lies in the vertical slab between $l_{i-1}$ and $l_{i}$, with $2\leq i\leq n-1$. Point $p_{1}$, the leftmost, lies in the unbounded vertical slab to the left of $l_{1}$, and $p_{n}$, the rightmost, lies in the unbounded vertical slab to the right of $l_{n-1}$. For a pseudo-triangulation $S$ of $\setp$ let $\P(S) = \{\pt{l_i}{S}\ |\ l_{i}\in\L\}$. The following result is the equivalent of Theorem~\ref{theorems:1} on page~\pageref{theorems:1} for T-paths and triangulations:

\begin{theorem}\label{c-tri:theorems:pt-paths:2}
	Let $S$ be a pseudo-triangulation with vertex set $\setp$. Then $\P(S)$ is enough to characterize $S$.
\end{theorem}

\begin{proof}
	We will prove something stronger, namely, we will prove that \emph{every} edge of a pseudo-triangulation $S$ is an edge of some PT-path in $\P(S)$, this clearly implies the theorem. Observe that to prove the stronger statement we just have to prove that for any given edge $e$ of $S$ there exists a line $l$ in $\L$ such that $e$ is good w.r.t.~$l$, or if there is no line of $\L$ that $e$ is good with respect to, then we have to show that there is a line $l$ of $\L$ such that $e$ is used to connect two consecutive good edges of $S$ w.r.t.~$l$, that is, $e$ is an edge of the common face of $S^{*}$ that those two consecutive good edges of $S$ w.r.t.~$l$ are part of. By a suitable rotation of the plane we will assume w.l.o.g.~that \emph{every} conceivable vertical line contains at most one point of $\setp$.
	
	Let $e$ be an edge of $S$. If $e$ is an edge of $\Conv(\setp)$ then there is clearly at least one line $l\in\L$ that intersects $e$, and thus it makes $e$ the very first or the very last edge of $\pt{l}{S}$. Now assume that $e$ lies strictly in the interior of $\Conv(\setp)$ and let $\overline{\triangle}, \underline{\triangle}$ be the two pseudo-triangles that $e$ is part of. By convention we will assume that a vertical line intersecting $e$ intersects $\overline{\triangle}$ immediately \emph{above} $e$, and intersects $\underline{\triangle}$ immediately \emph{below} $e$. 
	
	In pseudo-triangulations, as in triangulations, the notion of flipping an edge exists: This time a flip exchanges the diagonal of a pseudo-quadrilateral by its other diagonal, however, for pseudo-quadrilaterals it is not always true that both its diagonals intersect, see Figures~\ref{c-tri:figs:9:a} and~\ref{c-tri:figs:10:a}, while for triangulations that is always the case. Thus, both diagonals could appear in the same \emph{non-pointed} pseudo-triangulation, nevertheless, in a pseudo-triangulation only one of them appears at a time, since the presence of both destroys either planarity or pointedness. We will thus inspect two cases, depending on whether the dual edge $e^{\prime}$ of $e$ in the pseudo-quadrilateral $\square=\overline{\triangle}\cup\underline{\triangle}$ intersects $e$ or not.

	If $e$ and $e^{\prime}$ intersect, let $l$ be the vertical line containing their intersection point, see Figures~\ref{c-tri:figs:9:a} and~\ref{c-tri:figs:9:b}. The reader can easily verify that the supporting lines of the edges $\overline{e}$ of $\overline{\triangle}$ and $\underline{e}$ of $\underline{\triangle}$, intersected by $l$ right above and below $e$, intersect the supporting line of $e$ on different sides of $l$, making $e$ good w.r.t.~$l$. It remains to argue what happens if $l\not\in\L$, which can easily be the case. If $l\not\in\L$ then $l$ lies in the vertical slab between a pair of lines $l_{i-1}, l_{i}\in\L$, and $p_{i}$ is the only point of $\setp$ that also lies in that slab. Thus we can continuously sweep $l$ in one direction as to make it coincide with either $l_{i-1}$ or $l_{i}$ without destroying any argument.
	\begin{figure}[!htb]
		\begin{center}
			\begin{minipage}[b][7cm][t]{7cm}
				\begin{center}
					\includegraphics[height=4cm]{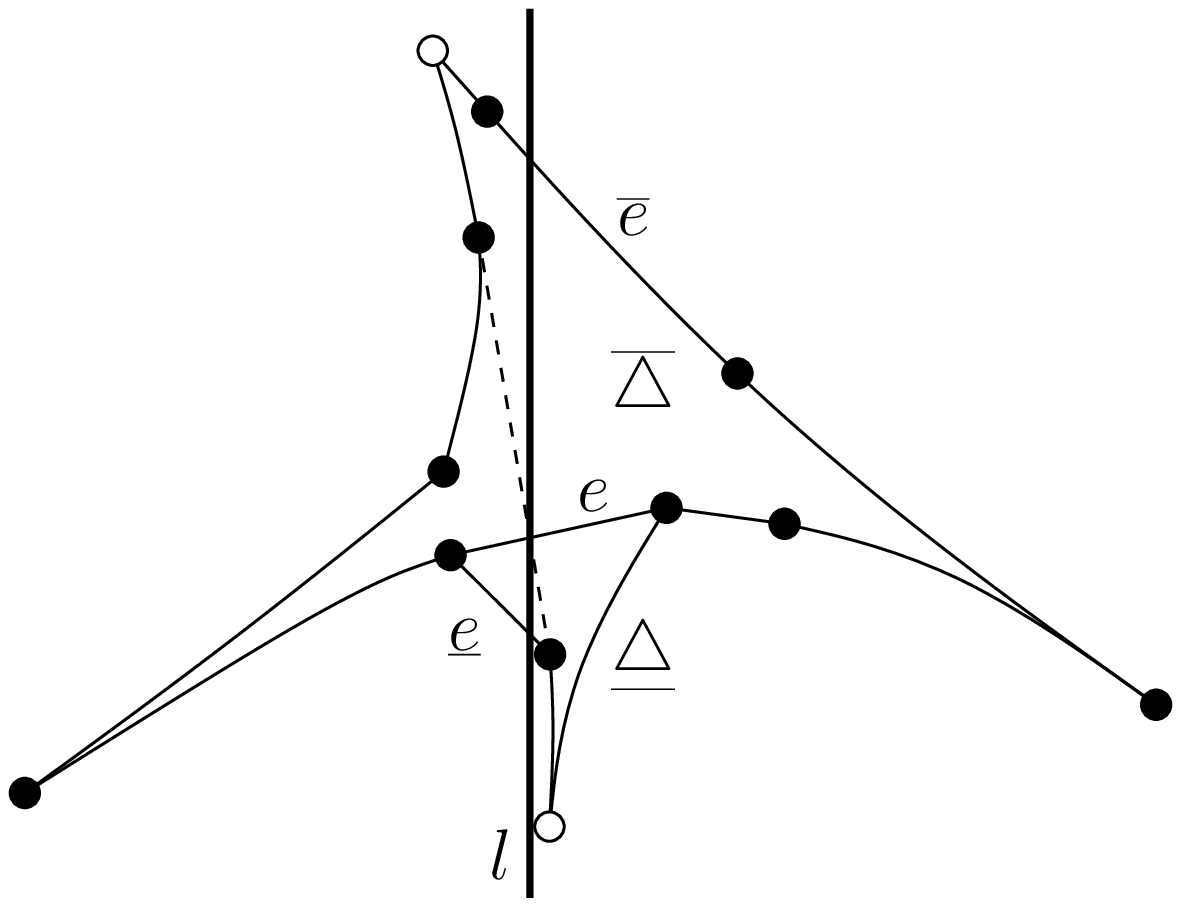}
					\caption{The flip edge $e^{\prime}$ of $e$ is shown dashed. If those two edges intersect, the $e$ is good w.r.t.~line $l$. The two vertices of $\square$ opposite to $e$ are shown in white.}
					\label{c-tri:figs:9:a}
				\end{center}
			\end{minipage}
		\quad
			\begin{minipage}[b][7cm][t]{7cm}
				\begin{center}
				\includegraphics[height=4cm]{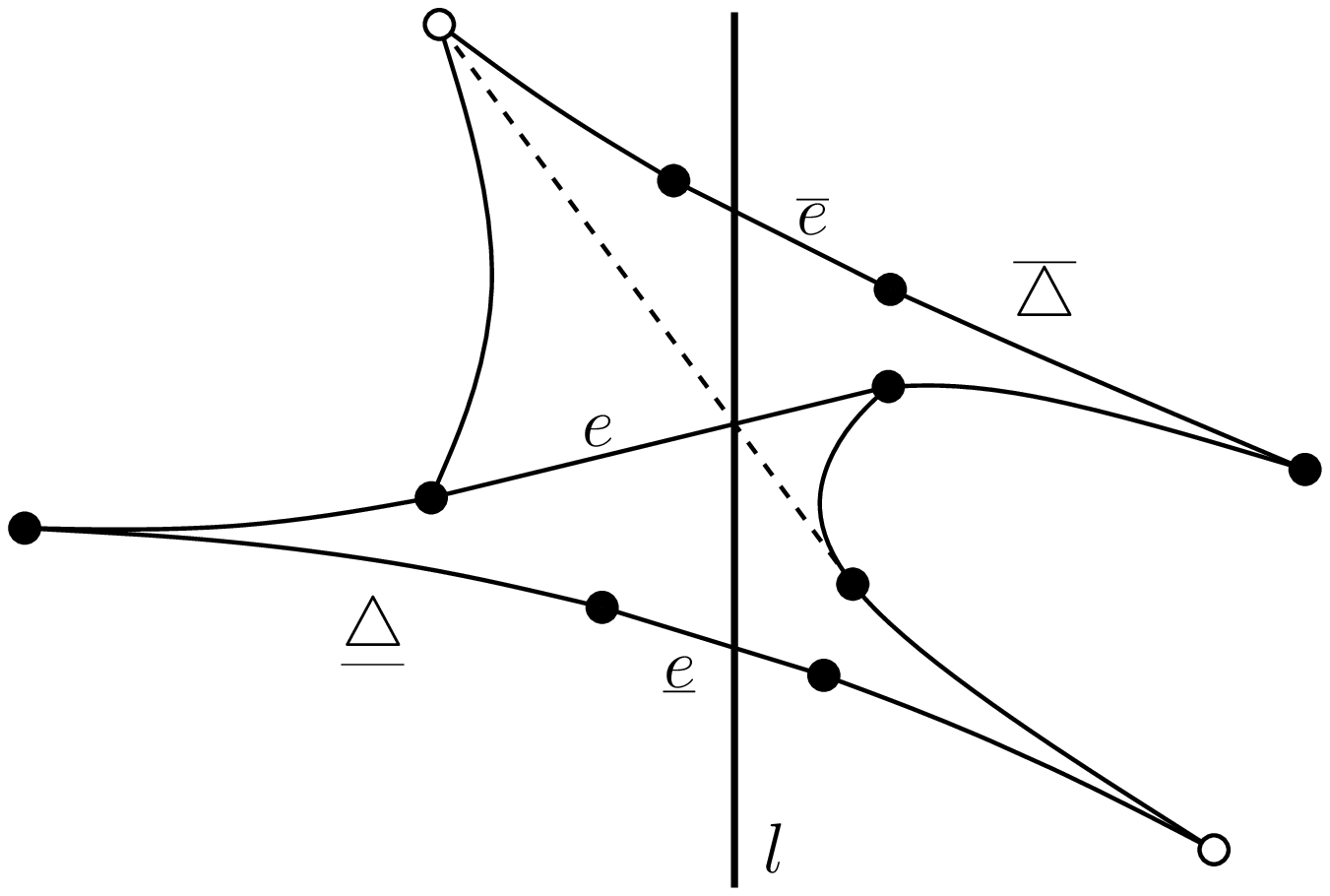}
					\caption{Another possibility for $\square$.}
					\label{c-tri:figs:9:b}
				\end{center}
			\end{minipage}
		\end{center}
	\end{figure}
	
	If $e$ and $e^{\prime}$ do not intersect, let us assume that there is no vertical line $l$ contained in the vertical slab defined by $e$ such that the supporting lines of the edges $\overline{e}, \underline{e}$ intersect the supporting line of $e$ on different sides of $l$, otherwise $e$ is good w.r.t.~to $l$, see Figure~\ref{c-tri:figs:10:a}. We will assume that the intersections between those supporting lines happen to the left of any vertical line that intersects $e$, see Figure~\ref{c-tri:figs:10:b}.

	\begin{figure}[!htb]
		\begin{center}
			\begin{minipage}[b][7cm][t]{7cm}
				\begin{center}
					\includegraphics[height=4cm]{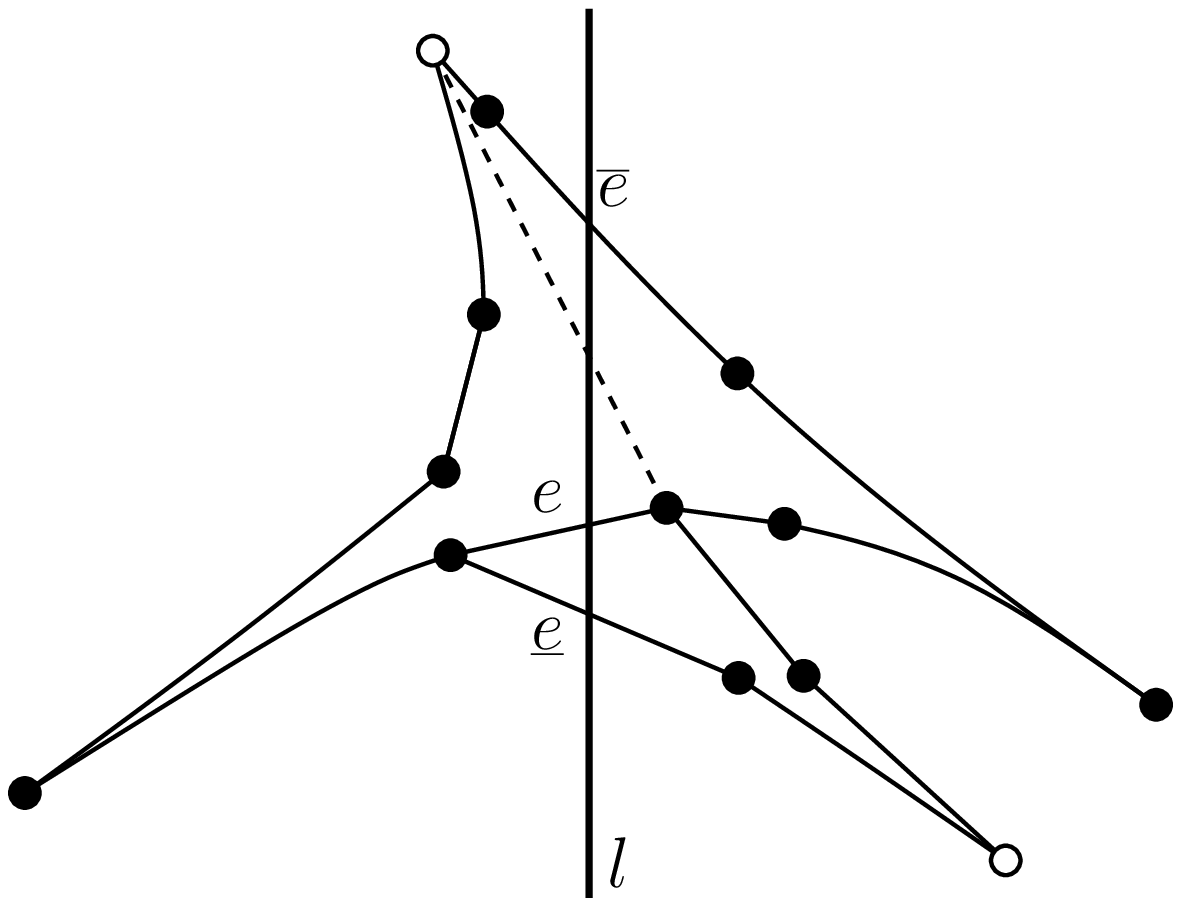}
					\caption{If $e$ and $e^{\prime}$ do not intersect, the pseudo-triangles of $\square$ can be oriented such that there is still a line $l$ that $e$ is good with respect to.}
					\label{c-tri:figs:10:a}
				\end{center}
			\end{minipage}
		\quad
			\begin{minipage}[b][7cm][t]{7cm}
				\begin{center}
				\includegraphics[height=4cm]{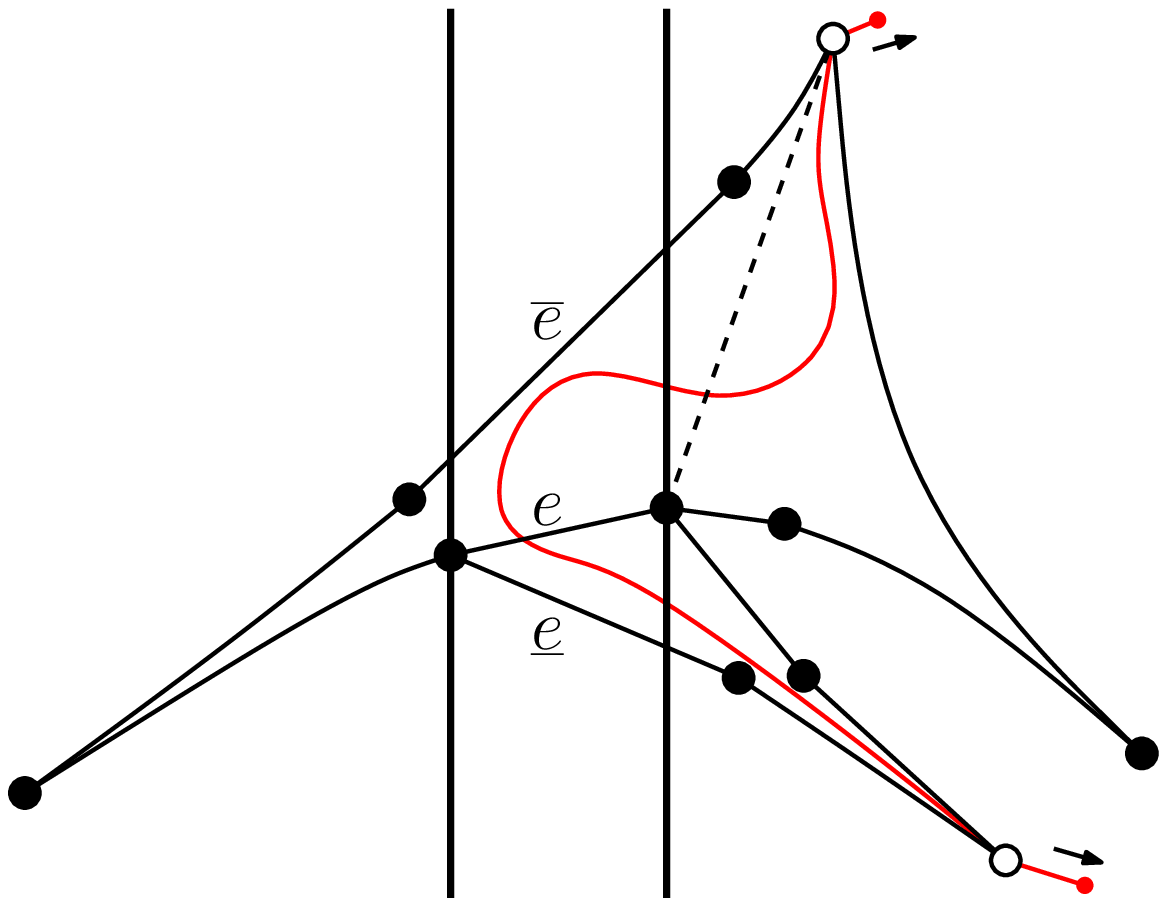}
					\caption{If the red path is pulled from its ends in the direction shown by the arrows, until its length is minimal, we end up having a geodesic path between the opposite vertices, where $e^{\prime}$ is the only new edge.}
					\label{c-tri:figs:10:b}
				\end{center}
			\end{minipage}
		\end{center}
	\end{figure}

	Therefore we have to prove that $e$ is actually used to connect two consecutive good edges of $S$ w.r.t.~some line that does not intersect $e$. Since $e$ and $e^{\prime}$ do not intersect, it must be the case that $e$ and $e^{\prime}$ share one vertex, this is because a flip can be seen as a \emph{geodesic path}\footnote{A geodesic path between two points in a region $R$ is the shortest path between the points that stays in $R$, including its boundary.} between the two corners of $\square$ opposite to $e$. This geodesic path coincides with the boundary of $\square$ except at exactly one edge, which is the flip $e^{\prime}$ of $e$. Since this path does not properly intersect $e$, but connects two points on different sides of the supporting line of $e$, it must happen that one endpoint of $e$ is part of the path, which is exactly the place where $e^{\prime}$ helps to complete the geodesic path, see Figure~\ref{c-tri:figs:10:b}.
	
	Let $p = p_{i}$ be the vertex of $e$ that is also shared by $e^{\prime}$. Note that $p$ is the only point of $\setp$ contained in the vertical slab defined by $l_{i-1}, l_{i}\in\L$. Also, observe that only one of those two lines intersects $e$, so let us assume w.l.o.g.~that $l_{i-1}$ is the one that intersects $e$. The configuration at which we arrive can be seen in Figure~\ref{c-tri:figs:11:a}. Another configuration arises when the other vertex of $e$ is the one shared by $e^{\prime}$; the configuration would be mirror-reflected to the one presented here.
	
	\begin{figure}[!htb]
		\begin{center}
			\begin{minipage}[b][6.2cm][t]{7cm}
				\begin{center}
					\includegraphics[height=4cm]{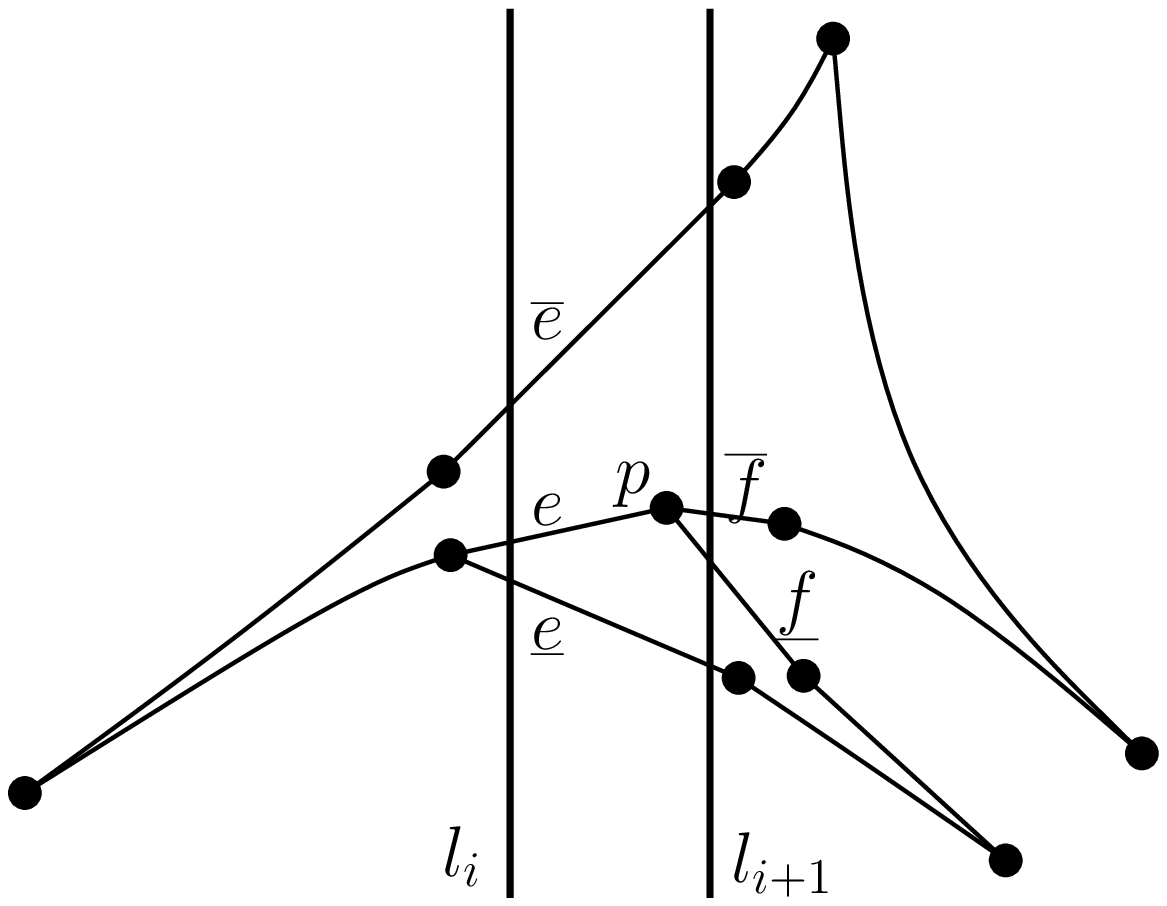}
					\caption{Point $p$ is the only point contained in the vertical slab between $l_{i}, l_{i+1}$. The configuration, if non-degenerate, must locally look like this.}
					\label{c-tri:figs:11:a}
				\end{center}
			\end{minipage}
		\quad
			\begin{minipage}[b][6.2cm][t]{7cm}
				\begin{center}
				\includegraphics[height=4cm]{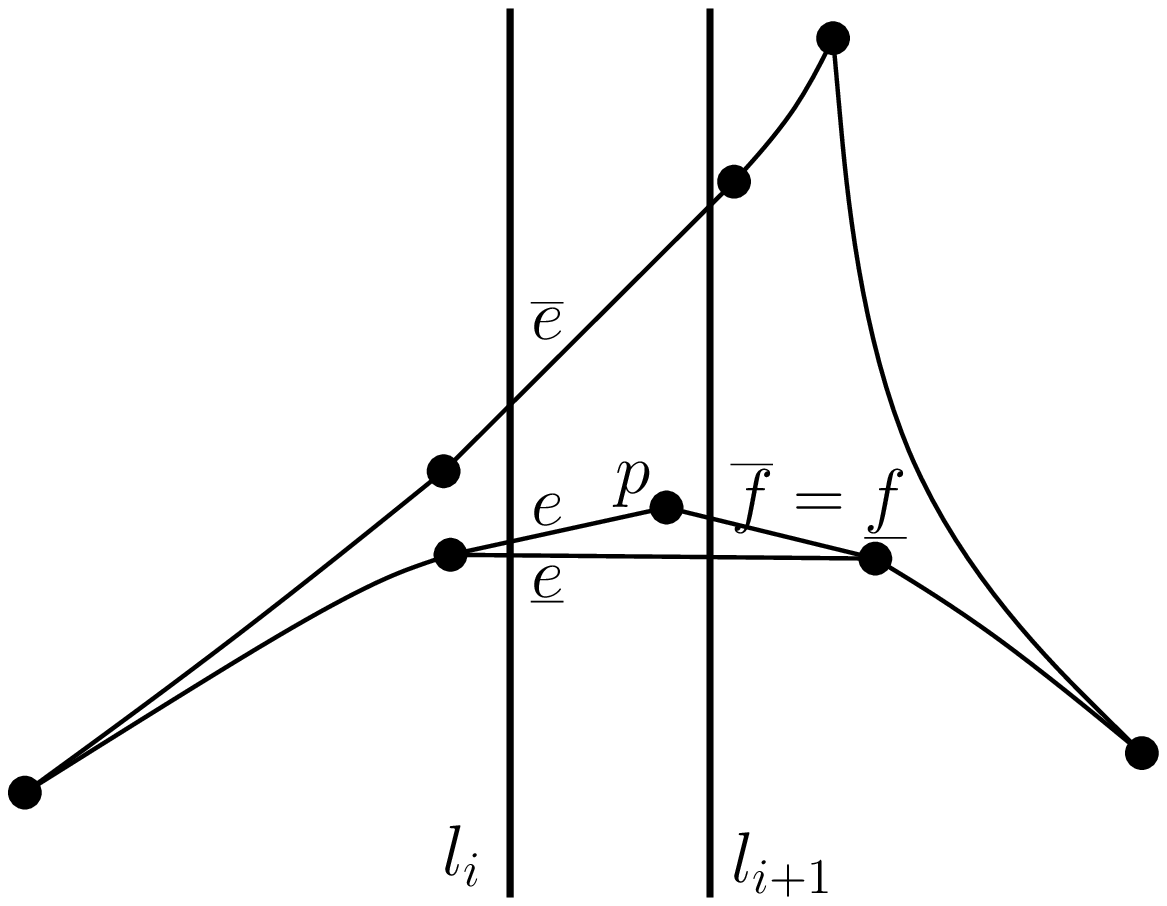}
					\caption{If $\square$ is degenerate, then the configuration looks like this.}
					\label{c-tri:figs:11:b}
				\end{center}
			\end{minipage}
		\end{center}
	\end{figure}
	
	Let $\overline{f}, \underline{f}$ be the other edges of $\overline{\triangle}, \underline{\triangle}$ adjacent to $p$ respectively. We claim that $\underline{f}$ is good w.r.t.~$l_{i+1}$: If $\square$ is non-degenerate, then $\overline{f}\neq\underline{f}$, as displayed in Figure~\ref{c-tri:figs:11:a}. In such a case observe that $\underline{e}, \underline{f}$ and $\overline{f}$ intersect $l_{i+1}$ consecutively, the latter two intersect to the left of $l_{i+1}$, at $p$, and the supporting lines of the former two intersect to the right of $l_{i+1}$, which proofs the claim in this case. If $\square$ is degenerate, as displayed in Figure~\ref{c-tri:figs:11:b}, then $\overline{f} = \underline{f}$, and thus $\overline{e}, \overline{f}$ and $\underline{e}$ intersect $l_{i+1}$ consecutively. Here, the latter two share an endpoint to the right of $l_{i+1}$, and the supporting lines of the former two intersect to the left of $l_{i+1}$, this makes $\overline{f} = \underline{f}$ good again w.r.t.~$l_{i+1}$. At this point observe that regardless of the case, the part of $\overline{\triangle}$ to the right of $l_{i+1}$ cannot be used in $\pt{l_{i+1}, S}$ to connect $\underline{f}$ with the good edge w.r.t.~$l_{i+1}$ that lies above $\underline{f}$, since that part along with $l_{i+1}$ does not form a pseudo-triangle, as the definition of a PT-path requires. Thus the part of $\overline{\triangle}$ to the left of $l_{i+1}$ will be used in $\pt{l_{i+1}}{S}$, but that means that $e$ will also be part of that PT-path, which concludes the proof.
\end{proof}

Hence, as for T-paths, every pseudo-triangulation $S$ of $\setp$ has a \emph{unique} set $\P(S)$. Let $\tcomp(l, \setp) = \{\pt{l}{S}\ |\ S\ {\text{\color{black}is a pseudo-triangulations of}}\ \setp\}$ be the set of \emph{all} PT-paths w.r.t.~to separating line $l$. What is now of interest to us is the opposite. Does every tuple $\{\pi_{1},\ldots, \pi_{n-1}\}$ of pairwise non-crossing PT-paths define a \emph{unique} pseudo-triangulation? Where $\pi_{i}\in\tcomp(l_{i},\setp)$ and $l_{i}\in\L$. The analogous statement for triangulations was clear, however, pseudo-triangulations might require more explanation. The answer is yes, as long as the union $\bigcup_{1\leq i\leq n-1}\pi_{i}$ is pointed. To see this, just observe that if that union is pointed, then it can be completed to a pseudo-triangulation $S_{1}$ by adding edges, while keeping planarity and pointedness, see Theorem~\ref{c-tri:theorems:pointed-pt}. Assume there is another pseudo triangulation $S_{2}$ that can be obtained from the union of the PT-paths $\pi_{i}$ by adding edges in a different way. Observe that every PT-path $\pi_{i}$, $1\leq i\leq n-1$, keeps being a PT-path of $S_{1}, S_{2}$ since the additional edges do not break planarity or pointedness. Thus by Theorem~\ref{c-tri:theorems:pt-paths:2} there is no other option but $\bigcup_{1\leq i\leq n-1}\pi_{i} = \P(S_{1}) = \P(S_{2})$. But in the proof of that theorem we actually showed that \emph{every} edge of $S_{1}, S_{2}$ is in some PT-path in $\P(S_{1}), \P(S_{2})$ respectively, thus $S_{1} = S_{2}$.

Thus the number of pointed pseudo-triangulations of $\setp$ equals the number of different sets $\P(S)$ that we can find on $\setp$. The algorithm for counting pseudo-triangulations is the same as the algorithm for counting triangulations presented in the previous section, so we just have to define the sets the algorithm works on. Also, the proof of correctness will remain essentially the same, we will just point out what the differences are.

By previous discussions we know that a tuple $\{\pi_{1},\ldots, \pi_{n-1}\}$ of PT-paths, where $\pi_{i}\in\tcomp(l_{i}, \setp)$ and $l_{i}\in\L$, defines a pseudo-triangulations iff those PT-paths are pairwise non-crossing and their union is pointed. As before, we will use the term \emph{compatible} for such a pointed and pairwise non-crossing set of PT-paths. We can now define the following set:
\begin{align*}
	\T(\pi_{j}) &= \{\{\pi_{1},\ldots, \pi_{j-1}\}\ |\ \{\pi_{1},\ldots, \pi_{j}\}\ {\text{\color{black}is compatible, and}}\ \pi_{i}\in\tcomp(l_{i}, \setp),l_{i}\in\L\}
\end{align*}

By the discussion above we have that the number of pointed pseudo-triangulations of $\setp$ is exactly $|\T(\pi)|$, where $\pi$ is the \emph{unique} PT-path of $\setp$ w.r.t.~$l_{n-1}\in\L$.

Finally, and for completeness, for each $\pi^{\prime}\in\tcomp(l_{i+1}, \setp)$ and each $\pi\in\tcomp(l_{i},\setp)$, with $l_{i}, l_{i+1}\in\L$, we define: 
\begin{align*}
\lambda(\pi^{\prime}) &= \{\pi\in\tcomp(l_{i},\setp)\ |\ \pi {\textcolor{black}{\text{ is compatible with }}} \pi^{\prime}\}\\
\mu(\pi) &= \{\pi^{\prime}\in\tcomp(l_{i+1}, \setp)\ |\ \pi^{\prime}{\textcolor{black}{\text{ is compatible with }}} \pi\}
\end{align*}

The notation $\tcomp(\cdot, \cdot), \T(\cdot), \lambda(\cdot)$ and $\mu(\cdot)$ is the same as the one used in \S~\ref{c-tri:sections:t-paths} for T-paths, but the definitions here reflect that we are now dealing with PT-paths instead.

Since the sweep line algorithm for counting pseudo-triangulations is the same as the one for counting triangulations, we just have to show how to obtain $\tcomp(l_{i+1}, \setp)$, as well as $|\T(\pi^{\prime})|$ for every $\pi^{\prime}\in\tcomp(l_{i+1}, \setp)$, having stored $\tcomp(l_{i}, \setp)$ and $|\T(\pi)|$ for every $\pi\in\tcomp(l_{i}, \setp)$, where $l_{i}, l_{i+1}\in\L$ are two consecutive \emph{event points} of the sweep line algorithm. This, as for T-paths, will be accomplished by doing \emph{local changes} to every PT-path $\pi\in\tcomp(l_{i}, \setp)$, which we explain next. From this local changes we directly obtain $\tcomp(l_{i+1}, \setp)$ as well as $\lambda(\pi^{\prime})$ for each $\pi^{\prime}\in\tcomp(l_{i+1}, \setp)$. Thus, obtaining $|\T(\pi^{\prime})|$ is easy since $|\T(\pi^{\prime})| = \sum_{\pi\in\lambda(\pi^{\prime})}|\T(\pi)|$. We will later prove that $\lambda(\pi^{\prime})$ can be correctly computed in time $O\left(n^{6}\cdot t_{i}\right)$, where $t_{i} = |\tcomp(l_{i}, \setp)|$. Therefore the overall running time of the algorithm is $\sum_{l_{j}\in\L}O\left(n^{6}\cdot t_{j}\right)\leq O\left(n^{7}\cdot t\right)$, where $t = \max\{t_{j}\}$.

Let us now explain what the local changes in general look like. Let $p = p_{i+1}\in\setp$ be the point lying between lines $l_{i}, l_{i+1}\in\L$. As for T-paths, the only obstacle of \emph{every} PT-path $\pi$ of $\tcomp(l_{i}, \setp)$ to be a PT-path $\pi^{\prime}$ of $\tcomp(l_{i+1}, \setp)$ is $p$. The changes are mostly equivalent (in form) to the ones for T-paths but this time they are more complicated. We have two possibilities, depending on whether $\pi$ has $p$ as a vertex or not. Let us see each one in turn:

\begin{enumerate}[label=(\oldstylenums{\arabic*})]	
	\item If $\pi\in\tcomp(l_{i}, \setp)$ has $p$ as vertex we have more sub-cases depending on whether $p$ lies inside $\Conv(\setp)$ or on $\Conv(\setp)$, and whether $p$ is the convex vertex of an empty pseudo-triangle bounded by $l_{i}$ or not. Let us see:

	\begin{itemize}
		\item If $p$ lies strictly inside $\Conv(\setp)$ let us first assume that $p$ is also the convex vertex of an empty pseudo-triangle of $\pi$ bounded by $l_{i}$. This case is equivalent to the one for triangulations displayed in Figure~\ref{c-tri:sections:t-paths:figs:14} on page~\pageref{c-tri:sections:t-paths:figs:14}. The situation is as displayed in Figures~\ref{c-tri:figs:4:a} and~\ref{c-tri:figs:4:b} with solid lines. Let $e, f$ be the good edges of $\pi$ w.r.t.~$l_{i}$ right below and above $p$ respectively, and let $e^{\prime}, f^{\prime}$ be the good edges of $\pi$ w.r.t.~$l_{i}$ adjacent to $p$ such that $e, e^{\prime}, f^{\prime}, f$ are ordered bottom-up along $l_{i}$. Let $\triangle\ (\triangle^{\prime})$ be the empty pseudo-triangle of $\pi$ to the left of $l_{i}$ having $e, e^{\prime}\ (f, f^{\prime})$ as edges and bounded by $l_{i}$. If $e$ and $f$ share their right endpoint, then a PT-path $\pi^{\prime}\in\tcomp(l_{i+1}, \setp)$ can be produced using only adjacencies from the original PT-path $\pi\in\tcomp(l_{i}, \setp)$, see Figure~\ref{c-tri:figs:4:b}. This situation can easily be detected.
		\begin{figure}[!htb]
			\begin{center}
				\begin{minipage}[b][5.8cm][t]{7cm}
					\begin{center}
						\includegraphics[height=4cm]{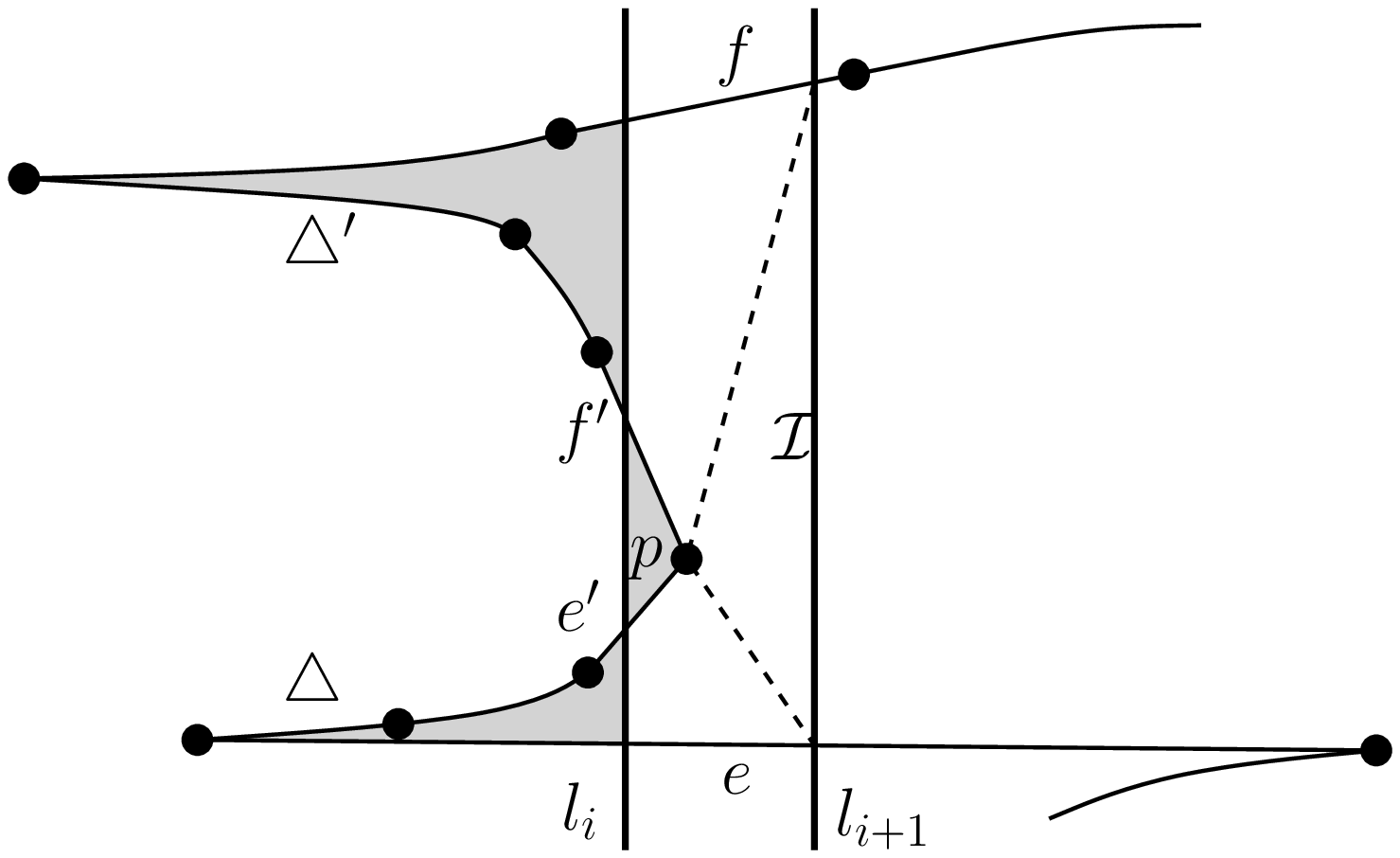}
						\caption{Here $e$ and $f$ do not share the right endpoint.}
						\label{c-tri:figs:4:a}
					\end{center}
				\end{minipage}
			\quad
				\begin{minipage}[b][5.8cm][t]{7cm}
					\begin{center}
						\includegraphics[height=4cm]{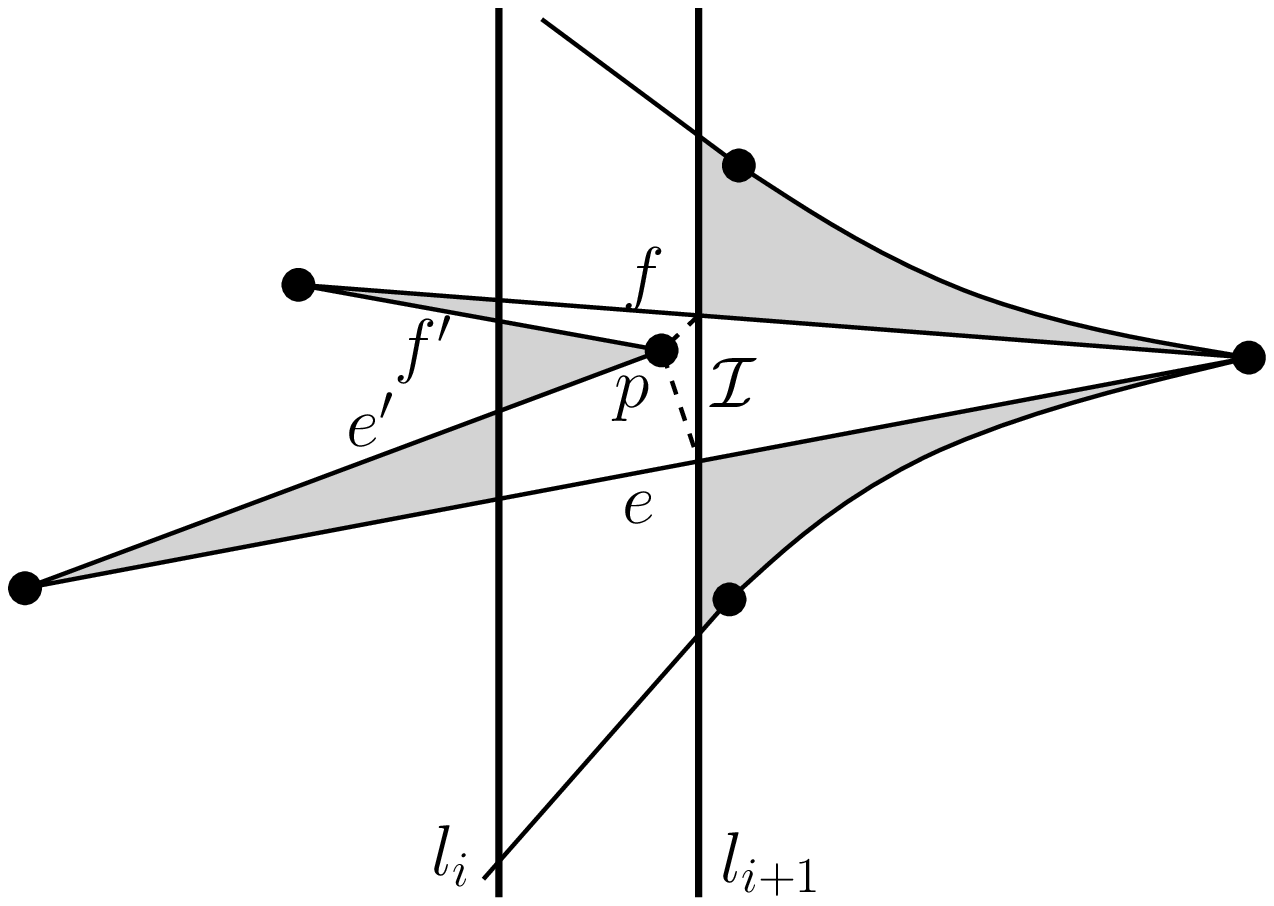}
						\caption{In this case a PT-path $\pi^{\prime}\in\tcomp(l_{i+1}, \setp)$ can be produced using only adjacencies from the original PT-path $\pi\in\tcomp(l_{i}, \setp)$.}
						\label{c-tri:figs:4:b}
					\end{center}
				\end{minipage}
			\end{center}
		\end{figure}
		
		If $e$ and $f$ do not share their right endpoint, then the situation is in general as displayed in Figure~\ref{c-tri:figs:4:a}. The local changes we are looking for are produced by every point $\alpha\in\setp$ such that the dotted adjacencies shown in Figure~\ref{c-tri:figs:4:a:bis} produce a PT-path $\pi^{\prime}\in\tcomp(l_{i+1}, \setp)$ with the property that $\pi\cup\pi^{\prime}$ is pointed. 
		\begin{figure}[!htb]
			\begin{center}
						\includegraphics[height=4cm]{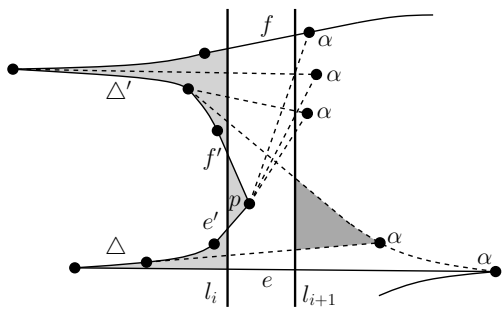}
						\caption{All points $\alpha$ can be used to produce a PT-path $\pi^{\prime}\in\tcomp(l_{i+1}, \setp)$.}
						\label{c-tri:figs:4:a:bis}
			\end{center}
		\end{figure}
		
		So, let us explain more carefully how these changes are really produced. Let $\mathcal{I}$ be the interval of $l_{i+1}$ seen by $p$ having the edges of $\pi$ as obstacles. The visibility cone of $p$ towards $l_{i+1}$ is shown dashed in Figures~\ref{c-tri:figs:4:a} and~\ref{c-tri:figs:4:b}. 	Observe that every $\alpha$ used for a change has a visibility ray to $\mathcal{I}$. So having the edges of $\pi$ as obstacles, obtain a list $A$ of all points to the right of $l_{i+1}$ having a visibility ray to $\mathcal{I}$. This can be done in total time $O\left(n^{2}\log(n)\right)$, see~\cite{ORourke:1987:AGT:40599}. Let $\alpha\in A$. We will assume that we have actually computed a visibility cone to $\mathcal{I}$ with apex at $\alpha$. We then regard $\alpha$ as the apex of an empty pseudo-triangle bounded by $\pi^{\prime}$ (to be constructed) and $l_{i+1}$, see the dark gray region to the right of $l_{i+1}$ with apex at one of the $\alpha$'s in Figure~\ref{c-tri:figs:4:a:bis}. The same $\alpha$ can give rise to different PT-paths of $\tcomp(l_{i+1}, \setp)$, see Figure~\ref{c-tri:figs:14}.
		
		\begin{figure}[!htb]
			\begin{center}
				\begin{minipage}[b][4.2cm][t]{7cm}
					\begin{center}
						\includegraphics[height=4cm]{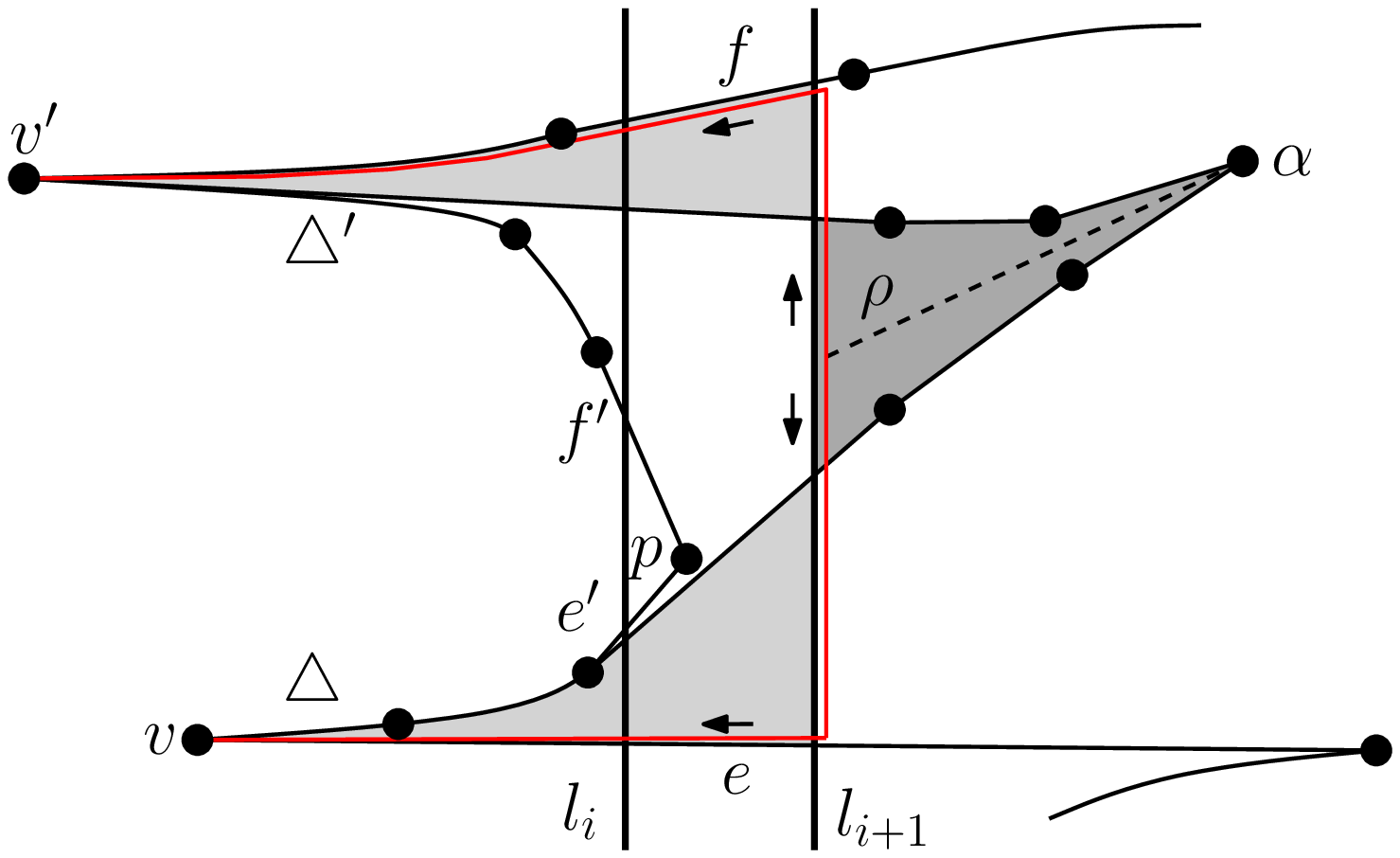}
					\end{center}
				\end{minipage}
				\quad
				\begin{minipage}[b][4.2cm][t]{7cm}
					\begin{center}
						\includegraphics[height=4cm]{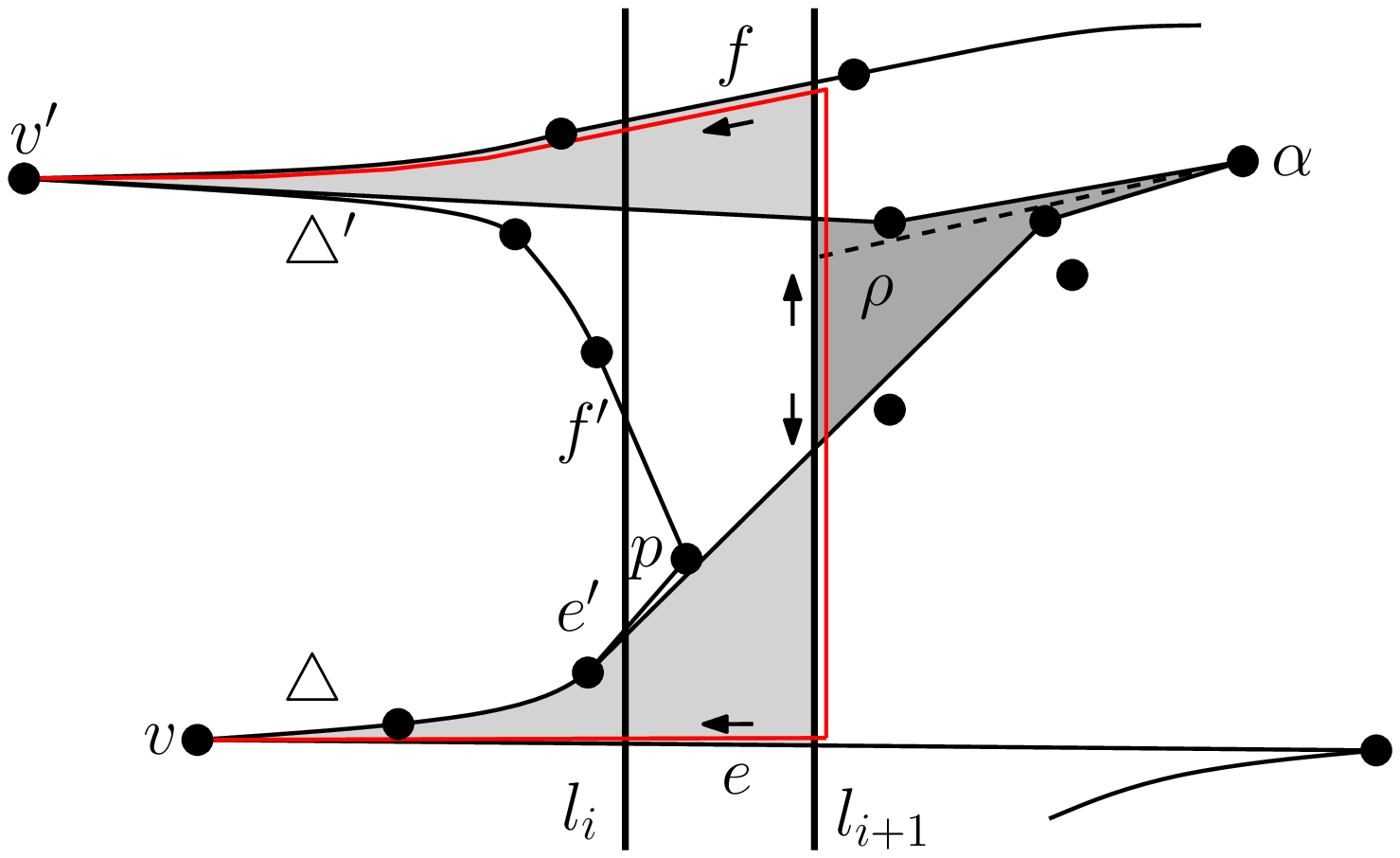}
					\end{center}
				\end{minipage}
			\caption{Two different possibilities for adjacencies connecting $\alpha$ to $\pi\in\tcomp(l_{i}, \setp)$. Each gives a different PT-path of $\tcomp(l_{i+1}, \setp)$.}
			\label{c-tri:figs:14}	
			\end{center}
		\end{figure}
		
		So the way we discern between all the PT-paths of $\tcomp(l_{i+1}, \setp)$ that can be obtained from a single $\alpha\in A$ is as follows: Shoot a visibility ray $\rho$ from $\alpha$ to $\mathcal{I}$ that is fully contained in the empty pseudo-triangle delimited by $l_{i+1}$ that $\alpha$ is apex of, the dashed lines of Figure~\ref{c-tri:figs:14}. From the intersection point between $\rho$ and $\mathcal{I}$ create two paths $\rho_{\downarrow}, \rho_{\uparrow}$ following $\mathcal{I}$ towards $e$ and $f$ respectively, so $\rho_{\downarrow}$ goes down, and $\rho_{\uparrow}$ goes up. Once $e$ and $f$ are reached, follow the adjacencies of $\pi$ towards the leftmost convex vertex $v$, $v^{\prime}$ of $\triangle$, $\triangle^{\prime}$ respectively. Paths $\rho_{\downarrow}, \rho_{\uparrow}$ are shown in red in Figure~\ref{c-tri:figs:14}. Now, the adjacencies that are joining $\pi\in\tcomp(l_{i}, \setp)$ with $\alpha$ are nothing but two shortest paths $\widetilde{\rho_{\downarrow}}, \widetilde{\rho_{\uparrow}}$ between $\alpha$ and $v$, $v^{\prime}$ respectively, the former homotopic to $\rho_{\downarrow}\cup\rho$ and the latter homotopic to $\rho_{\uparrow}\cup\rho$. Just imagine that if $\rho_{\downarrow}\cup\rho$ and $\rho_{\uparrow}\cup\rho$ are two strings between $\alpha$ and $v, v^{\prime}$ respectively, then pulling them as to make them of shortest length, having the points of $\setp$ as obstacles, will give the adjacencies connecting $\alpha$ to $\pi$, and thus complete the adjacencies of $\pi^{\prime}\in\tcomp(l_{i+1}, \setp)$.
		
		\begin{figure}[!htb]
			\begin{center}
				\begin{minipage}[b][5.5cm][t]{7cm}
					\begin{center}
						\includegraphics[height=4cm]{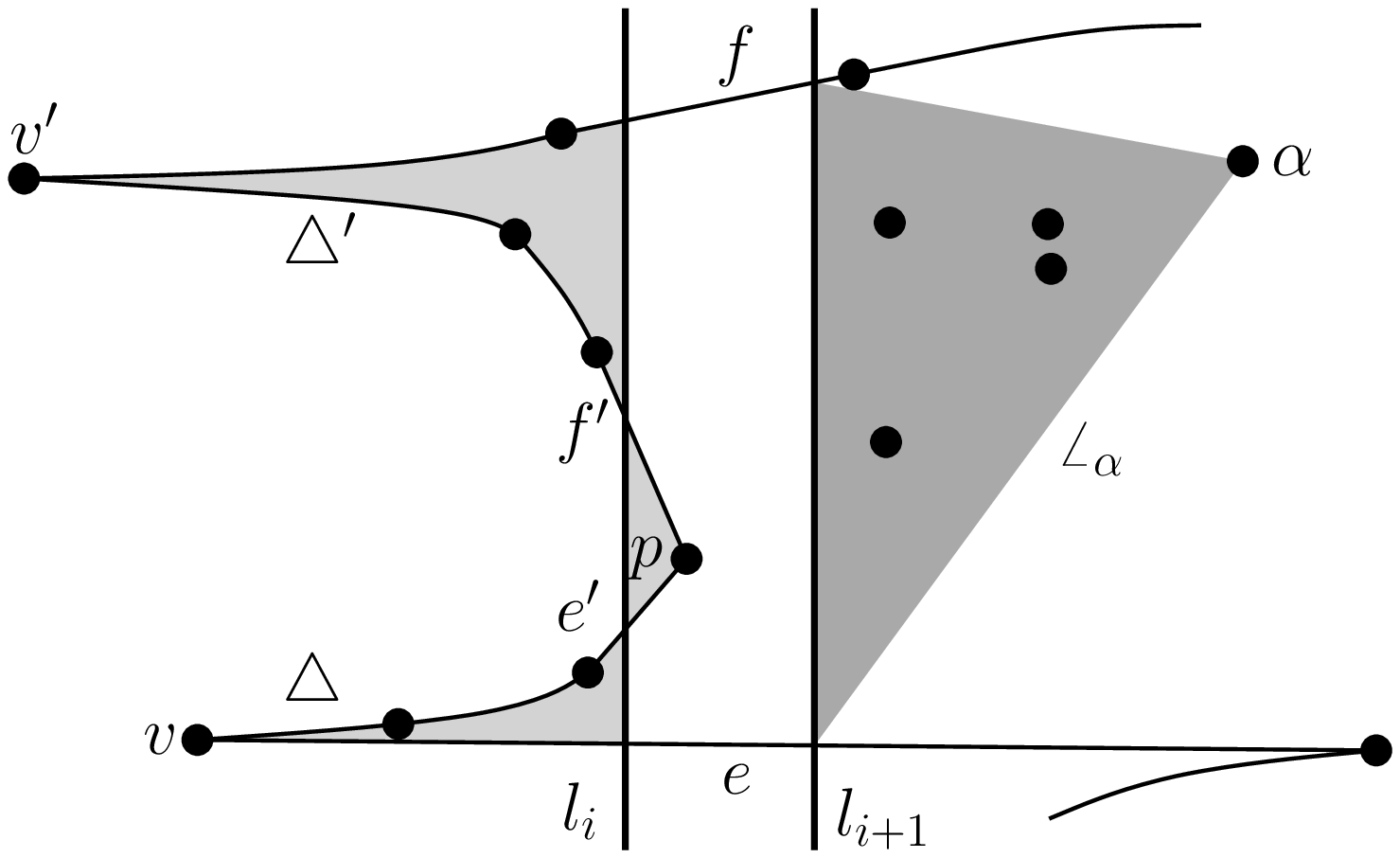}
						\caption{The visibility cone $\angle_{\alpha}$ (to the right of $l_{i+1}$) is shown in dark gray.}
						\label{c-tri:figs:15:a}
					\end{center}
				\end{minipage}
				\quad
				\begin{minipage}[b][5.5cm][t]{7cm}
					\begin{center}
						\includegraphics[height=4cm]{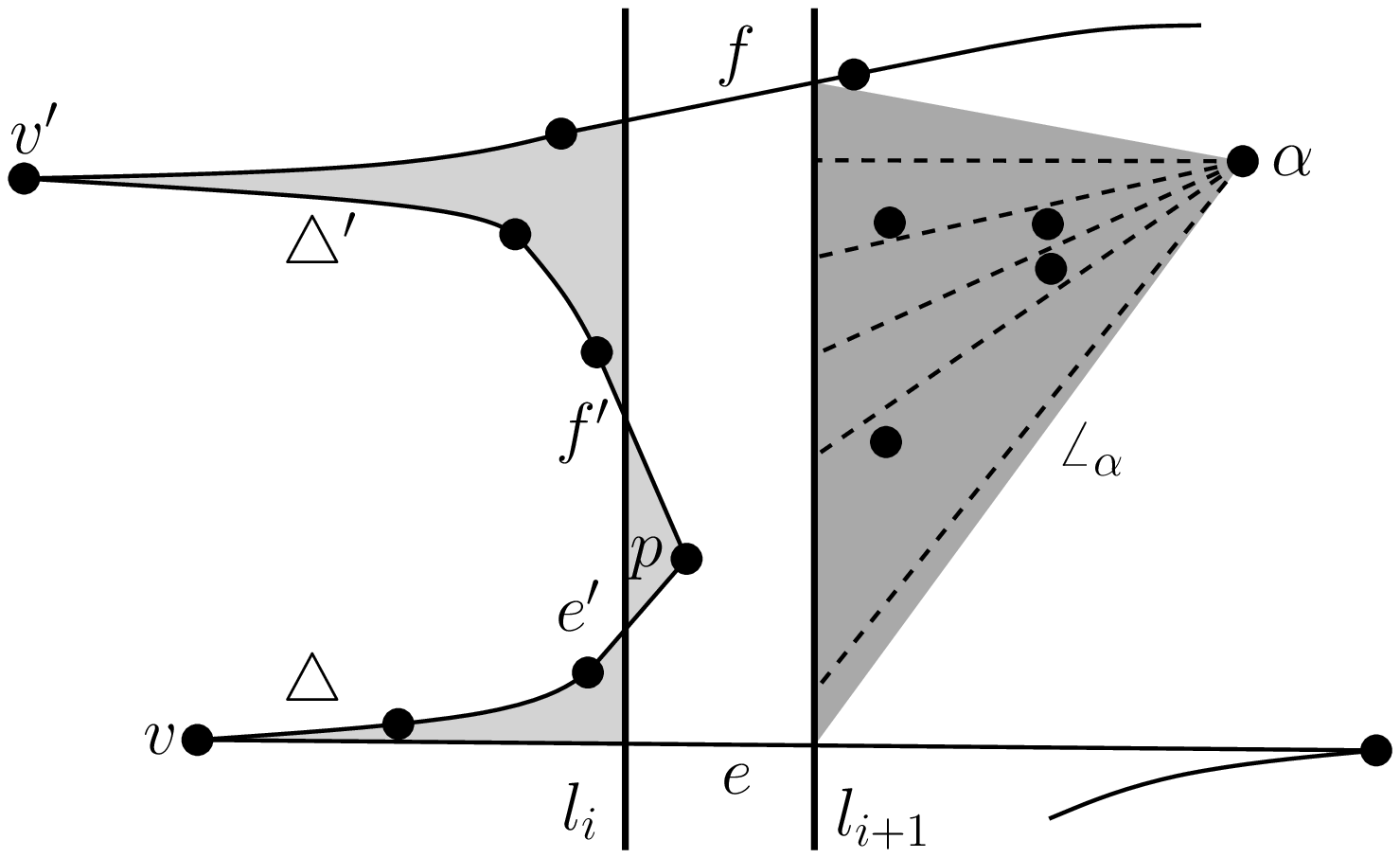}
						\caption{Each of the dashed lines defines an homotopy class.}
						\label{c-tri:figs:15:b}
					\end{center}
				\end{minipage}
			\end{center}
		\end{figure}

	Thus, in order to construct all PT-paths of $\tcomp(l_{i+1}, \setp)$ that can be obtained from $\alpha\in A$, we have to exhaust all its possibilities. This is done as follows: Consider the visibility cone  $\angle_{\alpha}$ to $\mathcal{I}$ with apex at $\alpha$, shown in dark gray in Figure~\ref{c-tri:figs:15:a}. If $\angle_{\alpha}$ is empty, then any visibility ray $\rho$ to $\mathcal{I}$ inside $\angle_{\alpha}$ will do to create $\rho_{\downarrow}$ and $\rho_{\uparrow}$. As a consequence of the emptiness of $\angle_{\alpha}$, point $\alpha$ will spawn only one PT-path of $\tcomp(l_{i+1}, \setp)$. Otherwise, sort the points of $\setp$ inside $\angle_{\alpha}$ angularly around $\alpha$ (clockwise). Now shoot visibility rays $\rho_{0}, \ldots, \rho_{k}$ from $\alpha$ to $\mathcal{I}$ such that between any two consecutive visibility rays there is exactly one point of $\setp$, and use each visibility ray $\rho = \rho_{i}$, $0\leq i\leq k$, to create paths $\rho_{\downarrow}$ and $\rho_{\uparrow}$ as described before. Since $\angle_{\alpha}$ is non-empty, ray $\rho$ defines the homotopy class that paths $\widetilde{\rho_{\downarrow}}, \widetilde{\rho_{\uparrow}}$ belong to. Thus, potentially, every ray $\rho_{i}$, $0\leq i\leq k$, could give a PT-path of $\tcomp(l_{i+1}, \setp)$. Figure~\ref{c-tri:figs:16} shows a configuration where a visibility ray does not produce a PT-path $\pi^{\prime}\in\tcomp(l_{i+1}, \setp)$ where $\alpha$ is a convex vertex of an empty pseudo-triangle of $\pi^{\prime}$ bounded by $l_{i+1}$.
	
		\begin{figure}[!htb]
			\begin{center}
				\includegraphics[height=4cm]{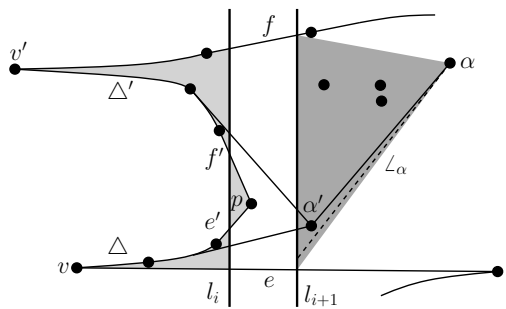}
				\caption{Visibility ray shown in dashed defines the homotopy that the adjacencies connecting $\alpha$ with $\pi$ should follow. In this case the created path is not a PT-path of $\tcomp(l_{i+1}, \setp)$ where $\alpha$ is a convex vertex. It would be nevertheless a PT-path of $\tcomp(l_{i+1}, \setp)$ where $\alpha^{\prime}$ is a convex vertex. This path will be detected when processing $\alpha^{\prime}$.}
				\label{c-tri:figs:16}
			\end{center}
		\end{figure}
		
		So, given $\alpha\in A$, obtaining the points of $\setp$ lying inside $\angle_{\alpha}$, and their sorted order around $\alpha$, can be done in $O(n\log(n))$ time. For each visibility ray $\rho\in\{\rho_{i}\}_{i = 0}^{k}$, we can construct the paths $\rho_{\downarrow}, \rho_{\uparrow}$ in $O(n)$ time, and the shortest homotopic paths $\widetilde{\rho_{\downarrow}}, \widetilde{\rho_{\uparrow}}$ can be computed in $O\left(n^{2}\right)$, see~\cite{DBLP:journals/jal/Bespamyatnikh03} and references therein. Thus, we spend $O\left(n^{3}\right)$ time to
exhaust all possibilities for $\alpha$, and it can spawn $O(n)$ different PT-paths of $\tcomp(l_{i+1}, \setp)$. Doing this for every element of $A$ takes $O\left(n^{4}\right)$ time in total, where also the total number of PT-paths produced is $O\left(n^{2}\right)$. Clearly, by construction, the union of each PT-path $\pi^{\prime}\in\tcomp(l_{i+1}, \setp)$ constructed this way from a PT-path $\pi\in\tcomp(l_{i}, \setp)$ is non-crossing and pointed.

		\begin{figure}[!htb]
			\begin{center}
				\begin{minipage}[b][7cm][t]{7cm}
					\begin{center}
						\includegraphics[height=4cm]{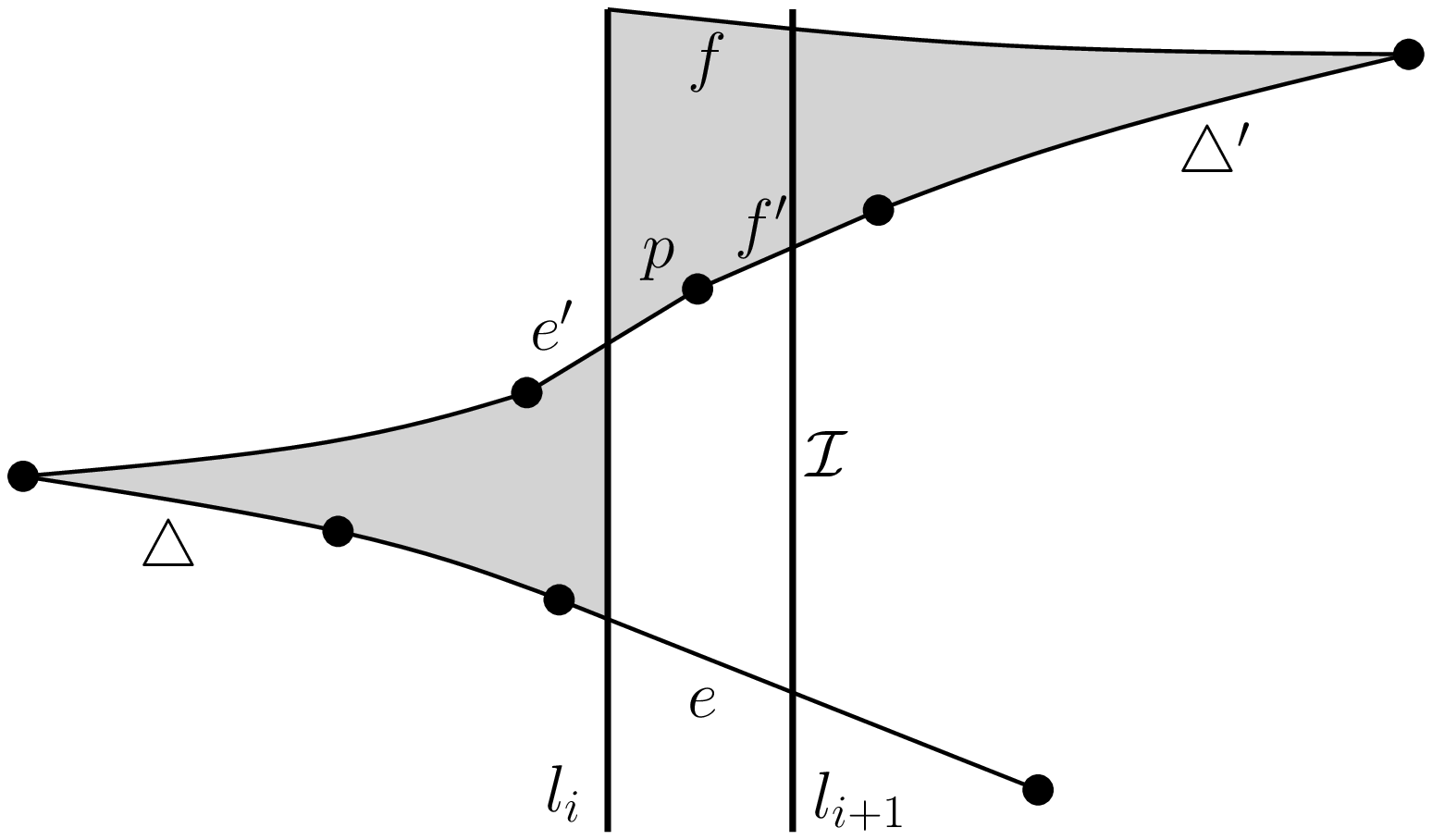}
						\caption{The symmetric configuration in which $\triangle$ and $\triangle^{\prime}$ lie on opposite sides is also possible.}
						\label{c-tri:figs:5}
					\end{center}
				\end{minipage}
			\quad
				\begin{minipage}[b][7cm][t]{7cm}
					\begin{center}
						\includegraphics[height=4cm]{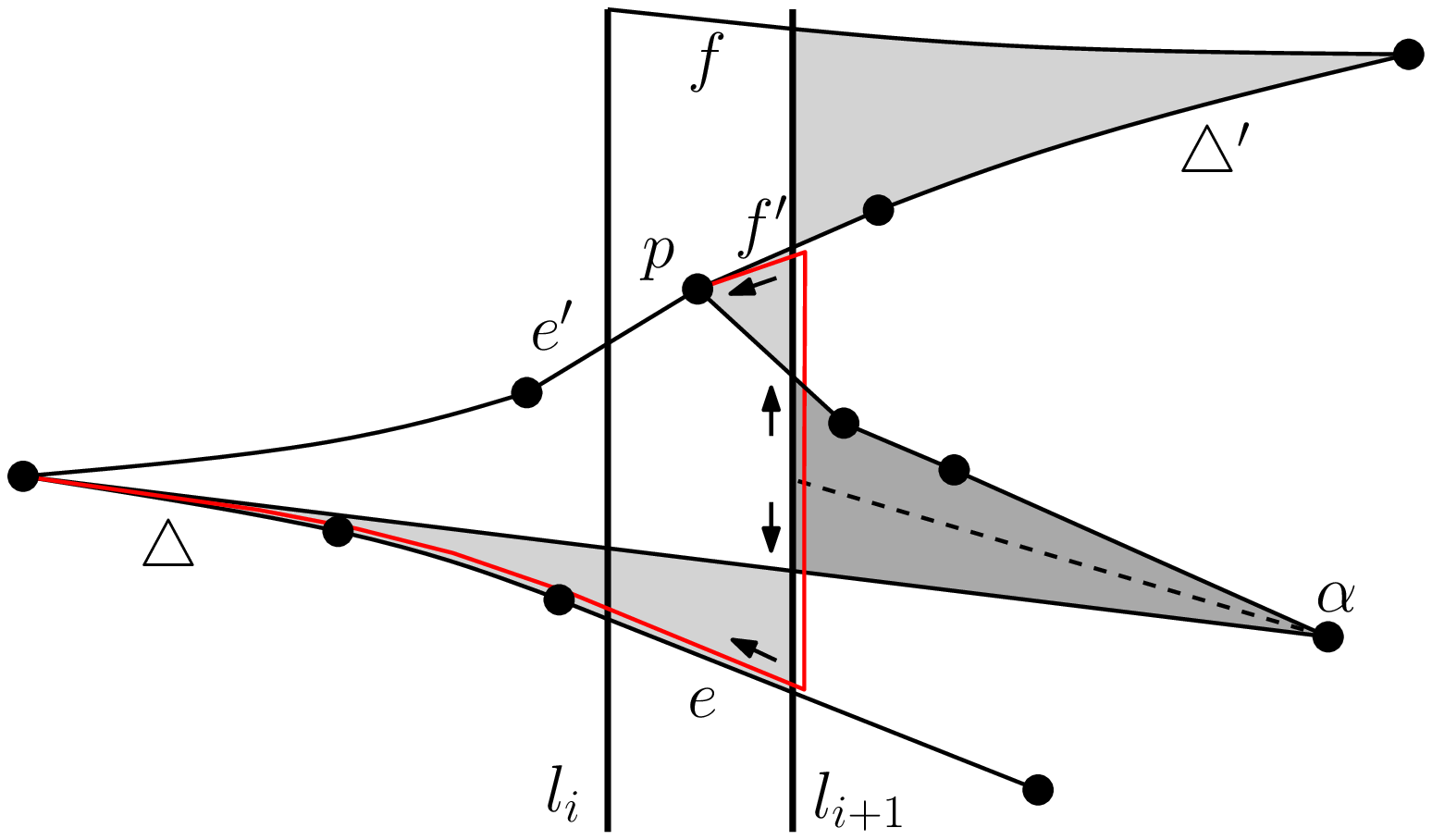}
						\caption{The red lines connect $\alpha$ to $p$ and to the leftmost convex vertex of $\triangle$ via the visibility ray shown dashed. These two paths define the homotopy the local changes must follow.}
						\label{c-tri:figs:6}
					\end{center}
				\end{minipage}
			\end{center}
		\end{figure}
		
		If $p$ is \emph{not} the convex vertex of an empty pseudo-triangle of $\pi$ bounded by $l_{i}$, then the situation is essentially like displayed in Figure~\ref{c-tri:figs:5}. A similar construction can be done that looks like mirror-reflected. Using the same notation as before, the empty pseudo-triangles $\triangle, \triangle^{\prime}$ lie on different sides of $l_{i}$ and $l_{i+1}$. Also, only the edges $e,f$ are good w.r.t.~$l_{i}$ and $l_{i+1}$, edge $e^{\prime}$ is good w.r.t.~$l_{i}$ only, and edge $f^{\prime}$ is good w.r.t.~$l_{i+1}$ only. In the ``mirror-reflected'' construction, edge $f^{\prime}$ is the one that is good w.r.t.~$l_{i}$, and edge $e^{\prime}$ is the one that is good w.r.t.~$l_{i+1}$.
		
		Observe that we cannot extend $\triangle$ to an empty pseudo-triangle bounded by $l_{i+1}$ since point $p$ would be a convex vertex of such extension, and thus that extension would be a pseudo-quadrilateral, see Figure~\ref{c-tri:figs:5}. No such a problem occurs with $\triangle^{\prime}$.
		
		The way we deal with this situation is very similar to the previous case. Let $\mathcal{I}$ and $A$ be as before. For every $\alpha\in A$ define again the visibility cone $\angle_{\alpha}$, and construct the set of rays $\{\rho_{i}\}_{i = 0}^{k}$ as well. For $\rho\in\{\rho_{i}\}_{i = 0}^{k}$, define the path $\rho_{\downarrow}$ just as before. This time, however, define $\rho_{\uparrow}$ as the path that connects the intersection point of $\rho$ and $l_{i+1}$ with $p$ by following $l_{i+1}$ up to edge $f^{\prime}$, and then $f^{\prime}$ to $p$. We now compute the two shortest paths $\widetilde{\rho_{\downarrow}}, \widetilde{\rho_{\uparrow}}$ homotopic to $\rho_{\downarrow}\cup\rho, \rho_{\uparrow}\cup\rho$ respectively. So again we exhaust all possibilities of every point in $A$. The time remains $O\left(n^{4}\right)$ in total, and again the number of PT-paths of $\tcomp(l_{i+1}, \setp)$ produced is $O\left(n^{2}\right)$. If $\alpha$ is the right endpoint of $e^{\prime}$ or of $f^{\prime}$, then one of the shortest homotopic paths overlaps with the adjacencies of $\pi$, and thus it must be ignored in the resulting PT-path of $\tcomp(l_{i+1}, \setp)$. The reader can use Figure~\ref{c-tri:figs:6} by imaging pulling $\alpha$ to the right endpoint of $e^{\prime}$. Another example of such a degeneracy will be seen later on. 
		
		Observe again that pointedness and planarity is kept.
				
		\item If $p$ lies on $\Conv(\setp)$ then one possible configuration is as the one shown in Figure~\ref{c-tri:figs:7:a}, in which $p$ is the last, or first, vertex of $\pi\in\tcomp(l_{i}, \setp)$. Another possibility arises when $p$ is the second, or second-to-the-last, vertex of $\pi$. Which shortest homotopic paths should be computed should be clear from the figure by now.
		
		\begin{figure}[!htb]
			\begin{center}
				\begin{minipage}[b][6.5cm][t]{7cm}
					\begin{center}
						\includegraphics[height=4cm]{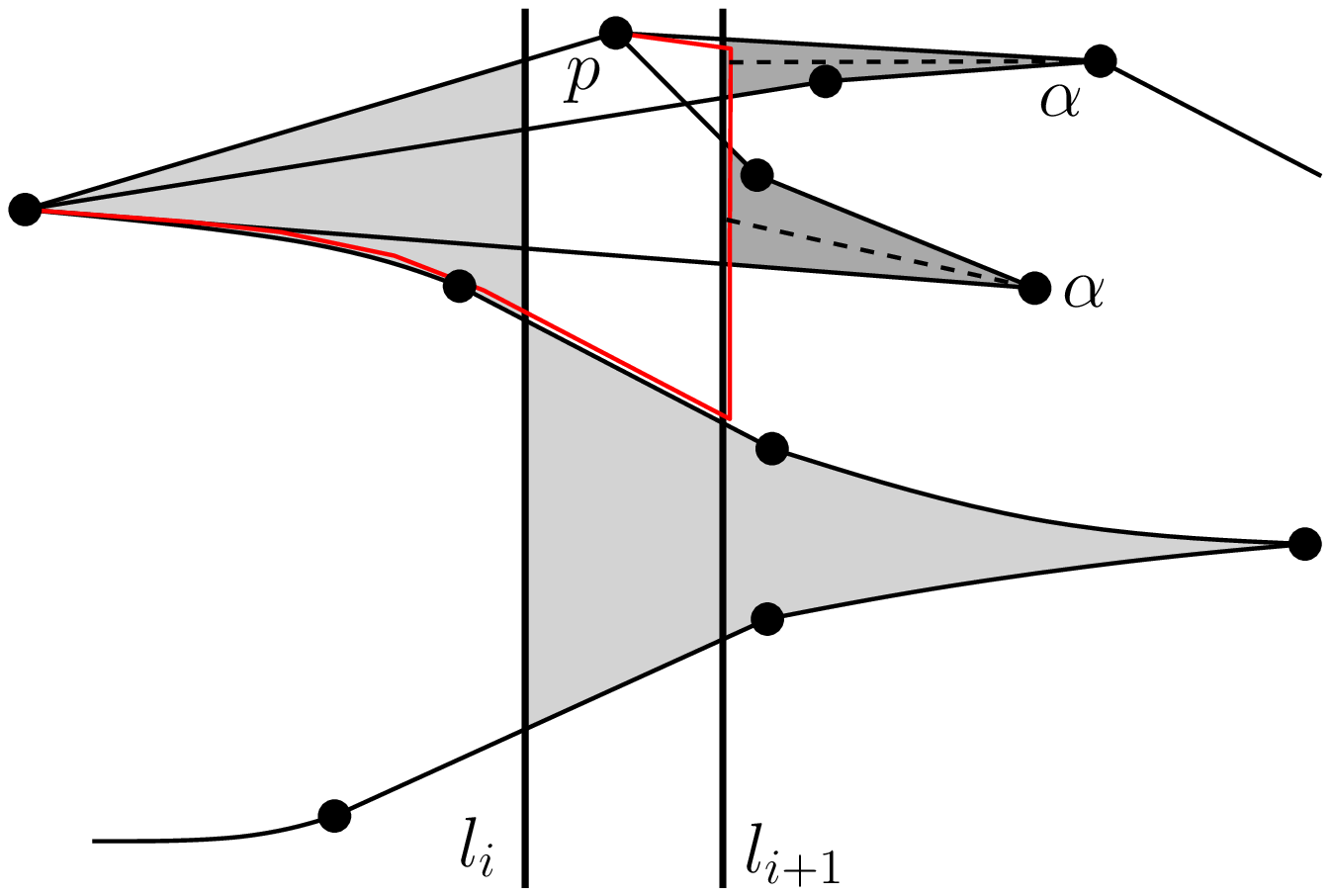}
						\caption{In this case $p$ lies on $\Conv(\setp)$ and its degree in $\pi\in\tcomp(l_{i}, \setp)$ is exactly one. Two possibilities using two different $\alpha$'s are shown.}
						\label{c-tri:figs:7:a}
					\end{center}
				\end{minipage}
			\quad
				\begin{minipage}[b][6.5cm][t]{7cm}
					\begin{center}
						\includegraphics[height=4cm]{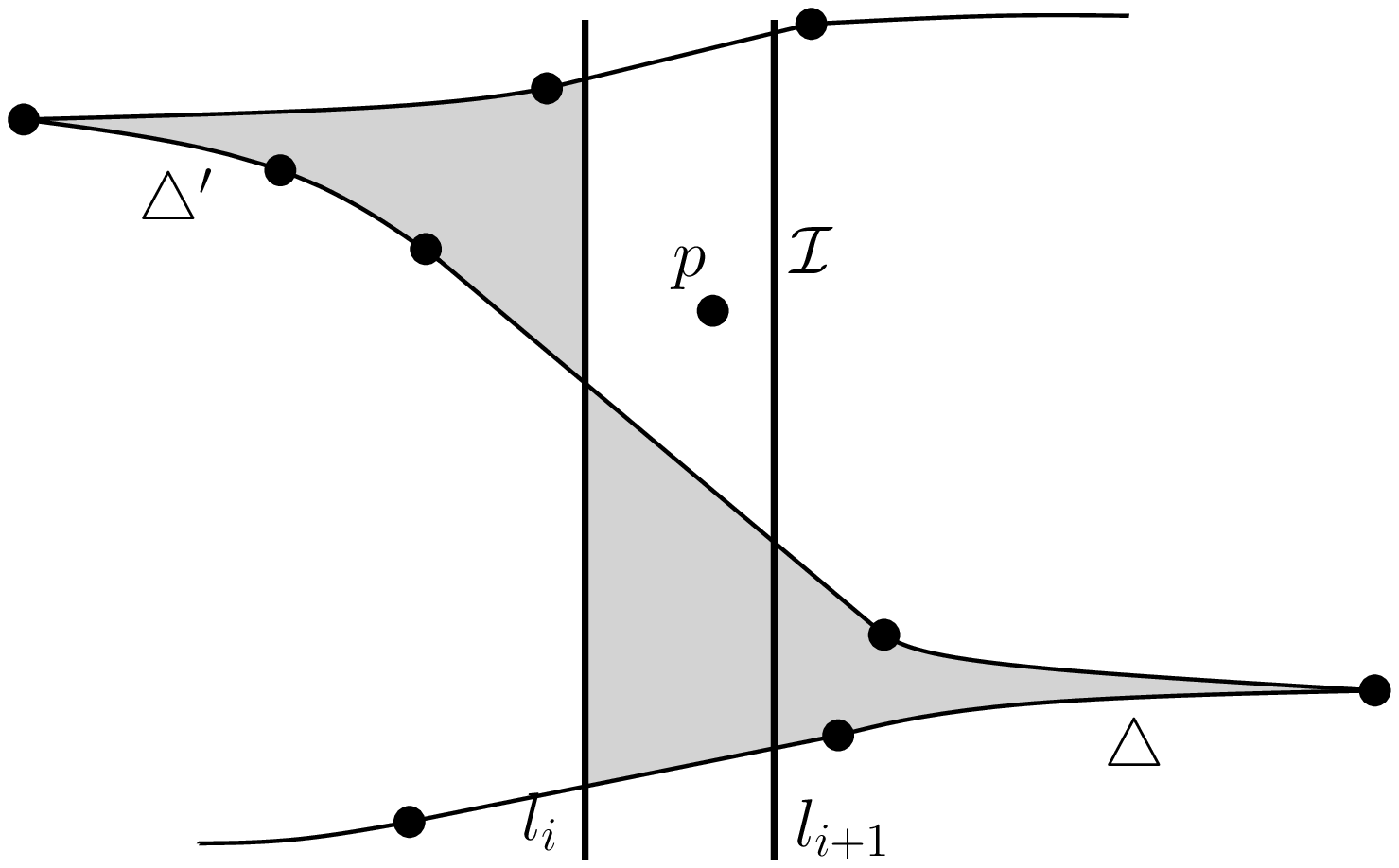}
						\caption{Although $p$ is not a vertex of $\pi\in\tcomp(l_{i}, \setp)$, it must be part of some $\pi^{\prime}\in\tcomp(l_{i+1}, \setp)$ since the empty pseudo-triangle $\triangle^{\prime}$ of $\pi$ cannot be extended further.}
						\label{c-tri:figs:7:b}
					\end{center}
				\end{minipage}
			\end{center}
		\end{figure}
		
	\end{itemize}
	
	\item If $\pi$ \emph{does not} have $p$ as a vertex, then $p$ must necessarily lie inside $\Conv(\setp)$. The situation is in general as displayed in Figure~\ref{c-tri:figs:7:b}. In this case there are two kinds of local changes that can be made; one kind is produced by a single point $\alpha\in\setp$, and the other kind is produced by pairs of points $\alpha, \beta\in\setp$, see Figures~\ref{c-tri:figs:17:a} and~\ref{c-tri:figs:17:b} for a reference.
	
		\begin{figure}[!htb]
			\begin{center}
				\begin{minipage}[b][5cm][t]{7cm}
					\begin{center}
						\includegraphics[height=4cm]{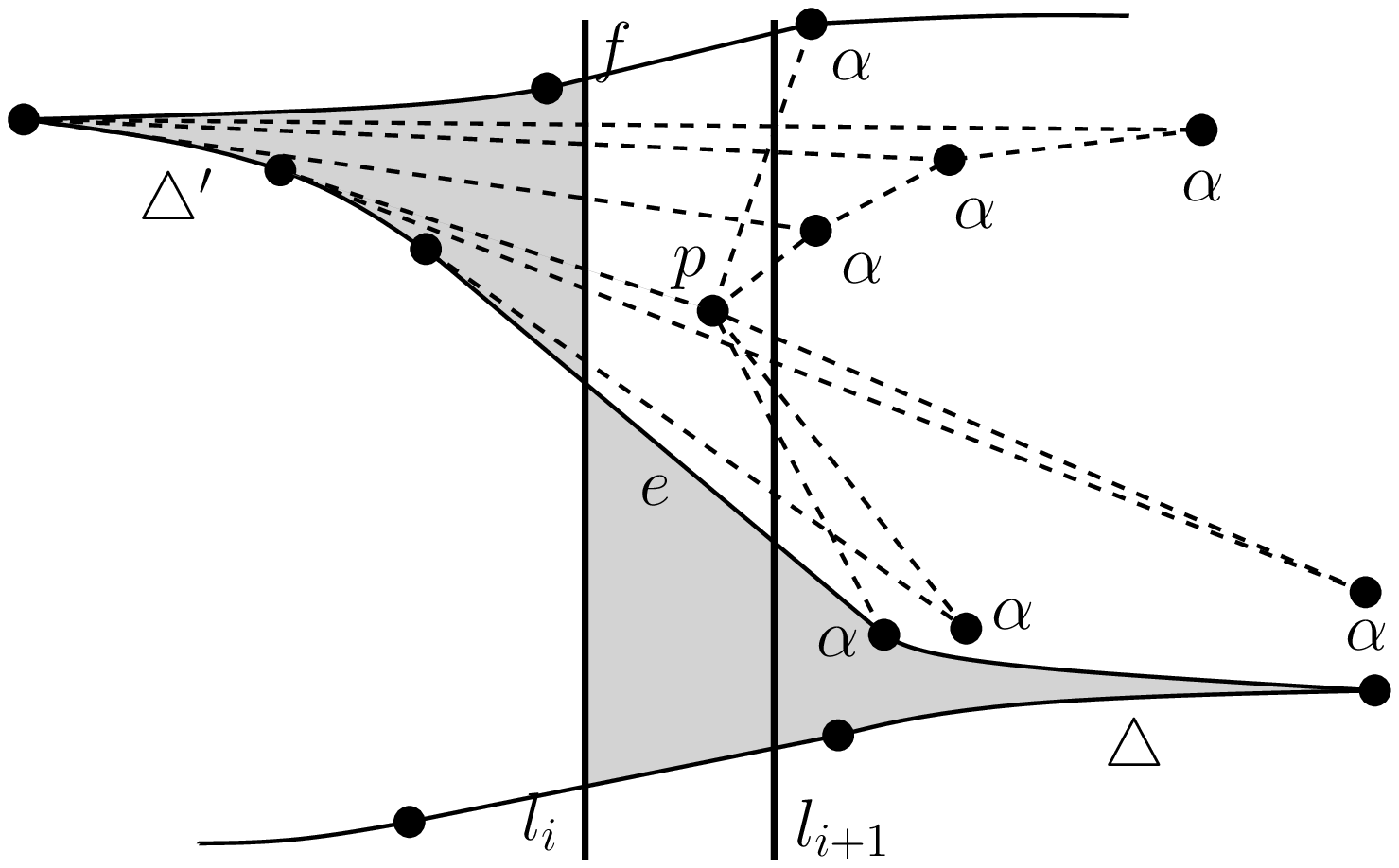}
						\caption{Changes are produced only by one point $\alpha$.}
						\label{c-tri:figs:17:a}
					\end{center}
				\end{minipage}
			\quad
				\begin{minipage}[b][5cm][t]{7cm}
					\begin{center}
						\includegraphics[height=4cm]{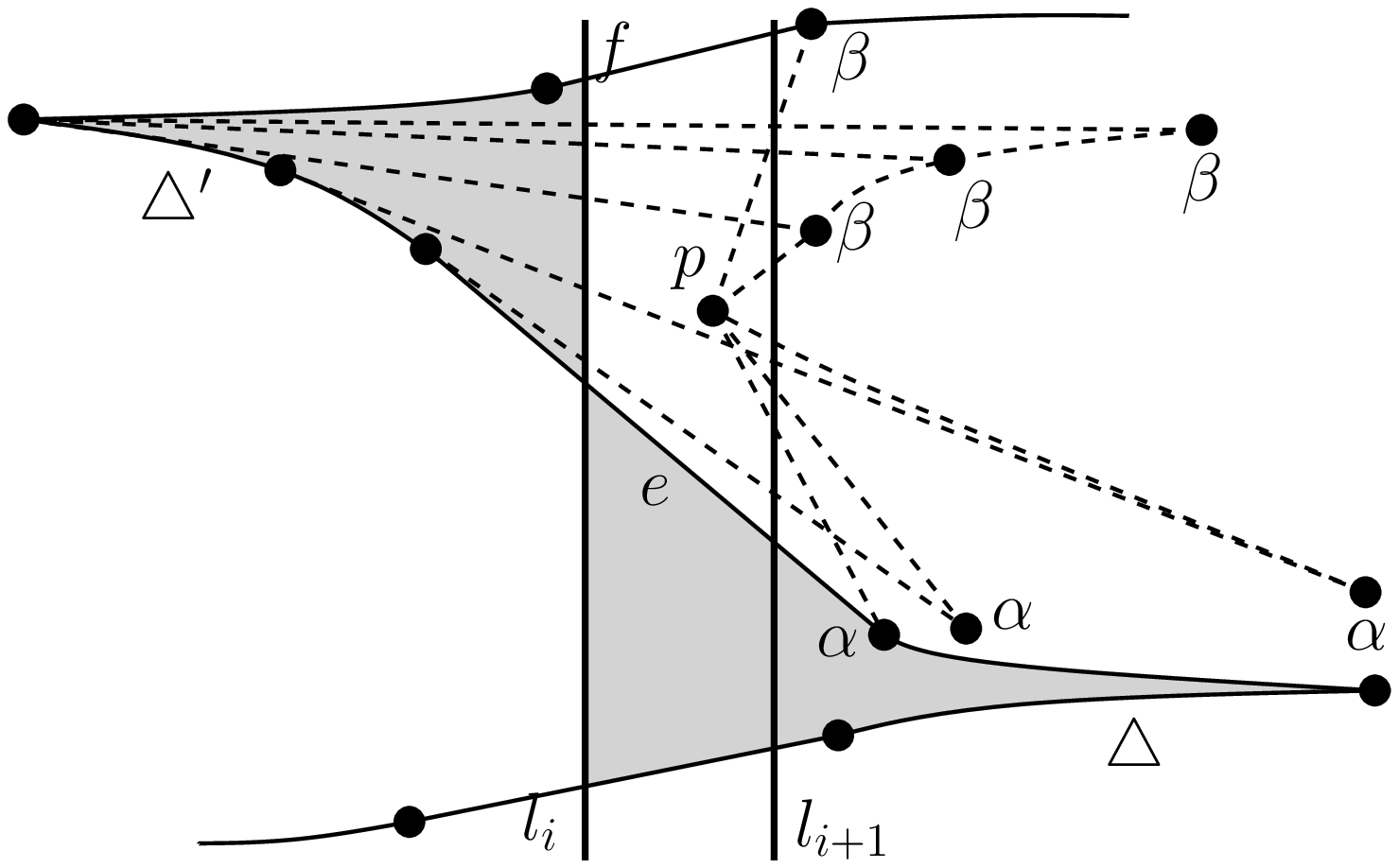}
						\caption{Changes are now produced by pairs of points $\alpha, \beta$.}
						\label{c-tri:figs:17:b}
					\end{center}
				\end{minipage}
			\end{center}
		\end{figure}

		Let $\mathcal{I}, A$ be defined as before. Let us see each kind of local changes in turn. For the local changes produced by just one point $\alpha\in A\subset\setp$, the PT-paths of $\tcomp(l_{i+1}, \setp)$ produced look like the ones in Figure~\ref{c-tri:figs:18}.
		
		\begin{figure}[!htb]
			\begin{center}
				\begin{minipage}[b][4.2cm][t]{7cm}
					\begin{center}
						\includegraphics[height=4cm]{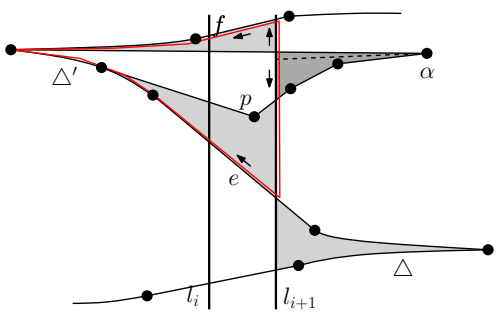}
					\end{center}
				\end{minipage}
			\quad
				\begin{minipage}[b][4.2cm][t]{7cm}
					\begin{center}
						\includegraphics[height=4cm]{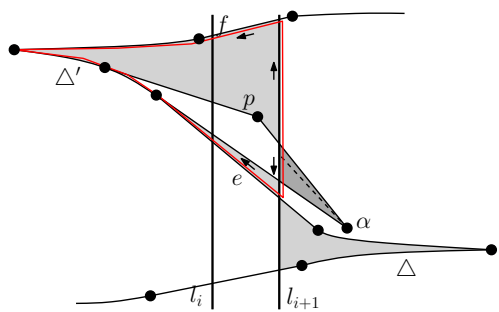}
					\end{center}
				\end{minipage}
			\caption{Two different PT-paths of $\tcomp(l_{i+1}, \setp)$ produced by two different points.}
			\label{c-tri:figs:18}
			\end{center}
		\end{figure}
		
		Using the same ideas as before, of following red paths, the adjacencies of $\alpha$ in a PT-path of $\tcomp(l_{i+1}, \setp)$ are two shortest paths homotopic to the two red paths shown in Figure~\ref{c-tri:figs:18}, one going up and the other going down, and the visibility ray from $\alpha$ to $\mathcal{I}$, shown dashed in Figure~\ref{c-tri:figs:18}. Using the visibility cone $\angle_{\alpha}$ we can again exhaust all possibilities for $\alpha$ in $O\left(n^{3}\right)$ time, and thus we exhaust all of $A$ in $O\left(n^{4}\right)$ time, producing $O\left(n^{2}\right)$ PT-paths of $\tcomp(l_{i+1}, \setp)$ in total. As a remark, observe that if $\alpha$ is the right endpoint of edge $f$ or $e$, then one of the shortest homotopic paths overlaps completely with adjacencies of $\pi\in\tcomp(l_{i}, \setp)$, this path can be ignored, and then the produced PT-path of $\tcomp(l_{i+1}, \setp)$ would look like the one in Figure~\ref{c-tri:figs:19:a}, where the path of $\pi$ connecting $\alpha$ with the leftmost convex vertex of $\triangle^{\prime}$ is the one ignored.
		
		\begin{figure}[!htb]
			\begin{center}
				\begin{minipage}[b][6.5cm][t]{7cm}
					\begin{center}
						\includegraphics[height=4cm]{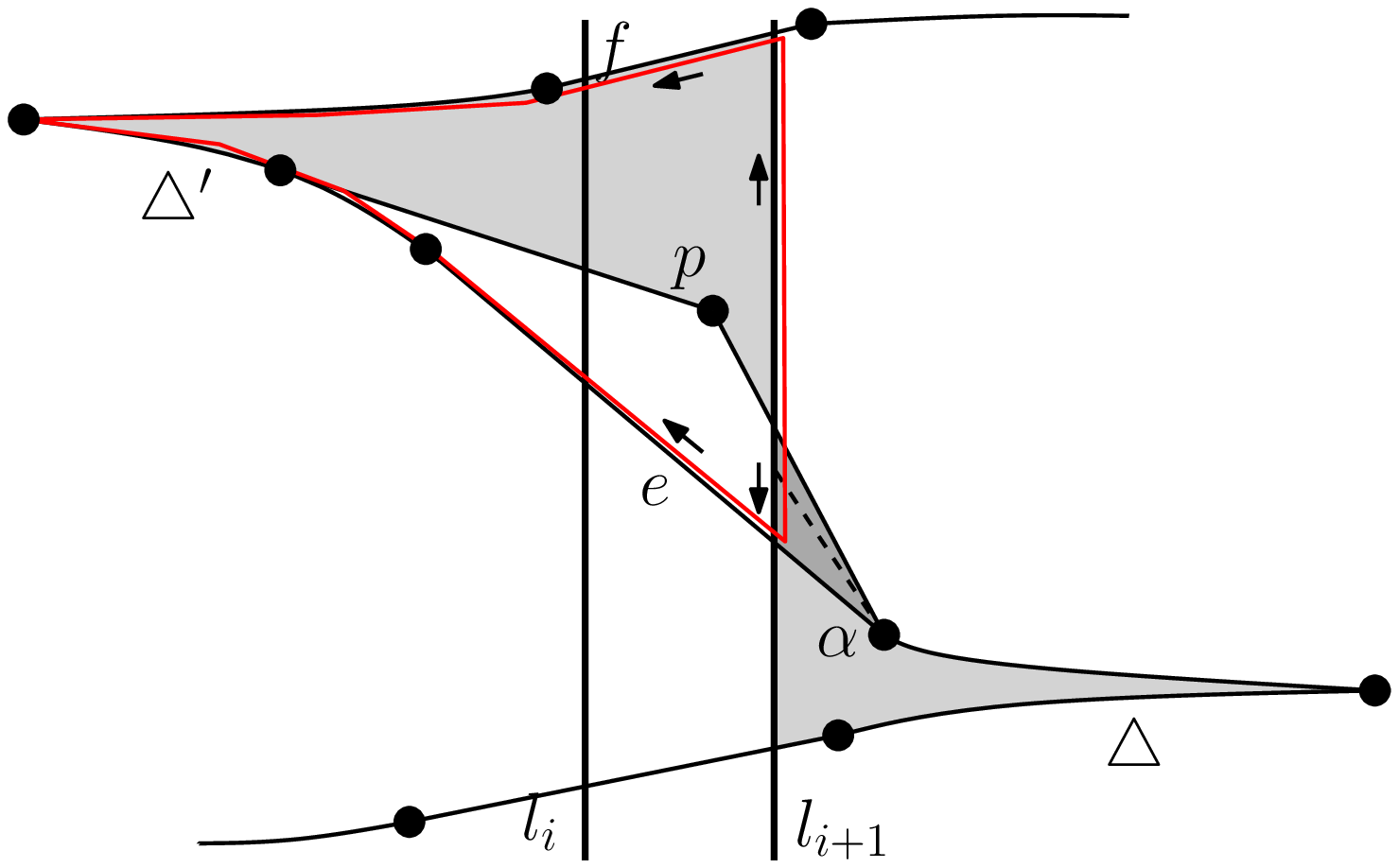}
						\caption{A particular case occurs if $\alpha$ coincides with an endpoint of $e$ or of $f$.}
						\label{c-tri:figs:19:a}
					\end{center}
				\end{minipage}
			\quad
				\begin{minipage}[b][6.5cm][t]{7cm}
					\begin{center}
						\includegraphics[height=4cm]{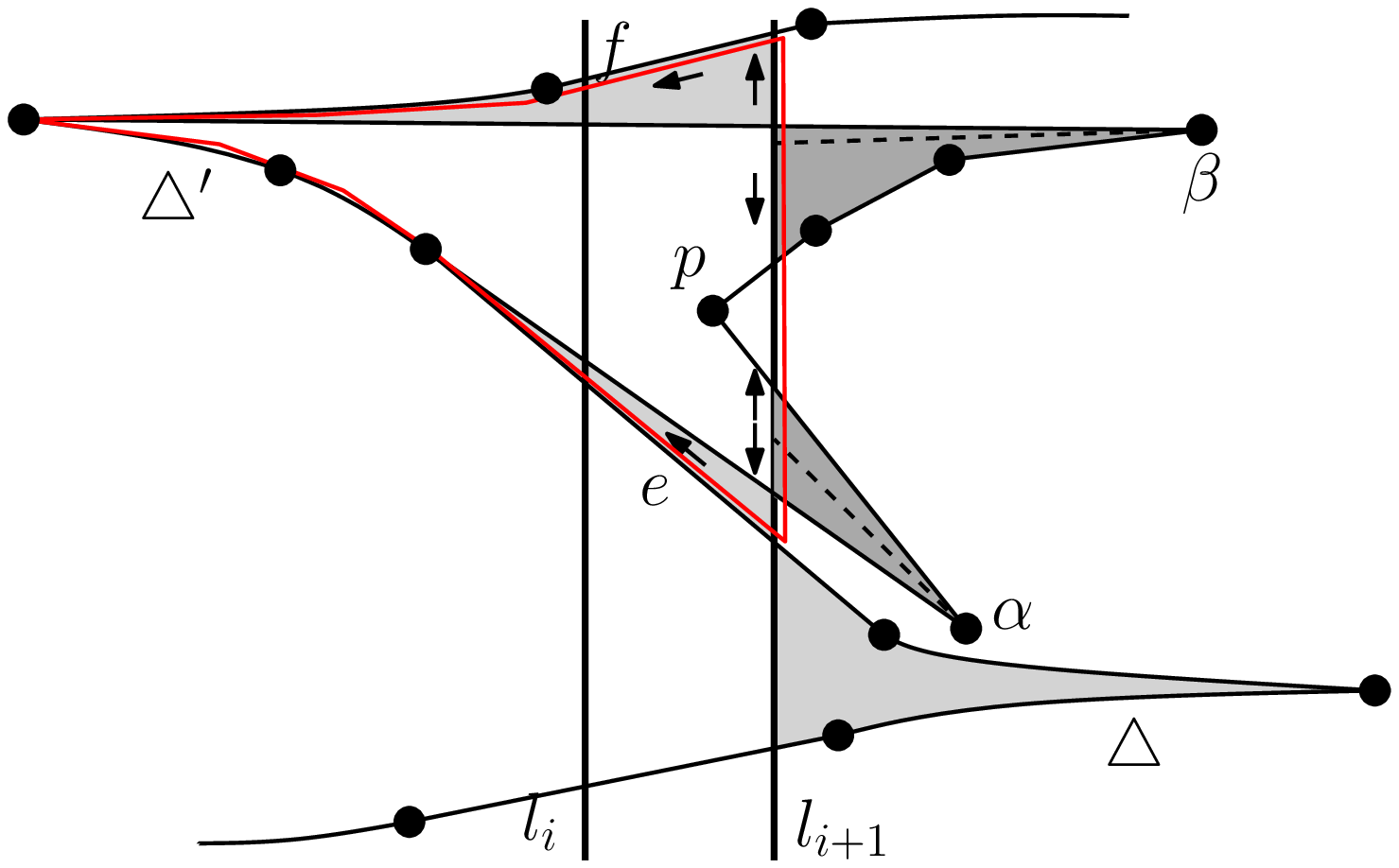}
						\caption{Combining the PT-paths shown in Figure~\ref{c-tri:figs:18} we obtained yet another PT-path of $\tcomp(l_{i+1}, \setp)$, we just had to remove the adjacencies of $p$ that make it non-pointed.}
						\label{c-tri:figs:19:b}
					\end{center}
				\end{minipage}
			\end{center}
		\end{figure}
		
		As for the local changes produced by pairs of points $\alpha,\beta\in A\subset\setp$, the PT-paths of $\tcomp(l_{i+1}, \setp)$ produced look like the one shown in Figure~\ref{c-tri:figs:19:b}. If we have constructed the PT-paths produced by a single $\alpha\in A$, then we can construct the paths produced by pairs $\alpha,\beta\in A$ by combining the local changes applied to $\alpha$, with all the local changes applied to $\beta$. For example, the PT-path $\pi^{\prime}\in\tcomp(l_{i+1}, \setp)$ shown in Figure~\ref{c-tri:figs:19:b} is obtained from the PT-paths of Figure~\ref{c-tri:figs:18}, by removing the adjacencies at $p$ that do not make it pointed. So, when combining changes we have, of course, to be careful about pointedness and planarity of the construction, which takes not much more effort to verify.
		
		Since the total number of different PT-paths produced by $\alpha, \beta$ is $O\left(n \right)$, by combining them we will obtained no more than $O\left(n^{2}\right)$ PT-paths. Thus, by going through every pair $\alpha,\beta\in A$, the total number of PT-paths of $\tcomp(l_{i+1}, \setp)$ produced is $O\left(n^{4}\right)$, and all this can be achieved in $O\left(n^{6}\right)$ time, since combining a pair can be achieved in $O\left(n^{2}\right)$ time. 
		
		This concludes the explanation of the local changes that need to be made to PT-paths as we sweep.

\end{enumerate}

As for T-paths, the local changes of PT-paths can be seen in reverse order, as going from line $l_{i+1}$ to $l_{i}$, so we will use again the notation $\pi\leftrightarrow\pi^{\prime}$ to denote the fact that $\pi^{\prime}\in\tcomp(l_{i+1}, \setp)$ is produced from $\pi\in\tcomp(l_{i}, \setp)$ in one direction, so $\pi^{\prime}\in\mu(\pi)$, and $\pi$ is produced from $\pi^{\prime}$ in the reverse direction, so $\pi\in\lambda(\pi^{\prime})$.

We can now prove the following result which is the equivalent to Lemma~\ref{lemmas:local} on page~\pageref{lemmas:local} for T-paths:

\begin{lemma}\label{c-tri:lemmas:pt-paths:1}
	Given $\tcomp(l_{i}, \setp)$, every PT-path of $\tcomp(l_{i+1}, \setp)$ is produced by the local changes just explained. Moreover, for each $\pi\in\tcomp(l_{i}, \setp)$, the cardinality of $\mu(\pi)$ is $O\left(n^{4}\right)$, and we can correctly compute $\lambda(\pi^{\prime})$, for each $\pi^{\prime}\in\tcomp(l_{i+1}, \setp)$, in time $O\left(n^{6}\cdot t_{i}\right)$, where $t_{i} = |\tcomp(l_{i}, \setp)|$.
\end{lemma}
\begin{proof}
	For the first part of the lemma an argument as the one we used for the first part of Lemma~\ref{lemmas:local} can be used. We can check that given any PT-path $\pi^{\prime}\in\tcomp(l_{i+1}, \setp)$ we can always obtain a PT-path $\pi\in\tcomp(l_{i}, \setp)$ by locally changing $\pi^{\prime}$, and thus every PT-path of $\tcomp(l_{i+1}, \setp)$ is produced by the relation $\pi\leftrightarrow\pi^{\prime}$. The second part, the correct computation of $\lambda(\pi^{\prime})$ for every $\pi^{\prime}\in\tcomp(l_{i+1}, \setp)$, also follows by a similar argument as the one we did in Lemma~\ref{lemmas:local} in the corresponding part, that is, $\pi\not\leftrightarrow\pi^{\prime}$ implies that $\pi$ and $\pi^{\prime}$ properly cross. 
	
	Finally, the size of $\mu(\pi)$ and the time it takes to compute $\lambda(\pi^{\prime})$, for \emph{every} $\pi^{\prime}\in\tcomp(l_{i+1}, \setp)$, follows from the explanations done while explaining the local changes of PT-paths. Hence the lemma follows.
\end{proof}

This concludes the proof of Theorem~\ref{c-tri:theorems:pt-paths}.

\section{Discussion and conclusions}\label{c-tri:sections:conclusionsT-ST}

The problem of ``algorithmically'' counting crossing-free structures defined on given sets of points is directly related to the problem of generating random crossing-free structures. For example, we might be interested in producing a triangulation of a given set of points $\setp$ uniformly at random, that is, \emph{every} triangulation of $\setp$ must appear with probability $\frac{1}{|\F_{T}(\setp)|}$. This allows us to study structural properties of an ``average'' triangulation of $\setp$, for example, to check how many of its vertices have a given degree, or to verify what fraction of its vertices has a degree of certain parity. This could allow us to make conjectures on triangulations and to try to prove them using induction, for which the base cases can be checked by computer. 

Methods to produce random triangulations are known, for example, in~\cite{DBLP:conf/compgeom/Aichholzer99} a method is explained that produces random triangulations using the divide-and-conquer algorithm therein presented. For the sweep line algorithms that we just presented another method can be used (due to a different paradigm): Assume we want to generate a random triangulation, generating random pointed pseudo-triangulations is the same. Remember that we sweep from left to right, so we store for every event point $l_{i}$, $1\leq i\leq n-1$, and for every T-path $\pi$ found w.r.t.~$l_{i}$, the cardinality of $\T(\pi)$, which is the number of structures to the left of $l_{i}$ that are compatible with $\pi$. We construct a random triangulation by sweeping in reverse order once the algorithm has finished the counting. Since there is only one path w.r.t.~$l_{n-1}$ we choose it. Going from $l_{i+1}$ to $l_{i}$, $1\leq i < n-1$, and having fixed a path $\pi_{i+1}$ w.r.t.~$l_{i+1}$, we choose a path $\pi_{i}$ w.r.t.~$l_{i}$ with probability $\frac{|\T(\pi_{i})|}{|\T(\pi_{i+1})|}$. By the time we arrive at $l_{1}$ we have generated a triangulation with probability:

\begin{align*}
1\cdot\frac{|\T(\pi_{n-2})|}{|\T(\pi_{n-1})|}\cdot\frac{|\T(\pi_{n-3})|}{|\T(\pi_{n-2})|}\cdots\frac{|\T(\pi_{1})|}{|\T(\pi_{2})|} = \frac{|\T(\pi_{1})|}{|\T(\pi_{n-1})|} = \frac{1}{|\F_{T}(\setp)|}
\end{align*}

\noindent since there is only one T-path w.r.t.~$l_{1}$. The downside of this method is that we need to compute the number of triangulations of $\setp$ beforehand. 

There is nevertheless a different method that seems to be quite good in practice, this method works by randomly flipping edges of a triangulation (with a pseudo-triangulation it would be the same). It is known that this method leads to a random triangulation in polynomial time for sets of points in convex position, see~\cite{tri-path, catalan-structs}. Note, however, that since the number of triangulations of a convex polygon is a Catalan number, a triangulation generated uniformly at random can be obtained in optimal linear time, see~\cite{DBLP:journals/tomacs/EpsteinS94} and references therein. For general sets of points nothing is known about the convergence of the random flipping procedure. This is a very interesting and challenging open problem.

\subsection{Conclusions}

In this paper we have presented algorithms to compute the number of triangulations and pseudo-triangulations of a given set of points $\setp$. Both algorithms are rather simple and they are based on T-paths, PT-paths and the sweep line paradigm. We also provided the first non-trivial upper bound for the number of T-paths of $\setp$ w.r.t.~to a given separating line. Unfortunately, this number turned out to be rather large, $O\left(9^{n}\right)$. We believe that the real upper bound for this number is closer to $4^{n}$, which remains being very large nevertheless. However, we are not aware of any configuration of points, large enough, having as many T-paths as triangulations. This has previously been supported by experiments and proven for many known configurations of points. 

It seems that our T-path algorithm really is counting triangulations in time sub-linear in the number of triangulations, so we believe that this algorithm is still very interesting from the theoretical point of view. We suspect the same about our PT-path algorithm for counting pseudo-triangulation. An easy argument can be done to show that these algorithms are, in any case, no worse than enumeration algorithms. Although this sounds pessimistic, there are algorithms for which such an argument cannot be done.

The holy grail of counting triangulations is to prove polynomial time or \#P-hardness. So far we have failed to prove any of them. Thus, the most interesting open questions at this moment are (in ascending order of importance): (\oldstylenums{1}) For $n$ large enough, is it true that there are always asymptotically more triangulations (pseudo-triangulations) than T-paths (PT-paths) w.r.t.~a given separating line? (\oldstylenums{2}) Is it possible to count triangulations (pseudo-triangulations) in sub-exponential time? Or even count approximately in polynomial time? (\oldstylenums{3}) Is the problem of counting triangulations (pseudo-triangulations) in P, or is it \#P-complete? Each one of these questions looks very challenging.

\section{Acknowledgement}

We thank Raimund Seidel for valuable feedback and interesting discussions.

\small
\bibliographystyle{ieeetr}
\bibliography{bibliography-sweep-line-algorithm}
\end{document}